\documentclass[twoside,11pt]{entics}
\pdfoutput=1

\usepackage{lscape}
\usepackage{amsfonts}
\usepackage{amsopn}
\usepackage{amsmath}
\usepackage{amssymb}
\usepackage{subcaption}
\usepackage{scalerel}
\usepackage{mathrsfs}
\usepackage{hyphenat}
\usepackage{cancel}
\usepackage{mathpartir}
\usepackage{xfrac}
\usepackage{cite}
\newcommand{\mathsc}[1]{{\normalfont\textsc{#1}}}
\newcommand{\biobj}[2]{{\protect\substack{ #1 \\ #2 }}}

\usepackage{braket}
\renewcommand{\ket}[1]{\ensuremath{|{\kern.1em}#1{\kern.1em}\rangle}}
\newcommand{\bigket}[1]{\ensuremath{\big|{\kern.1em}#1{\kern.1em}\big\rangle}}
\newcommand{\ketstrut}{\vrule height 8.5pt depth 4.5pt width 0pt}
\newcommand{\Bigket}[1]{\ensuremath{\left|\ketstrut{\kern.1em}\right.{\kern-.2em}#1{\kern-.2em}\left.\ketstrut{\kern0em}\right>}}

\usepackage{xypic}
\usepackage[all]{xy}
\xyoption{all}

\usepackage{xcolor}

\definecolor{nordred}{HTML}{bf616a}
\definecolor{bordeaux}{HTML}{821529}
\definecolor{bluelink}{HTML}{003399}
\definecolor{nordred}{HTML}{bf616a}
\definecolor{nordblue}{HTML}{81a1c1}
\definecolor{norddarkblue}{HTML}{5e81ac}
\definecolor{nordgreen}{HTML}{a3be8c}
\definecolor{nordnight}{HTML}{4c566a}

\AtEndPreamble{
  \RequirePackage{hyperref}
  \RequirePackage{cleveref}
  \hypersetup{
    breaklinks = true,
    linktocpage,
    colorlinks = true,
    linkcolor = nordnight,
    citecolor = nordgreen,
    urlcolor = nordblue
  }
} %
\makeatletter
\newcommand{\nicelinktarget}[1]{\Hy@raisedlink{\hypertarget{#1}{}}}
\makeatother

\newcommand\linkdef[2]{\nicelinktarget{#1}{\color{black}{#2}}}
\newcommand\defining[1]{\nicelinktarget{#1}{}}
\usepackage[leftmargin=1em,rightmargin=1ex,vskip=1ex]{quoting}
\usepackage{listings}
\lstnewenvironment{code}
    {\lstset{}%
      \csname lst@SetFirstLabel\endcsname}
    {\csname lst@SaveFirstLabel\endcsname}
\usepackage[utf8]{inputenc}
\usepackage{cmll}

\newcommand\scalemath[3]{\scalebox{#1}[#2]{\mbox{\ensuremath{\displaystyle #3}}}}

\newcommand{\leftarrowtailnotip}{\ensuremath{\tikz\draw[line width=0.5pt,-<] (0,0) -- (10pt,0);}}

\newcommand{\unicodeStar}{\ensuremath{\star}}
\DeclareUnicodeCharacter{2605}{\unicodeStar}

\DeclareUnicodeCharacter{21D2}{\ensuremath{\Rightarrow}}
\DeclareUnicodeCharacter{2218}{\ensuremath{\circ}}
\DeclareUnicodeCharacter{2022}{\ensuremath{\bullet}}
\DeclareUnicodeCharacter{2219}{\ensuremath{\bullet}}
\DeclareUnicodeCharacter{2026}{\ensuremath{\dots}}
\DeclareUnicodeCharacter{2208}{\ensuremath{\in}}
\DeclareUnicodeCharacter{2192}{\ensuremath{\to}}
\DeclareUnicodeCharacter{2190}{\ensuremath{\gets}}
\DeclareUnicodeCharacter{2919}{\ensuremath{\leftarrowtailnotip}}

\DeclareUnicodeCharacter{00D7}{\ensuremath{\times}}
\DeclareUnicodeCharacter{00B7}{\ensuremath{\cdot}}
\DeclareUnicodeCharacter{222B}{\ensuremath{\int}}
\DeclareUnicodeCharacter{22A4}{\ensuremath{\top}}
\DeclareUnicodeCharacter{22A5}{\ensuremath{\bot}}
\DeclareUnicodeCharacter{2264}{\ensuremath{\leq}}
\DeclareUnicodeCharacter{039B}{\ensuremath{\Lambda}}

\newcommand{\unicodecolon}{\ensuremath{\colon}}
\DeclareUnicodeCharacter{FE55}{\unicodecolon}
\newcommand{\unicodeleftpar}{\ensuremath{\left(}}
\DeclareUnicodeCharacter{27EE}{\unicodeleftpar}
\newcommand{\unicoderightpar}{\ensuremath{\right)}}
\DeclareUnicodeCharacter{27EF}{\unicoderightpar}
\DeclareUnicodeCharacter{2260}{\neq}
\DeclareUnicodeCharacter{22A9}{\Vdash}
\DeclareUnicodeCharacter{2237}{\proportion}
\DeclareUnicodeCharacter{2124}{\mathbb{Z}}
\DeclareUnicodeCharacter{27E8}{\langle}
\DeclareUnicodeCharacter{27E9}{\rangle}
\DeclareUnicodeCharacter{21A6}{\mapsto}
\DeclareUnicodeCharacter{22A2}{\vdash}
\DeclareUnicodeCharacter{2090}{\ensuremath{{}_a}}
\DeclareUnicodeCharacter{A71B}{{}^\uparrow}
\DeclareUnicodeCharacter{A71C}{{}^\downarrow}
\DeclareUnicodeCharacter{27E6}{\llbracket}
\DeclareUnicodeCharacter{27E7}{\rrbracket}

\newcommand{\unicoderightcircle}{\ensuremath{\RIGHTcircle}}
\DeclareUnicodeCharacter{25D1}{\unicoderightcircle}
\newcommand{\unicodeleftcircle}{\ensuremath{\LEFTcircle}}
\DeclareUnicodeCharacter{25D0}{\unicodeleftcircle}
\DeclareUnicodeCharacter{229B}{\circledast}

\newcommand{\unicodebbA}{\ensuremath{\mathbb{A}}}
\DeclareUnicodeCharacter{1D538}{\unicodebbA}
\newcommand{\unicodebbB}{\ensuremath{\mathbb{B}}}
\DeclareUnicodeCharacter{1D539}{\unicodebbB}
\newcommand{\unicodebbC}{\ensuremath{\mathbb{C}}}
\DeclareUnicodeCharacter{2102}{\unicodebbC}
\DeclareUnicodeCharacter{1D53B}{\ensuremath{\mathbb{D}}}
\DeclareUnicodeCharacter{2115}{\ensuremath{\mathbb{N}}}
\DeclareUnicodeCharacter{211D}{\ensuremath{\mathbb{R}}}
\DeclareUnicodeCharacter{1D543}{\ensuremath{\mathbb{L}}}
\newcommand\UnicodeBlackboardP{\ensuremath{\mathbf{P}}} \DeclareUnicodeCharacter{2119}{\UnicodeBlackboardP}
\DeclareUnicodeCharacter{211A}{\ensuremath{\mathbb{Q}}}
\DeclareUnicodeCharacter{1D544}{\ensuremath{\mathbb{M}}}
\DeclareUnicodeCharacter{1D54C}{\ensuremath{\mathbb{U}}}
\DeclareUnicodeCharacter{1D54D}{\ensuremath{\mathbf{V}}}
\DeclareUnicodeCharacter{1D54E}{\ensuremath{\mathbb{W}}}
\DeclareUnicodeCharacter{1D546}{\ensuremath{\mathbb{O}}}
\DeclareUnicodeCharacter{1D540}{\ensuremath{\mathbb{I}}}
\DeclareUnicodeCharacter{1D54A}{\ensuremath{\mathbb{S}}}
\DeclareUnicodeCharacter{1D53C}{\ensuremath{\mathbb{E}}}

\newcommand{\unicodecalS}{\ensuremath{\mathcal{S}}}
\newcommand{\unicodecalT}{\ensuremath{\mathcal{T}}}
\newcommand{\unicodecalC}{\ensuremath{\mathcal{C}}}
\newcommand{\unicodecalD}{\ensuremath{\mathcal{D}}}
\newcommand{\unicodecalX}{\ensuremath{\mathcal{X}}}
\newcommand{\unicodecalN}{\ensuremath{\mathcal{N}}}
\newcommand{\unicodecalE}{\ensuremath{\mathcal{E}}}
\DeclareUnicodeCharacter{1D4D4}{\unicodecalE}
\DeclareUnicodeCharacter{1D4D2}{\unicodecalC}
\DeclareUnicodeCharacter{1D4D3}{\unicodecalD}
\DeclareUnicodeCharacter{1D4DE}{\mathcal{O}}
\DeclareUnicodeCharacter{1D4E2}{\unicodecalS}
\DeclareUnicodeCharacter{1D4E3}{\unicodecalT}
\DeclareUnicodeCharacter{1D4E7}{\unicodecalX}
\DeclareUnicodeCharacter{1D4D0}{\ensuremath{\mathcal{A}}}
\DeclareUnicodeCharacter{1D4D1}{\ensuremath{\mathcal{B}}}
\DeclareUnicodeCharacter{1D4D6}{\ensuremath{\mathcal{G}}}
\DeclareUnicodeCharacter{1D4D7}{\ensuremath{\mathcal{H}}}
\DeclareUnicodeCharacter{1D4DB}{\ensuremath{\mathcal{L}}}
\DeclareUnicodeCharacter{1D4DC}{\ensuremath{\mathcal{M}}}
\DeclareUnicodeCharacter{1D4DD}{\unicodecalN}
\DeclareUnicodeCharacter{1D4B1}{\ensuremath{\mathcal{V}}}
\DeclareUnicodeCharacter{1D4E5}{\ensuremath{\mathcal{V}}}
\DeclareUnicodeCharacter{1D4E6}{\ensuremath{\mathcal{W}}}
\DeclareUnicodeCharacter{1D4E4}{\ensuremath{\mathcal{U}}}

\DeclareUnicodeCharacter{03B1}{\alpha}
\DeclareUnicodeCharacter{03B2}{\beta}
\DeclareUnicodeCharacter{03BC}{\mu}
\DeclareUnicodeCharacter{03B4}{\delta}
\DeclareUnicodeCharacter{03B5}{\varepsilon}
\DeclareUnicodeCharacter{03B7}{\eta}
\DeclareUnicodeCharacter{03BB}{\lambda}
\DeclareUnicodeCharacter{03C1}{\rho}
\DeclareUnicodeCharacter{03C8}{\psi}
\DeclareUnicodeCharacter{03C4}{\tau}
\DeclareUnicodeCharacter{03A8}{\Psi}
\DeclareUnicodeCharacter{03C3}{\sigma}
\DeclareUnicodeCharacter{03C6}{\varphi}
\DeclareUnicodeCharacter{03A6}{\Phi}
\DeclareUnicodeCharacter{03A3}{\Sigma}
\DeclareUnicodeCharacter{03D5}{\phi}
\DeclareUnicodeCharacter{03B8}{\theta}
\DeclareUnicodeCharacter{03C0}{\ensuremath{\pi}}
\DeclareUnicodeCharacter{0393}{\Gamma}
\DeclareUnicodeCharacter{0394}{\Delta}
\DeclareUnicodeCharacter{03BA}{\kappa}
\DeclareUnicodeCharacter{03BD}{\nu}
\DeclareUnicodeCharacter{25A0}{\blacksquare}
\DeclareUnicodeCharacter{25AA}{\blacksquare}

\newcommand{\hirayo}{\scaleobj{0.9}{\text{\usefont{U}{min}{m}{n}\symbol{'210}}}}
\DeclareUnicodeCharacter{3088}{\hirayo}
\DeclareFontFamily{U}{min}{}
\DeclareFontShape{U}{min}{m}{n}{<-> udmj30}{}

\newcommand\UnicodeWhiteRightPointingSmallTriangle{\triangleright}
\DeclareUnicodeCharacter{25B9}{\mathbin{\UnicodeWhiteRightPointingSmallTriangle}}
\newcommand\UnicodeWhiteDownPointingSmallTriangle{\triangledown}
\DeclareUnicodeCharacter{25BF}{\mathbin{\UnicodeWhiteDownPointingSmallTriangle}}
\newcommand\UnicodeWhiteUpPointingSmallTriangle{\scalemath{1}{-1}{{}^{\triangledown}}}
\DeclareUnicodeCharacter{25B5}{\mathbin{\UnicodeWhiteUpPointingSmallTriangle}}

\DeclareUnicodeCharacter{2080}{\ensuremath{{}_0}}
\DeclareUnicodeCharacter{2081}{\ensuremath{{}_1}}
\DeclareUnicodeCharacter{2082}{\ensuremath{{}_2}}
\DeclareUnicodeCharacter{2083}{\ensuremath{{}_3}}

\DeclareUnicodeCharacter{1D62}{\ensuremath{{}_i}}
\DeclareUnicodeCharacter{1D65}{\ensuremath{{}_v}}
\DeclareUnicodeCharacter{2C7C}{\ensuremath{{}_j}}
\DeclareUnicodeCharacter{02B3}{\ensuremath{{}^r}}
\DeclareUnicodeCharacter{02E1}{\ensuremath{{}^\ell}}
\DeclareUnicodeCharacter{1D48}{\ensuremath{{}^d}}
\DeclareUnicodeCharacter{1D50}{\ensuremath{{}^m}}
\DeclareUnicodeCharacter{1D58}{\ensuremath{{}^u}}
\DeclareUnicodeCharacter{209A}{\ensuremath{{}_p}}
\DeclareUnicodeCharacter{2096}{\ensuremath{{}_k}}

\DeclareUnicodeCharacter{2245}{\ensuremath{\cong}}
\DeclareUnicodeCharacter{2286}{\subseteq}

\DeclareUnicodeCharacter{22C5}{\cdot}
\DeclareUnicodeCharacter{25C3}{\ensuremath{\triangleleft}}
\DeclareUnicodeCharacter{25B9}{\ensuremath{\triangleright}}

\newcommand\smallmath[2]{#1{\raisebox{\dimexpr \fontdimen 22 \textfont 2
      - \fontdimen 22 \scriptscriptfont 2 \relax}{$\scriptscriptstyle #2$}}}
\newcommand\smalloplus{\smallmath\mathbin\oplus}

\DeclareUnicodeCharacter{2295}{\smalloplus}
\DeclareUnicodeCharacter{2297}{\otimes}
\DeclareUnicodeCharacter{214B}{\parr}
\DeclareUnicodeCharacter{2298}{\oslash}
\DeclareUnicodeCharacter{25C0}{\mathbin{\blacktriangleleft}}
\DeclareUnicodeCharacter{25C1}{\mathbin{\vartriangleleft}}
\DeclareUnicodeCharacter{22B3}{\mathbin{\triangleright}}
\DeclareUnicodeCharacter{22B2}{\mathbin{\triangleleft}}
\DeclareUnicodeCharacter{FF5C}{\mid}
\DeclareUnicodeCharacter{227A}{\mathbin{\prec}}
\DeclareUnicodeCharacter{227B}{\mathbin{\succ}}
\DeclareUnicodeCharacter{22A3}{\mathbin{\dashv}}
\DeclareUnicodeCharacter{219D}{\ensuremath{\leadsto}}
\DeclareUnicodeCharacter{1361}{\colon}

\DeclareUnicodeCharacter{29D1}{\mathrel{\multimapdotbothB}}
\DeclareUnicodeCharacter{29D2}{\mathrel{\multimapdotbothA}}
\DeclareUnicodeCharacter{22C4}{\mathbin{\diamond}}
\DeclareUnicodeCharacter{226B}{\mathrel{\gg}}
\DeclareUnicodeCharacter{25A1}{\Box}
\DeclareUnicodeCharacter{266F}{\sharp}

\DeclareUnicodeCharacter{2099}{_n}
\DeclareUnicodeCharacter{2098}{_m}
\DeclareUnicodeCharacter{1D5AD}{\ensuremath{\mathsf{N}}}
\DeclareUnicodeCharacter{1D4DF}{\mathcal{P}}

\DeclareUnicodeCharacter{2026}{\dots}
\DeclareUnicodeCharacter{226B}{\gg}
\DeclareUnicodeCharacter{2248}{\approx}
\DeclareUnicodeCharacter{03C9}{\omega}

\usepackage{stmaryrd}
\newcommand{\unicodeRelationalComposition}{\fatsemi}
\DeclareUnicodeCharacter{2A3E}{\unicodeRelationalComposition}
\DeclareUnicodeCharacter{112A9}{\mathbin{\pmb{;}}}

\usepackage{accsupp}
\newcommand*{\llbrace}{%
  \BeginAccSupp{method=hex,unicode,ActualText=2983}%
    \textnormal{\usefont{OMS}{lmr}{m}{n}\char102}%
    \mathchoice{\mkern-4.05mu}{\mkern-4.05mu}{\mkern-4.3mu}{\mkern-4.8mu}%
    \textnormal{\usefont{OMS}{lmr}{m}{n}\char106}%
  \EndAccSupp{}%
}
\newcommand*{\rrbrace}{%
  \BeginAccSupp{method=hex,unicode,ActualText=2984}%
    \textnormal{\usefont{OMS}{lmr}{m}{n}\char106}%
    \mathchoice{\mkern-4.05mu}{\mkern-4.05mu}{\mkern-4.3mu}{\mkern-4.8mu}%
    \textnormal{\usefont{OMS}{lmr}{m}{n}\char103}%
  \EndAccSupp{}%
}
\DeclareUnicodeCharacter{2983}{\llbrace}
\DeclareUnicodeCharacter{2984}{\rrbrace}
\DeclareUnicodeCharacter{22CA}{\rtimes}
\DeclareUnicodeCharacter{22C9}{\ltimes}

\DeclareUnicodeCharacter{03B3}{\gamma}
\DeclareUnicodeCharacter{2217}{\hyperlink{linkAst}{\ast}} %
\usepackage[xcolor,no patch,hyperref,quotation,electronic]{knowledge}
\knowledgeconfigure{notion}
\knowledgestyle{notion}{color=nordnight}
\knowledgestyle{intro notion}{color=nordnight}
\knowledgestyle{kl unknown}{}

\knowledge{notion}
| effectful triple
| effectful triples
| Effectful triple
| Effectful triples

\knowledge{notion}
| signature
| signatures
| Signature
| Signatures
\newcommand{\signature}{\kl{signature}}

\newcommand{\Signatures}{\kl{Signatures}}

\knowledge{notion}
| congruence
| congruences
| Congruence
| Congruences

\knowledge{notion}
| symmetry
| symmetries
| Symmetry
| Symmetries

\knowledge{notion}
| profunctor
| profunctors
| Profunctor
| Profunctors

\knowledge{notion}
| left inclusion
| Left inclusion

\knowledge{notion}
| right inclusion
| Right inclusion

\knowledge{notion}
| left whiskering
| Left whiskering

\knowledge{notion}
| right whiskering
| Right whiskering

\knowledge{notion}
| interchange axiom
| Interchange axiom

\knowledge{notion}
| inclusion and whiskering
| Inclusion and whiskering
| whiskering interacts with inclusion
| Whiskering interacts with inclusion

\knowledge{notion}
| whiskering
| Whiskering

\knowledge{notion}
| rewiring preserves whiskering
| Rewiring preserves whiskering
| whiskering preserves rewiring
| Whiskering preserves rewiring

\knowledge{notion}
| rewiring preserves interchange
| Rewiring preserves interchange

\knowledge{notion}
| rewiring preserves composition
| Rewiring preserves composition

\knowledge{notion}
| rewiring
| Rewiring
| rewire 
| Rewire
| definition of rewiring

\knowledge{rewiring is an action}[]{notion}

\knowledge{notion}
| composition
| Composition
| term composition
| Term composition

\knowledge{notion}
| term composition is associative
| Term composition is associative

\knowledge{notion}
| term composition is unital
| Term composition is unital

\knowledge{notion}
| substituting
| Substitution

\knowledge{notion}
| tensoring
| Tensoring

\knowledge{notion}
| arrow notation preterm
| Arrow notation preterm
| arrow notation preterms
| Arrow notation preterms

\knowledge{notion}
| arrow notation term
| Arrow notation term
| arrow notation terms
| Arrow notation terms

\knowledge{notion}
| observe-arrow notation
| Observe-arrow notation
| observe-arrow notation term
| Observe-arrow notation term
| observe-arrow notation terms
| Observe-arrow notation terms

\knowledge{notion}
| subdistribution
| Subdistribution
| subdistributions
| Subdistributions
| subdistributional

\knowledge{notion}
| rescaling
| Rescaling

\knowledge{notion}
| normalisation
| normalisations
| Normalisation
| Normalisations
| normalization
| normalizations
| Normalization
| Normalizations
| normalised
| Normalised

\knowledge{notion}
| observe-arrow notation copy-discard-compare category of terms
\newcommand{\obsArrow}{\kl[observe-arrow notation copy-discard-compare category of terms]{\ensuremath{\mathsf{obsArrow}}}}

\knowledge{notion}
| category of arrow notation terms
\newcommand{\preArrow}{\kl[category of arrow notation terms]{\ensuremath{\mathsf{preArrow}}}}

\knowledge{notion}
| copy-discard-compare categories
| copy-discard-compare category
| Copy-discard-compare categories
| Copy-discard-compare category

\usepackage[bibliography=common]{apxproof}

\renewcommand{\appendixprelim}{\clearpage\appendix}
\renewcommand{\appendixsectionformat}[2]{Appendix for Section {#1} (#2)} %
\renewcommand\Set{\hyperlink{linkSet}{\mathbf{Set}}}

\newcommand{\len}[1]{\mathsf{#1}}
\newcommand{\lenx}{n}
\newcommand{\leny}{m}
\newcommand{\lenz}{p}
\newcommand{\lenk}{k}
\newcommand{\lenu}{q}

\newcommand{\fsetx}{[n]}
\newcommand{\fsety}{[m]}
\newcommand{\fsetz}{[p]}
\newcommand{\fsetk}{[k]}
\newcommand{\fsetu}{[q]}

\newcommand\premonoidalCategory{\hyperlink{linkpremonoidal}{premonoidal category}}

\newcommand{\colorPre}[1]{{\color{nordred}{#1}}}
\newcommand{\colorMon}[1]{{\color{norddarkblue}{#1}}}

\newcommand\R{\colorPre{R}}

\newcommand\id[1][]{{{\mathrm{id}}_{#1}}}
\newcommand\Subd{\mathbf{D}_{\leq 1}}
\newcommand\Dst{\mathbf{D}}
\newcommand\return{\mathsf{return}}

\newcommand{\observe}{\mathsc{observe}}
\newcommand{\uniform}{\mathsc{uniform}}
\renewcommand{\return}{\mathsc{return}}

\newcommand{\host}{\mathrm{host}}
\newcommand{\caseOf}[1]{\mathsc{case}\ {#1}\ \mathsc{of}}

\newcommand{\car}{\mathit{car}}
\newcommand{\player}{\mathit{player}}

\newcommand{\pardon}{\mathit{pardon}}

\newcommand{\reply}{\mathit{reply\text{-}to\text{-}A}}

\newcommand{\ports}{\mathit{ports}}
\newcommand{\choice}{\mathit{choice}}
\newcommand{\correctness}{\mathit{correctness}}

\newcommand{\smasheq}[1]{\ensuremath{\smash{#1}}}

\newcommand{\prediction}{\mathit{prediction}}
\newcommand{\outcome}{\mathit{outcome}}

\newcommand{\guide}{\mathit{guide}}
\newcommand{\coin}{\mathit{coin}}

\newcommand{\subst}[3]{#1[#2 \!\setminus\! #3]}

\newcommand{\cont}{\mathrm{cont}}

\newcommand{\copyDiscardCategory}{\hyperlink{linkCopyDiscardCategory}{copy-discard category}}
\newcommand{\copyDiscardCategories}{\hyperlink{linkCopyDiscardCategory}{copy-discard categories}}
\newcommand{\strictCopyDiscardCompareCategory}{\hyperlink{linkStrictCopyDiscardCompareCategory}{strict copy-discard-compare category}}
\newcommand{\strictCopyDiscardCompareCategories}{\hyperlink{linkStrictCopyDiscardCompareCategory}{strict copy-discard-compare categories}}

\newcommand{\StrictCopyDiscardCompareCategories}{\hyperlink{linkStrictCopyDiscardCompareCategory}{Strict copy-discard-compare categories}}

\newcommand{\MontyHall}{\hyperlink{linkMontyHall}{Monty Hall}}
\newcommand{\MontyHallProblem}{\hyperlink{linkMontyHall}{Monty Hall problem}}
\newcommand{\MontyFallProblem}{\hyperlink{linkMontyFall}{Monty Fall problem}}
\newcommand{\MontyFall}{\hyperlink{linkMontyFall}{Monty Fall}}
\newcommand{\Fall}{\hyperlink{linkMontyFall}{Fall}}
\newcommand{\ThreePrisonersProblem}{\hyperlink{linkThreePrisonersProblem}{Three Prisoners problem}}
\newcommand{\SailorChildProblem}{\hyperlink{linkSailorChildProblem}{Sailor's child problem}}
\newcommand{\SailorsChild}{\hyperlink{linkSailorChildProblem}{Sailor's child}}
\newcommand{\NewcombsProblem}{\hyperlink{linkNewcombParadox}{Newcomb's paradox}}
\newcommand{\NewcombsParadox}{\hyperlink{linkNewcombParadox}{Newcomb's paradox}}

\newcommand{\imperfectNewcombsParadox}{\hyperlink{linkImperfectNewcombParadox}{imperfect Newcomb's paradox}}

\newcommand{\Sig}{\hyperlink{linkSignatureCategory}{\ensuremath{\mathbf{Sig}}}}
\newcommand{\signatureInterpretation}{\hyperlink{linkSignatureInterpretation}{signature interpretation}}

\newcommand{\signatureHomomorphism}{\hyperlink{linkSignatureHomomorphism}{morphism of signatures}}
\newcommand{\signatureHomomorphisms}{\hyperlink{linkSignatureHomomorphism}{morphisms of signatures}}

\newcommand{\copyDiscardCompareFunctor}{\hyperlink{linkCopyDiscardCompareFunctor}{copy-discard-compare functor}}
\newcommand{\copyDiscardCompareFunctors}{\hyperlink{linkCopyDiscardCompareFunctor}{copy-discard-compare functors}}

\newcommand{\CopyDiscardCompareFunctors}{\hyperlink{linkCopyDiscardCompareFunctor}{Copy-discard-compare functors}}
\newcommand{\copyDiscardCompareCategory}{\hyperlink{linkCopyDiscardCompareCategory}{copy-discard-compare category}}

\newcommand{\copyDiscardCompareCategories}{\hyperlink{linkCopyDiscardCompareCategory}{copy-discard-compare categories}}

\newcommand{\subdistribution}{\hyperlink{linkSubdistribution}{subdistribution}}
\newcommand{\subdistributions}{\hyperlink{linkSubdistribution}{subdistributions}}

\newcommand{\cdcCat}{\hyperlink{linkCdcCategory}{\ensuremath{\mathbf{Cdc}}}}

\newcommand{\observeDoNotationTerms}{\hyperlink{linkObserveDoNotationTerm}{arrow notation terms}}

\newcommand{\arrowNotation}{\hyperlink{linkArrowNotation}{arrow notation}}
\newcommand{\ArrowNotation}{\hyperlink{linkArrowNotation}{Arrow notation}}
\newcommand{\ArrowNotationTerms}{\hyperlink{linkArrowNotationTerm}{Arrow notation terms}}
\newcommand{\arrowNotationTerms}{\hyperlink{linkArrowNotationTerm}{arrow notation terms}}
\newcommand{\arrowNotationTerm}{\hyperlink{linkArrowNotationTerm}{arrow notation term}}

\newcommand{\arrowF}{\hyperlink{linkArrowF}{\ensuremath{\mathsf{arrow}}}}
\newcommand{\forgetF}{\hyperlink{linkArrowF}{\ensuremath{\mathsf{forget}}}}

\newcommand{\tensor}{\otimes}
\newcommand{\dcomp}{\fatsemi} %

\newcommand{\given}{\mid}

\newcommand{\normalisation}{\hyperlink{linkNormalisation}{normalisation}}
\newcommand{\Normalisation}{\hyperlink{linkNormalisation}{Normalisation}}
\newcommand{\normalisations}{\hyperlink{linkNormalisation}{normalisations}}

\newcommand{\norm}[1]{\operatorname{n}(#1)}

\newcommand{\axiomsOfArrowNotation}{\hyperlink{linkAxiomsOfArrowNotation}{axioms of arrow notation}}
\newcommand{\interchangeAxiom}{\hyperlink{linkAxiomsOfArrowNotation}{interchange axiom}}
\newcommand{\symmetryAxiom}{\hyperlink{linkAxiomsOfArrowNotation}{symmetry axiom}}

\newcommand{\FrobeniusAxiom}{\hyperlink{linkAxiomsOfArrowNotation}{Frobenius axiom}}
\newcommand{\idempotencyAxiom}{\hyperlink{linkAxiomsOfArrowNotation}{idempotency axiom}}

\newcommand{\vret}{\mathrm{ret}}
\newcommand{\vobs}{\mathrm{obs}}
\newcommand{\vsym}[2]{\kl[symmetry]{σ}_{#1,#2}}
\newcommand{\vinc}[2]{\iota_{#2}}
\newcommand{\vrnc}[2]{\nu_{#2}}
\newcommand{\vlinc}[1]{\iota_{#1}}
\newcommand{\vrinc}[1]{\nu_{#1}}

\newcommand{\vv}{\hyperlink{linkSignatureFixed}{v}}
\newcommand{\vcomp}[2]{{#1} \circ {#2}}

\newcommand{\prearrow}{\mathsf{prearrow}} %
\usepackage{amsthm}

\theoremstyle{plain}
\newtheoremrep{theorem}{Theorem}[section] %
\newtheoremrep{proposition}[theorem]{Proposition}

\newtheoremrep{lemma}[theorem]{Lemma}
\newtheoremrep{corollary}[theorem]{Corollary}
\theoremstyle{definition}
\newtheorem{definition}[theorem]{Definition}
\theoremstyle{remark}
\newtheorem{remark}[theorem]{Remark}

\newtheorem{example}[theorem]{Example} %
\usepackage{tikz}
\usepackage{tikz-cd}

\allowdisplaybreaks
\makeatletter
\def\@copyrightspace{\relax}
\makeatother

\begin{document}
\begin{frontmatter}
\title{A Simple Formal Language \\ for Probabilistic Decision Problems}
\author{Elena Di Lavore\thanksref{a}\thanksref{c}}
\author{Bart Jacobs\thanksref{b}}
\author{Mario Román\thanksref{c}}
\address[a]{Università di Pisa, Italy}
\address[b]{Radboud University Nijmegen, Netherlands}
\address[c]{University of Oxford, United Kingdom}
\maketitle 
\begin{abstract}
Probabilistic puzzles can be confusing, partly because they are formulated
in natural languages --- full of unclarities and ambiguities --- and partly because
there is no widely accepted and intuitive formal language to express them.  We
propose a simple formal language with arrow notation $(\gets)$ for sampling from
a distribution and with observe statements for conditioning (updating, belief
revision). We demonstrate the usefulness of this simple language by solving
several famous puzzles from probabilistic decision theory. The operational
semantics of our language is expressed via the (finite, discrete)
subdistribution monad. Our broader message is that proper formalisation dispels
confusion.
\end{abstract}
\end{frontmatter}
\section{Introduction}

Probabilistic decision problems are famously controversial: the
solutions to the \MontyHallProblem{}
\cite{selvin75:montyhall,vosSavant}, \NewcombsProblem{}
\cite{nozick1969newcomb}, or the \SailorChildProblem{}
\cite{elga2020self} have been much debated both in philosophy and
mathematics
\cite{gibbard1978counterfactuals,piccione97:imperfectrecall,neal06puzzles,yudkowsky2017functional}.
Care must be taken in interpreting all details of a problem: small changes in
interpretation may completely transform the mathematical
content of the problem and its solution
\cite{rosenthal08:montyFall,levinstein2020cheating}.  
This issue is compounded by the inherent ambiguity in natural language: 
from the
narrative of a decision problem, one can infer different implicit
assumptions to solve~it.
A formal language and a solving procedure rendering these assumptions explicit
can help settle controversies. It seems fair to say that there is no common
language and shared procedure for solving probabilistic decision problems.

We propose a simple syntax and semantics for probabilistic decision problems.
The syntax is close to the operational description of decision problems; it is
an idealised version of the syntax that the functional programming language
Haskell uses for its \emph{arrow} data
structure~\cite{heunen2006arrows,hughes2000generalising,jeffrey1997premonoidal1},
using arrows for sampling $(←)$ and $\observe$ statements for constraints. The
semantics is based on \subdistributions{}; it is simple enough to be followed
easily with pen and paper or implemented on a computer.\footnote{We have made
available an implementation of the \arrowNotation{} formalised in this paper as
a domain specific language built over Haskell's rebindable monadic combinators
(\url{https://github.com/mroman42/observe}).} We will provide explicit
examples, see \Cref{sec:illustrations} below.  The
\arrowNotation{} syntax that we propose is sound and complete for
\emph{\copyDiscardCompareCategories{}}, see
\Cref{th:arrowNotationInternalLanguage}.

\subsection{Outline} \Cref{sec:subdistributional} starts by recalling \subdistributions{}, and
the operations of restriction and rescaling on them.
\Cref{sec:illustrations} illustrates several examples of probabilistic
puzzles.  We show how to express the statements of such puzzles in
\arrowNotation{} and informally compute their semantics in
\subdistributions{} in a step-by-step manner.
\Cref{sec:observeDoNotation} formally defines \arrowNotation{}, shows how
its terms form the free \copyDiscardCompareCategory{} on a
signature, and gives functorial semantics to these terms.
\Cref{sec:normalisation} proves a result on compositionality of
\normalisation{}: the \subdistribution{} semantics can be computed
equivalently via their \normalisations{}, as proper distributions,
step-by-step. Most of the proofs, and the category theory involved, are relegated to
an extended version that will appear alongside this text~\cite{aSimpleLanguage}.

\subsection{Related work}

Markov categories are a well-studied categorical framework for synthetic
probability theory
\cite{cho2019disintegration,fritz2020synthetic,fritzLiang23:freeMarkov}. Recent
work developed the syntax of Markov categories and started employing it for
decision problems \cite{dilavore23:evidential,jacobs24:gettingwiser}. We never
mention Markov categories, but they inspire our variant of \arrowNotation{}. We
go beyond \cite{dilavore23:evidential} by introducing a simple syntax for
sampling and observation, which is close to what one finds in probabilistic
programming languages \cite{hasegawa97,
stay2013bicategorical,heunenKammar16:semanticsProbabilistic,
heunen_kammar_staton_yang_2017,stein2021structural,ehrhard2017measurable,
dahlqvist19semantics, vakar2019domain}. There, one usually employs more complex
syntax, less suitable for pen-and-paper computation. Our work tries to
constitute a minimal setup that, while less expressive, may be easier to employ
when discussing decision problems.  On the functional programming side, Gibbons
and Hinze already mention that Haskell's \arrowNotation{}  
\cite{hughes2000generalising,paterson2001new,heunen2006arrows} is suitable for
decision problems such as the \MontyHallProblem{} \cite{gibbons2011just}.
Finally, we note that alternatives in the literature require more constructors and
more structure: monadic semantics needs a cartesian closed base for the monad
\cite{moggi1991notions,fuhrmann2000structure,lindley10}; usual commutative
probabilistic programming asks for normalization to be a primitive
\cite{staton2017commutative}.

\section{Subdistributions}\label{sec:subdistributional}

\kl{Subdistributions} form the computational basis for our language.
\kl{Subdistributions} are finite formal sums whose coefficients --- non-negative real numbers --- %
represent a probabilistic choice between some elements of a set. Contrary to
distributions, whose probabilities must add up to exactly $1$, \kl{subdistributions} allow some
probability mass to be left unassigned: we use this left-over probability, if
any, to represent the failure of a certain property. This section introduces the
basic ingredients for our \kl{subdistributional} semantics.

\newcommand{\supp}{\operatorname{supp}}
\begin{definition}[Subdistribution]%
  \defining{linkSubdistribution}%
  \label{def:subdistribution}%
  \AP %
  A finitely supported \intro{subdistribution} on a set $X$ consists of a function $\sigma ፡ X → [0,1]$, from $X$ to the unit interval $[0,1] ⊆ ℝ$, such that:
  \begin{enumerate}
    \item its support, $\supp(\sigma) = \{ x ∈ X \mid
      \sigma(x) > 0 \}$, is finite,
      i.e., $\#|\supp(σ)| < \infty$;
    \item and its values add up to less or equal than one, i.e., \smasheq{\sum_{x ∈ \supp(σ)} σ(x) ≤ 1}.
  \end{enumerate}

  We shall write $\Subd(X)$ for the set of such \kl{subdistributions} on $X$,
  and $\Dst(X) ⊆ \Subd(X)$ for the subset of distributions, for which coefficients
  add up to precisely~$1$. For any two \kl{subdistributions} over two possibly
  different sets, $σ ∈ \Subd(X)$ and $ρ ∈ \Subd(Y)$, we write $σ ⊗ ρ ∈ \Subd(X ×
  Y)$ for their parallel (tensor) product, defined pointwise as $(σ ⊗ ρ)(x,y) =
  σ(x) · ρ(y)$.
\end{definition}

\begin{remark}[Ket notation]
  When the support is finite, $\supp(σ) = \{x_{1}, \ldots, x_{n}\}$, we often use ket-notation and
  write $σ$ as $r_{1}\ket{x_{1}} + \cdots + r_{n}\ket{x_n}$, where the
  probabilities, $r_{i} = \sigma(x_{i}) \in [0,1]$, satisfy $\sum_{i} r_{i} \leq
  1$.  We refer to each (formal) summand \(r_{i}\ket{x_{i}}\) of $σ$ as a
  ``monomial'' of $σ$.

  The product of \kl{subdistributions} is computed by pairwise multiplying the summands in the two \subdistributions{}.
  For instance, the second line (on the right) in \Cref{fig:sailorchild:solution} contains the product distribution
  \[ \begin{array}{rcl}
       \Big(\frac{1}{2}\ket{H} + \frac{1}{2}\ket{T}\Big)
       \otimes
       \Big(\frac{1}{2}\ket{A} + \frac{1}{2}\ket{B}\Big)
       & = &\frac{1}{4}\ket{H,A} + \frac{1}{4}\ket{H,B} +
       \frac{1}{4}\ket{T,A} + \frac{1}{4}\ket{T,B}.
     \end{array}\]
   In this way, we use tensor products $(⊗)$ to keep track of the state of the calculation as it develops.
\end{remark}

There is another feature of \kl{subdistributions} that is crucial in our
examples, namely, restriction with respect to a property (event,
observation). In combination with rescaling, restriction allows the
updating of a distribution.

\begin{definition}[Restriction \& rescaling]
\label{def:updating}%
Let $σ ∈ \Subd(X)$ be a \kl{subdistribution} over a set $X$.
\begin{enumerate}
  \item \label{def:updating:restriction} For a subset/constraint $U ⊆ X$ we define a
    \emph{restricted} \subdistribution{} $σ|_{U} ∈ \Subd(X)$ as $σ|_{U}(x) =
    σ(x)$ for \(x ∈ \supp(σ|_{U})\), where the support is the intersection of
    the support of \(σ\) with the subset \(U\),
    \(\supp(\sigma|_{U}) = \supp(\sigma) \cap U\).
  \item \label{def:updating:rescaling} \intro{Rescaling} is an isomorphism,
    $\operatorname{rescal} ፡ \Subd(X) → 1 + (0,1] × \Dst(X),$
  defined by
  \[\begin{array}{rcl}
    \operatorname{rescal}(\sigma) & = &
    \begin{cases}
    * ∈ 1, & \mbox{if $σ$ is the always-zero function}, \\
    (v, \frac{1}{v} · σ) \quad & \mbox{with validity } v = \sum_{x} σ(x).
    \end{cases}
    \end{array}\]
The $1+(-)$ in the output type of the rescaling function is thus used
to handle that rescaling is a partial operation.
\end{enumerate}
\end{definition}

Restriction happens in the examples in \Cref{sec:illustrations} via the crossing
out of parts of \kl{subdistributions}, in the columns to the right of the
formalisations (in the various figures).  The elements of the
\kl{subdistribution} that are not in the subset defined by the \observe{}
statement are removed, as in the intersection in the first point of the
definition.  \kl{Rescaling} happens at the very end of the examples, when a
validity $v$ is computed and used to turn the \kl{subdistribution} at that final
stage into a proper distribution, in $\Dst(X)$, that is the outcome of the
example. Equivalently, we could use the type $1 + (0,1] × \Dst(X)$ on the
right-hand-side
at all stages (and not just the last one), but that makes the notation
cumbersome; instead, in \Cref{sec:normalisation}, we will discuss
how \kl{normalisation} addresses this problem.

\section{Illustrations from probabilistic decision theory}%
\label{sec:illustrations}

This section describes several famous problems from decision theory, both
verbally and in a formal syntax — \arrowNotation{} — with associated
calculations. \Cref{sec:ex:montyHall} shows the first problem and intuitively
explains how to read \arrowNotation{}. \Cref{sec:observeDoNotation} will
introduce the notation formally, together with its semantics.

The language we propose is an idealised variant of Haskell's
\emph{\arrowNotation{}}~\cite{hughes2000generalising,paterson2001new,heunen2006arrows}
that includes a primitive ($\observe$) for declarative Bayesian
updates (using `sharp' equality
predicates~\cite{jacobs15:newDirections}).  Its idealised nature
allows us to sketch how it is a sound and complete language for the
algebraic structure of \copyDiscardCompareCategories{}, which we also
introduce in \Cref{sec:observeDoNotation}.

\subsection{Example --- the Monty Hall problem}%
\label{sec:ex:montyHall}

The \MontyHallProblem{} first appeared in a letter by Steve Selvin to
the editor of the \emph{American Statistician} in
1975~\cite{selvin75:montyhall}.  It was vos Savant's discussion in the
\emph{Parade} magazine, prompted by a reader (Craig F.~Whitaker), that
brought controversy and fame to the problem~\cite{vosSavant}.  We
reformulate the problem in a precise manner.
\begin{quote}\linkdef{linkMontyHall}{} \emph{\textbf{Monty Hall problem.} In a
  game show, (1)~one car is behind one of three doors --- Left, Middle, or
  Right --- where each option has the same probability. There is a goat behind
  the two other doors. (2)~The Player aims to win the car and (randomly) chooses
  a door, say Middle; this door remains closed at this stage; (3)~The Host knows
  where the car is and chooses a door different from Middle and different from
  the door that hides the car; the Host chooses randomly, if possible; the door
  chosen by the Host, say Left, is opened and discloses a goat. (4)~Given this
  situation, the Player is offered the option to either stick to the original
  choice (Middle) or switch to the other unopened (Right) door. Does switching
  doors give a higher probability of winning the car?}
\end{quote}

\Cref{fig:solveMontyHall:donotation} formalizes the \MontyHallProblem{} on the
left; we (manually) compute its solution on the right, in a step by step manner.

\begin{figure}[ht!]
\centering
\begin{tabular}{cll}
(1) & \smasheq{\car \gets \uniform\{L, M, R\}} \quad\quad
& \smasheq{\frac{1}{3}\ket{L} + \frac{1}{3}\ket{M} + \frac{1}{3}\ket{R}}
\\[+0.2em]
(2) & \smasheq{\host ← \caseOf{\car}} \\
    & \smasheq{\qquad L \mapsto 1\ket{R}} \\
    & \smasheq{\qquad M \mapsto \frac{1}{2}\ket{L} + \frac{1}{2}\ket{R}} \\
    & \smasheq{\qquad R \mapsto 1\ket{L}}
    & \smasheq{\frac{1}{3}\ket{L,R} + \frac{1}{6}\ket{M,L} + \frac{1}{6}\ket{M,R} + \frac{1}{3}\ket{R,L}}\\[+0.2em]
(3) &\smasheq{\observe(\host = L)} &
    \smasheq{\cancel{\frac{1}{3}\ket{L,R}} + \frac{1}{6}\ket{M,L} + \cancel{\frac{1}{6}\ket{M,R}} + \frac{1}{3}\ket{R,L}}
    \\[+0.2em]
    (4) &\smasheq{\return(\car)} &
    \smasheq{\frac{1}{6}\ket{M} + \frac{1}{3}\ket{R}}
    \\[+0.5em]
    & Validity: & \smasheq{\frac{1}{6} + \frac{1}{3} = \frac{1}{2}}
    \\[+0.2em]
    & Posterior: & \smasheq{\frac{1}{3}\ket{M} + \frac{2}{3}\ket{R}}
  \end{tabular}
  \caption{Formal description and calculations for the Monty Hall
    problem, with the player choosing the middle door and the host
    opening the left door. A crucial point is that only when the car
    is behind the middle door, the host has a choice.}
  \label{fig:solveMontyHall:donotation}
\end{figure}

\subsection{About Arrow Notation.}%
\label{sec:reading-arrow}%

When writing \arrowNotation{} statements, we will keep a \subdistribution{} on
the side, see for instance in \Cref{fig:solveMontyHall:donotation},
line~\((1)\).  This \subdistribution{} represents the ``state'' of the current
computation: it starts with the first arrow declaration and is recomputed after
each statement.

Every time we write a function statement, we gather all the formal monomials
\(r_{i}\ket{x_{i}}\) corresponding to the previous line; we compute the function
over each one of the possible outcomes \(x_{i}\), obtaining a \subdistribution{}
\(\sigma_{i} = \sum_{j} s_{j}^{i}\ket{y_{j}^{i}}\) for each one of them; we
merge these \subdistributions{} into a joint \subdistribution{}, \smasheq{\sum\nolimits_{i}
r_{i}\sigma_{i}\ket{x_{i}} = \sum_{i}r_{i} \sum_{j} s_{j}^{i}
\ket{x_{i},y_{j}^{i}}}.
In this way, each monomial keeps a list of values, listed in the order
of appearance of the variables in the problem, see
\Cref{fig:solveMontyHall:donotation}, line~\((2)\).  The $\mathsc{case-of}$
formulation is used to define a function on a finite set by listing
its action on the elements of the set.  The ``$\mathsc{case-of}$'' statement
will not be part of the formal syntax of \arrowNotation{} — it only
appears in this subdistributional interpretation, which is discrete and finitely supported.

When writing an $\observe$ statement, we cancel out all monomials that do not
satisfy its constraint. The validity of the resulting term may decrease, as in
\Cref{fig:solveMontyHall:donotation}, line~\((3)\). Finally, whenever we write
the $\return$ statement, we keep on the side — for each monomial —
only the monomials corresponding to the variables being returned, see
\Cref{fig:solveMontyHall:donotation}, line~\((4)\). The $\return$ statement acts
as a projection, producing the output.  The operational semantics sketched here
is a reformulation of the monadic semantics of the \subdistribution{} monad,
$\Subd{} ፡ \Set{} → \Set{}$~\cite{paterson2001new,moggi1991notions}.

Finally, the outcome ``posterior'' distribution is obtained by rescaling the
final \subdistribution{}, as in
Definition~\ref{def:updating}~\eqref{def:updating:rescaling}.  In the Monty Hall
example, in \Cref{fig:solveMontyHall:donotation}, the posterior reflects that,
from the fact that the host has opened the left door, we can deduce that the car
is with probability \smasheq{\frac{1}{3}} behind the middle door, and with
probability \smasheq{\frac{2}{3}} behind the right door.  Hence, as originally
argued by vos Savant, it does make sense to switch --- from middle to right ---
intuitively because the host reveals information.

Like in the above description, it is assumed that the player chooses the middle
door and the host opens the left door, while the car is initially at a
(uniformly) random position. The description in
\Cref{fig:solveMontyHall:donotation} can be generalised, including also the
choice of the player, with the statement ``$\player \gets \uniform\{L, M,
R\}$''.  The host then has to make a case distinction over~$9$ options.  We
choose to keep things elementary at this stage and describe a simplified
situation, with the chosen and opened doors already fixed, but the interested
reader may wish to elaborate this more general formulation or check
\Cref{sec:normalisation} and \Cref{fig:solveMontyHall:donotation:normalise}.

\subsection{Example --- the Monty `Fall' problem}%
\label{sec:montyFall}%

The \MontyFallProblem{} is a variant of the \MontyHallProblem{} where
the host opens a door undeliberately and unconsciously, instead of
consciously avoiding to pick the door with the car behind. This
formulation, due to Rosenthal~\cite{rosenthal08:montyFall}, is a
well-known variant that illustrates how delicate the statement of the
\MontyHallProblem{} \cite{gill10:montyhall,gibbons2011just} is: what
makes decision theory challenging to formalise is that such slight
variations on the statement may lead to radically different
conclusions.
\begin{quote}
\defining{linkMontyFall}{} \emph{\textbf{Monty Fall problem.} In this
variant, once Player has selected one of the three doors, say the
Middle once again, the Host slips on a banana peel and accidentally
pushes open another door, say the Left one; everyone sees that there
is no car behind it. Does it still make sense for Player to switch?}
\end{quote}

\noindent In this situation the knowledge of the host where the car is
hidden does not play a role, so the player gets no additional
information.  Rosenthal argues~\cite{rosenthal08:montyFall} that, if
people distrusted vos Savant's solution~\cite{vosSavant} of the
problem --- that it makes sense to switch --- they may have been
thinking of this `\MontyFall{}' version instead. We formalise it as
follows.
\begin{center}
  \begin{tabular}{lll}
    (1) & \smasheq{\car \gets \uniform \{L,M,R\}} \qquad \qquad &
    \smasheq{\frac{1}{3}\ket{L} + \frac{1}{3}\ket{M} + \frac{1}{3}\ket{R}}
    \\[+0.2em]
    (2) & \smasheq{\observe (\car \neq L)} & \smasheq{\frac{1}{3}\cancel{\ket{L}} + \frac{1}{3}\ket{M} + \frac{1}{3}\ket{R}}
    \\[+0.2em]
    (3) & \smasheq{\return(\car)} & \smasheq{\frac{1}{3}\ket{M} + \frac{1}{3}\ket{R}}
    \\[+0.5em]
    & Validity: & \smasheq{\frac{2}{3}}
    \\[+0.2em]
    & Posterior: & \smasheq{\frac{1}{2}\ket{M} + \frac{1}{2}\ket{R}}
  \end{tabular}
\end{center}

\noindent Now, both positions of the car are equally likely, so it
does not make sense to switch. We skip any polemic here, since our
point is that formalisation helps to dissolve confusion. Once the
assumptions are rendered explicit, the \MontyHall{} or \Fall{}
problem, has a single, easily computable solution.

\subsection{Example --- Three Prisoners problem}
\label{sec:prisoners}

The next challenge may be found, for instance,
in the book by Casella and Berger~\cite[Ex.~1.3.4]{casella2002:statistical} and is presented below in our own formulation.
\begin{quote}
\defining{linkThreePrisonersProblem}{} \emph{\textbf{Three prisoners
  problem.} Three prisoners named $A$, $B$, and $C$, are on death row,
in isolation, without any communication between them. The responsible
governor decides to pardon one of them and chooses at random the
prisoner to pardon. She informs the warden of the prisoners of her
choice but does not allow him to give information to any of the
prisoners, about their own fate. The warden is honest and careful and
does not lie. Prisoner $A$ tries to get the warden to tell him who
has been pardoned. The warden refuses; $A$ then asks which of $B$ or
$C$ will be executed. The warden thinks for a while, then tells $A$
that $B$ is to be executed. Does $A$ learn anything about his own
fate?}
\end{quote}

\noindent Prisoner $A$ may think that his chances have risen from
$\frac{1}{3}$ to $\frac{1}{2}$. But this is not correct: the situation
of $A$ has not changed, but the chance of $C$ of being pardoned has
risen to $\frac{2}{3}$; see the formalisation in
\Cref{fig:threeprisoners:computation}.

\begin{figure}[ht!]
  \centering
  \begin{tabular}{lll}
    (1) & \smasheq{\pardon ← \uniform\{A, B, C\}} \quad\quad
    & \smasheq{\frac{1}{3}\ket{A} + \frac{1}{3}\ket{B} + \frac{1}{3}\ket{C}}
    \\[+0.2em]
    (2) & \smasheq{\reply ← \caseOf{\pardon}}
    \\
     & \smasheq{\qquad A \mapsto \frac{1}{2}\ket{B} +  \frac{1}{2}\ket{C};} \\ %
     & \smasheq{\qquad B \mapsto 1\ket{C};} \\ %
     & \smasheq{\qquad C \mapsto 1\ket{B};} &
    \smasheq{\frac{1}{6}\ket{A,B} + \frac{1}{6}\ket{A,C}
      + \frac{1}{3}\ket{B,C} + \frac{1}{3}\ket{C,B}}
    \\[+0.2em]
    (3) & \smasheq{\observe(\reply = B)}
    & \smasheq{\frac{1}{6}\ket{A,B} + \cancel{\frac{1}{6}\ket{A,C}}
    + \cancel{\frac{1}{3}\ket{B,C}} + \frac{1}{3}\ket{C,B}}
    \\[+0.2em]
    (4) & \smasheq{\return(\pardon)} & \smasheq{\frac{1}{6}\ket{A} + \frac{1}{3}\ket{C}}
    \\[+0.5em]
    & Validity: & \smasheq{\frac{1}{6} + \frac{1}{3} = \frac{1}{2}}
    \\[+0.2em]
    & Posterior: & \smasheq{\frac{1}{3}\ket{A} + \frac{2}{3}\ket{C}}
  \end{tabular}
  \caption{Formulation and calculation for the
    \protect\ThreePrisonersProblem{}, where the reply to $A$ is about
    who gets executed, and thus not pardoned.}
  \label{fig:threeprisoners:computation}
\end{figure}

\subsection{Example --- Sailor's Child problem}
\label{sec:sailorchild}

The next ``\SailorsChild{}'' problem is another example of a problem whose solution has created controversy~\cite{neal06puzzles}.
It is an equivalent formulation of the ``Sleeping Beauty'' problem~\cite{elga2020self}, without some of its philosophically contentious points.

\begin{quote}
  \defining{linkSailorChildProblem}{} \emph{\textbf{Sailor's child
    problem.} A Sailor sails regularly between two ports $A$ and
  $B$. In each of these he stays with a woman, both of whom wish to
  have a child by him. The Sailor decides that he will have either one
  child (with one of the women) or two children (one child with each
  of them). The number of children is decided by a (fair) coin toss
  --- one if Heads, two if Tails. Furthermore, the Sailor decides that
  if the coin lands Heads, he will have the selected option of one
  child with the woman who lives in the city listed first in The
  Sailor’s Guide to Ports, a book that is actually unkown to
  him. Hence the choice between ports $A$ and $B$ is in this case (of
  Heads) decided by chance, in a fair way.}

  \emph{\indent{} Now, suppose that you are a child of this Sailor,
  born and living in port $A$, and that neither you nor your mother
  know whether he had a child with the other woman. You do not have a
  copy of the book, but you do know that matters were decided as
  described above. What is the probability that you are the Sailor's
  only child?}
\end{quote}

\begin{figure}[h!]
  \centering
  \begin{tabular}{lll}
    (1) & \smasheq{\coin ← \uniform \{H,T\}}
    & \smasheq{\frac{1}{2}\ket{H} + \frac{1}{2}\ket{T}}
    \\[+0.2em]
    (2) & \smasheq{\guide ← \uniform \{A,B\}}
    & \smasheq{\frac{1}{4}\ket{H,A} + \frac{1}{4}\ket{H,B} +
       \frac{1}{4}\ket{T,A} + \frac{1}{4}\ket{T,B}}
    \\[+0.2em]
    (3) & \smasheq{\ports  ← \mathsc{case}\ (\coin,\guide)\ \mathsc{of}} \quad
    \\
    & \smasheq{\quad (H,A) \mapsto 1\ket{\{A\}}}
    \\
    & \smasheq{\quad (H,B) \mapsto 1\ket{\{B\}}} &
    \\
    & \smasheq{\quad (T,A) \mapsto  1\ket{\{A,B\}}}
    & \smasheq{\frac{1}{4}\ket{H,A,\{A\}} + \frac{1}{4}\ket{H,B,\{B\}}}
    \\
    & \smasheq{\quad (T,B) \mapsto  1\ket{\{A,B\}}}
    & \smasheq{\quad+\,\frac{1}{4}\ket{T,A,\{A,B\}} +
       \frac{1}{4}\ket{T,B,\{A,B\}}}
    \\[+0.2em]
    (4) & \smasheq{\observe(A \in\ports)} &
   \smasheq{\frac{1}{4}\ket{H,A,\{A\}} + \cancel{\frac{1}{4}\ket{H,B,\{B\}}}}
   \\
   & & \smasheq{\quad+\,\frac{1}{4}\ket{T,A,\{A,B\}} +
       \frac{1}{4}\ket{T,B,\{A,B\}}}
    \\[+0.2em]
    (5) & \smasheq{\return(\coin)} & \smasheq{\frac{1}{4}\ket{H} + \frac{1}{2}\ket{T}}
    \\[+0.5em]
    & Validity: & \smasheq{\frac{1}{4} + \frac{1}{2} = \frac{3}{4}}
    \\[+0.2em]
    & Posterior: & \smasheq{\frac{1}{3}\ket{H} + \frac{2}{3}\ket{T}}
  \end{tabular}
  \caption{Calculations for the \protect\SailorChildProblem{}.}
  \label{fig:sailorchild:solution}
\end{figure}

The story is rather complex so we do not expect that readers immediately have an
intuition, like in Subsections~\ref{sec:ex:montyHall} and~\ref{sec:prisoners}.
Again, the formalisation is thus very helpful to highlight the assumptions and
provide the answer.

In philosophy, the \SailorsChild{} problem is used to exemplify the
\emph{anthropic principle}~\cite{neal06puzzles,elga2020self}: the mere presence
of the deciding agent may be an important bit of information to update on.
Mathematically, we can simplify this discussion and treat ``anthropic''
observations no differently from usual observations: if we had disregarded the
anthropic observation, we would have omitted line \((4)\) and obtained a
different result, namely, the uniform distribution on \(\{H,T\}\). At the end,
it is interesting to note that the fact that the child lives in port~$A$ is
irrelevant. The outcome --- \smasheq{\frac{1}{3}} probability for the single child
option --- is the same if the child lives in port~$B$. 

\subsection{Example --- Newcomb's Paradox}
\label{sec:newcombs}

\NewcombsParadox{} is famously controversial: different paradigms of decision
theory answer it differently~\cite{nozick1969newcomb,gibbard1978counterfactuals}
and therefore it is often called a paradox. It is generally accepted that a
naive formalisation that does not take into account statistical correlations
produces a wrong answer~\cite{ahmed2014evidence}.
\begin{quote}
\defining{linkNewcombParadox}{} \emph{\textbf{Newcomb's paradox.} Suppose there is a Being that can precisely predict your
choices. There are two boxes in front of you, B1 and B2, where B1
certainly contains \$1. B2 contains either \$10 or nothing. You can
choose between: (a)~taking what is in both boxes, or (b)~taking only
what is in B2.  The Being acts as follows. If it predicts that you
will take what is in both boxes, it forces B2 to be empty. If it
predicts that you will only take what is in B2, it makes sure that B2 contains \$10.}

\emph{Things work as follows. First the Being makes its prediction
about your choice. Then it puts the \$10 in B2, or not, in line with
its prediction. Now you make your choice. Which choice maximises the
outcome?}
\end{quote}

\begin{figure}[ht!]
  \centering
  \begin{tabular}{cll}
    (1) & \smasheq{\prediction \gets \uniform\{a,b\} \qquad} &
    \smasheq{\frac{1}{2}\ket{a} + \frac{1}{2}\ket{b}}
    \\[+0.2em]
    (2) & \smasheq{\choice \gets \uniform\{a,b\} \qquad} &
    \smasheq{\frac{1}{4}\ket{a,a} + \frac{1}{4}\ket{a,b} +
    \frac{1}{4}\ket{b,a} + \frac{1}{4}\ket{b,b}}
    \\[+0.2em]
    (3) & \smasheq{\observe(\prediction = \choice)} &
    \smasheq{\frac{1}{4}\ket{a,a} + \cancel{\frac{1}{4}\ket{a,b}} +
    \cancel{\frac{1}{4}\ket{b,a}} + \frac{1}{4}\ket{b,b}}
    \\[+0.2em]
    (4) & \smasheq{\outcome ← \mathsc{case}\ (\prediction,\choice)\
      \mathsc{of}} \quad
    \\
    & \smasheq{\quad (a,a) \mapsto  1\ket{\$1}} & \\
    & \smasheq{\quad (a,b) \mapsto  1\ket{\$0}} & \\
    & \smasheq{\quad (b,a) \mapsto  1\ket{\$11}} & \\
    & \smasheq{\quad (b,b) \mapsto  1\ket{\$10}} &
    \smasheq{\frac{1}{4}\bigket{a,a,\$1} + \frac{1}{4}\bigket{b,b,\$10}}
    \\[+0.2em]
    & & \\
    (5a) &\smasheq{\observe(\choice = a)} &
    \smasheq{\frac{1}{4}\bigket{a,a,\$1}}
    \\[+0.2em]
    (6a) &\smasheq{\return(\outcome)} &
    \smasheq{1\bigket{\$1}}\\
    && \\
    (5b) &\smasheq{\observe(\choice = b)} &
    \smasheq{\frac{1}{4}\bigket{b,b,\$10}}
    \\[+0.2em]
    (6b) &\smasheq{\return(\outcome)} &
    \smasheq{1\bigket{\$10}}\\
  \end{tabular}
  \caption{Solution of \protect\NewcombsProblem{}, where we leave
    the final normalization step implicit.}
  \label{fig:solveNewcomb}
\end{figure}

\noindent Having made the relevant case distinctions in
\Cref{fig:solveNewcomb}, we separately elaborate the two choice
scenarios~(a) and~(b), in steps~(5) and~(6). We see that choosing the one-box option~(b) gives the highest outcome.

\subsection{Example --- Imperfect Newcomb}%
\label{sec:newcomb:imperfect}
\defining{linkImperfectNewcombParadox}{}

\begin{figure}[ht!]
  \centering
  \begin{tabular}{cll}
    (1) & \smasheq{\prediction \gets \uniform\{a,b\} \qquad} &
    \smasheq{\frac{1}{2}\ket{a} + \frac{1}{2}\ket{b}}
    \\[+0.2em]
    (2) & \smasheq{\choice \gets \uniform\{a,b\} \qquad} &
    \smasheq{\frac{1}{4}\ket{a,a} + \frac{1}{4}\ket{a,b} +
    \frac{1}{4}\ket{b,a} + \frac{1}{4}\ket{b,b}}
    \\[+0.2em]
    (3) & \smasheq{\correctness ← \mathsc{case}\ (\prediction,\choice)\
      \mathsc{of}} \quad
    \\
    & \smasheq{\quad (x,x) \mapsto \frac{4}{5}\ket{T} + \frac{1}{5}\ket{F}} &
    \\[+0.2em]
    & \smasheq{\quad (x,y) \mapsto \frac{1}{5}\ket{T} + \frac{4}{5}\ket{F},
       \quad\mbox{for }x\neq y} &
    \\[+0.2em]
    (4) & \smasheq{\observe(\correctness = T)} &
    \smasheq{\frac{1}{5}\ket{a,a,T} + \frac{1}{20}\ket{a,b,T}}
    \\
    & & \smasheq{\quad+\,\frac{1}{20}\ket{b,a,T} + \frac{1}{5}\ket{b,b,T}}
    \\[+0.2em]
    (5) & \smasheq{\outcome ← \mathsc{case}\ (\prediction,\choice)\
      \mathsc{of}} \quad
    \\
    & \smasheq{\quad (a,a) \mapsto  1\ket{\$1}} & \\
    & \smasheq{\quad (a,b) \mapsto  1\ket{\$0}} & \\
    & \smasheq{\quad (b,a) \mapsto  1\ket{\$11}} &
    \smasheq{\frac{1}{5}\ket{a,a,T,\$1} + \frac{1}{20}\ket{a,b,T,\$0}}
    \\
    & \smasheq{\quad (b,b) \mapsto  1\ket{\$10}} &
    \smasheq{\quad+\,\frac{1}{20}\ket{b,a,T,\$11} + \frac{1}{5}\ket{b,b,T,\$10}}
    \\[+0.2em]
    & & \\
    (6a) &\smasheq{\observe(\choice = a)} &
    \smasheq{\frac{1}{20}\ket{a,b,T,\$11} + \frac{1}{5}\ket{a,a,T,\$1}}
    \\[+0.2em]
    (7a) &\smasheq{\return(\outcome)} &
    \smasheq{\frac{1}{5}\ket{\$11} + \frac{4}{5}\ket{\$1}}
    \\
    && \\
    (6b) &\smasheq{\observe(\choice = b)} &
    \smasheq{\frac{1}{5}\ket{b,b,T,\$10} + \frac{1}{20}\ket{a,b,T,\$0}}
    \\[+0.2em]
    (7b) &\smasheq{\return(\outcome)} &
    \smasheq{\frac{4}{5}\ket{\$10} + \frac{1}{5}\ket{\$0}}
  \end{tabular}
  \caption{Solution of the \protect\imperfectNewcombsParadox.}
  \label{fig:solveNewcomb:imperfect}
\end{figure}

We consider a variation on \NewcombsParadox{} where the Being has a $80\%$
chance of predicting right, both for equality and inequality. One may ask what
is then the best strategy. The answer is elaborated in
\Cref{fig:solveNewcomb:imperfect}. The correctness of the prediction is
expressed as a Boolean: True ($T$) or False ($F$).  For readability we skip some
of the subdistributions in the column on the right. We now have to use the
expected value to evaluate the outcome. For choice (b) of 2 boxes it is now more
than 1, namely \smasheq{\frac{1}{5}\cdot \$11 + \frac{4}{5}\cdot \$1 = \$3}.
When only 1 box is chosen one still gets a higher outcome
\smasheq{\frac{4}{5}\cdot \$10 + \frac{1}{5}\cdot \$0 = \$8}, but this lower
than the outcome of \(\$10\) in the perfect case. Hence when the Being's
prediction is not perfect, it makes less sense to choose one box.

\section{Arrow Notation}\label{sec:observeDoNotation}

\ArrowNotation{} (more precisely, ``\arrowNotation{} for
\copyDiscardCategories{}'', as we see later) is an idealised version of the
syntax Haskell uses for its \emph{arrow} data structure%
\footnote{Haskell's notation for arrows is sometimes called \emph{do-notation}.
We prefer to use the term \emph{arrow notation} to avoid confusion with Pearl's
\emph{do-calculus} \cite{pearl2009causality}, which is unrelated and common
in causality theory.}~%
\cite{heunen2006arrows,hughes2000generalising,jeffrey1997premonoidal1}: it
consists of a series of statements declaring the input and output variables of
every function, ended by a ``\return{}'' statement.

\subsection{Arrow notation}  

Let us first define an idealised version of \arrowNotation{} as a simple type
theory. We start from a \signature{} $Σ$ of types and generators, and define
terms by the rules in \Cref{def:arrowNotation} and \Cref{fig:dotypetheory}.

\begin{definition}[Signature]
  \label{def:signature}%
  \AP A ""signature"" $Σ$ consists of a set of generating types, $Σ_{type}$,
  together with, for each list of input types, $X_1, …, X_{\lenx} ∈ Σ_{type}$,
  of length $\lenx ∈ ℕ$, and each list of output types, $Y_1,…,Y_\leny ∈ Σ_{type}$, of
  length $\leny ∈ ℕ$, a set of generators, \smasheq{Σ\big(X_1,...,X_{\lenx}; Y₁, …,
  Y_{\leny}\big)}.
\end{definition}

\begin{toappendix}
\begin{definition}[Morphism of signatures]
  A \emph{morphism of signatures} $H ፡ Σ → Ψ$ is given by a function $H ፡
  Σ_{type} → Ψ_{type}$ together with a collection of functions
  $$Σ\big(X₁,...,X_{\len{x}}; Y₁,...,Y_{\len{y}}\big) → Ψ\big(H(X₁),...,H(X_{\len{x}});
  H(Y₁),...,H(Y_{\len{y}})\big).$$
\end{definition}
\end{toappendix}

\begin{toappendix}
\begin{proposition}
  \Signatures{} and \signatureHomomorphisms{} form a category, \Sig{}.
\end{proposition}
\end{toappendix}

\begin{definition}[Arrow notation]\label{def:observeDoNotation}\label{def:arrowNotation}%
  \defining{linkArrowNotation}{}\defining{linkObserveDoNotationTerm}{}\defining{linkArrowNotationTerm}{}%
  \AP An \intro{arrow notation preterm}, over a \kl{signature} $Σ$, with context given by 
  $Γ = x_1 : X_1, ..., x_{\lenx} : X_{\lenx}$, and
  of output type $Δ ∈ \mathsf{List}(Σ_{type})$, is inductively defined to be either
  \begin{enumerate}
    \item a return statement, \smasheq{Γ ⊢ \return(x_{α(1)},…,x_{α(\leny)}) :
    X_{α(1)},…,X_{α(\leny)}}, for any list of (possibly repeated) variables
    from the context, $(x_{α(i)} : X_{α(i)}) ∈ Γ$ for $i = 1,…,\leny$, that are
    picked according to a function between finite sets, \smasheq{α ፡ \fsety → \fsetx};
    \item a generator statement, $Γ ⊢ (y_1,…,y_{\leny} ←
    f(x_{β(1)},…,x_{β(\lenk)}) 𑊩 t) : Δ$, for any correctly typed generator
    $f ∈ Σ(X_{β(1)},...,X_{β(\lenk)}; Y_1,...,Y_{\leny})$ with $\lenk$
    inputs and $\leny$ outputs, any list of (possibly repeated) variables
    from the context $(x_{β(i)} : X_{β(i)}) ∈ Γ$ for $i = 1,…,\lenk$, picked
    according to a function between finite sets, $β ፡ \fsetk → \fsetx$, and, finally, any
    choice of fresh variables $y₁ : Y₁,…,y_{\leny} : Y_{\leny}$, and any
    term accessing these fresh variables, $Γ, y₁ : Y₁,…,y_{\leny} : Y_{\leny} ⊢ t : Δ$.
  \end{enumerate}
  \begin{figure}[ht!]
    \centering
    \begin{mathpar}
      \inferrule
        {Γ = x_1 : X_1,..., x_\lenx : X_\lenx \\
        α ፡ \fsety → \fsetx}
        {Γ ⊢ \return(x_{α(1)}, …, x_{α(\leny)}) ፡ X_{α(1)}, …, X_{α(\leny)}}
      \\
      \inferrule[for each generator, {{$f ∈ Σ(X_{β(1)}, …, X_{β(\lenk)};Y₁, …,Y_{\leny})$}}]
          {%
          Γ = x_1 : X_1,..., x_\lenx : X_\lenx \\ 
          β ፡ \fsetk → \fsetx \\ 
          Γ, y_1 : Y_1, …, y_{\leny} : Y_{\leny} ⊢ t : Δ}
          {Γ ⊢ (y₁, …, y_{\leny} ← f(x_{β(1)}, …, x_{β(\lenx)}) 𑊩 t) ፡ Δ}
    \end{mathpar}
    \caption{Rules for \protect\kl{arrow notation preterms}.}
    \label{fig:dotypetheory}
  \end{figure}  
Reading \arrowNotation{} becomes easier when we replace the statement
separator symbol $(𑊩)$ by a line jump. As exemplified in
\Cref{sec:illustrations}, we do so when convenient.

  We consider \kl{arrow notation preterms} to be quotiented up to
  $α$-equivalence: renaming of variables does not change their meaning.
  Substitution is defined as usual: any term $Γ ⊢ t : Δ$ with two variables of
  the same type, $x : X$ and $u : X$, induces a term $Γ ⊢ \subst{t}{x}{u} : Δ$
  where all occurrences of $x$ have been substituted by the new variable $u$.
\end{definition}
\begin{remark}[Arrow notation combinators]
Variables are useful for intuition, but not necessary. We can reformulate
\arrowNotation{} (\Cref{fig:var:arrowNotation}) in terms of combinators and
functions between finite sets. While arguably less readable for humans, the
combinatorial form sidesteps variable management and helps formalization
(\Cref{th:doNotationInternalLanguage}).
  An \kl{arrow notation preterm} is exactly either
  \begin{enumerate}
    \item a $\return$ statement, $\vret(α) ፡ X₁,...,X_{\lenx} →
    X_{α(1)},...,X_{α(\leny)}$, for any function, $α ፡ \leny → \lenx$;
    \item a generator statement, $f(β) 𑊩 t ፡ X₁,..,X_{\lenx} →
    Z₁,..,Z_{\lenz}$, for any function $β ፡ \lenk → \lenx$, any generator
    $f ∈ Σ(X_{β(1)},..,X_{β(\lenk)}; Y_1,..,Y_{\leny})$, and any term $t ፡
    X₁,...,X_{\lenx},Y₁,...,Y_{\leny} → Z₁,...,Z_{\lenz}.$
  \end{enumerate}
  In other words, the rules of \Cref{fig:var:arrowNotation} are equivalent to
  the rules of \Cref{fig:dotypetheory}.
  \begin{figure}[ht!]
    \centering
    \begin{mathpar}
      \inferrule
        {α ፡ \fsety → \fsetx}
        {\vret(α) ፡ X₁,...,X_{\lenx} → X_{α(1)},...,X_{α(\leny)}}
      \qquad
      \inferrule[for each generator, {{$f ∈ Σ(X_{β(1)}, …, X_{β(\lenk)};Y₁, …,Y_{\leny})$}}]
          {t ፡ X₁,..., X_{\lenx},Y₁,...,Y_{\leny} → Z₁,...,Z_{\lenz}
          \qquad  β ፡ \fsetk → \fsetx}
          {(f(β) 𑊩 t) ፡ X₁,..., X_{\lenx} → Z₁,...,Z_{\lenz}}
    \end{mathpar}
    \caption{Equivalent combinators for arrow notation, c.f.~\Cref{fig:dotypetheory}.}
    \label{fig:var:arrowNotation}
  \end{figure}
\end{remark}

\noindent
Finally, we can \kl{rewire}: i.e.,~substitute the context from which variables are picked.
Given a term $Γ ⊢ t ፡ Δ$ in context $Γ = x_1 : X_1,...,x_{\lenx} : X_{\lenx}$
and a function between finite sets $φ ፡ \fsetx → \fsetu$, we can derive a term
$Γ' ⊢ (φ ∗ t) ፡ Δ$ in context $Γ' = u_1 : U_1, ..., u_{\lenu} : U_{\lenu}$
defined inductively by
\begin{enumerate}
  \item $φ ∗ (\return(x_{α(1)},...,x_{α(\leny)})) = \return(u_{φ(α(1))},...,u_{φ(α(\leny))})$;
  \item $φ ∗ (y_1,...,y_{\leny} ← f(x_{α(1)},...,x_{α(\lenk)}) 𑊩 t) = (y_1,...,y_{\leny} ← f(u_{φ(α(1))},...,u_{φ(α(\lenk))}) 𑊩 (φ + \id[\leny]) ∗ t)$;
\end{enumerate}

\begin{definition}[Rewiring]
  \defining{linkAst}%
  \AP The \intro{rewiring} of an \kl{arrow notation preterm}, $t ፡ X₁,...,X_{\lenx} →
  Y₁,...,Y_{\leny}$ by a finite function, $φ ፡ \fsetx → \fsetu$, is the preterm, $(φ ∗ t) ፡
  X_{φ(1)},...,X_{φ(\lenx)} → Y_1,...,Y_\leny$, defined inductively by
  $φ ∗ \vret(α) = \vret(φ ∘ α)$ and $φ ∗ (f(α) 𑊩 t) = f(φ ∘ α) 𑊩 ((φ + \id[\leny]) ∗ t)$.
\end{definition}

\begin{definition}[Interchange axiom]
  Any two statements whose input variables are disjoint from the output
  variables of the other can interchange. That is, the \kl{interchange axiom} is
  the minimal congruence equating the following pair of terms, whenever $y_i ≠
  x_j$ and $z_i ≠ x_j$ for each pair of indices $i,j$.
  $$\begin{aligned}
  & (y_1,...,y_{\leny} ← f(x_{α(1)},...,x_{α(q)}) 𑊩 z_1,...,z_{\lenz} ← g(x_{β(1)},...,x_{β(r)}) 𑊩 t) & = \\
  & (z_1,...,z_{\lenz} ← g(x_{β(1)},...,x_{β(r)}) 𑊩 y_1,...,y_{\leny} ← f(x_{α(1)},...,x_{α(q)}) 𑊩 t).
  \end{aligned}$$
  It is important to state explicitly that, while the variables in both
  instances of the expression $t$ are the same, the position from which they are
  fetched from the context changes: $t$ appears under two different contexts, depending
  on the order in which the output variables of $f$ and $g$ appear.
  If we write $Γ = x_1 : X_1,... , x_{\lenx} : X_{\lenx}$, the two different contexts
  are
  $$Γ,y_1 : Y_1,...y_{\leny} : Y_{\leny},z_1 : Z_1,...,z_{\lenz} : Z_{\lenz} ≠ 
  Γ,z_1 : Z_1,...,z_{\lenz} : Z_{\lenz},y_1 : Y_1,...y_{\leny} : Y_{\leny}.$$
  Formally, the second appearance of $t$ must swap the variables it uses — via
  the swap function $σ_{\leny,\lenz} ፡ \fsety + \fsetz → \fsetz +
  \fsety$ — and so it is really $(\id[\lenx] + σ_{\leny,\lenz}) ∗ t$.
  Both generators must also fetch their variables from a context that includes them —
  via inclusion functions $\vlinc{\leny} ፡ \fsetx → \fsetx + \fsety$ and
  $\vlinc{\lenz} ፡ \fsetx → \fsetx + \fsetz$ — with the outputs of the
  previous generator. The \kl{interchange axiom}, under combinator encoding, is
  the following equality
  $$\begin{aligned}
    f(α) 𑊩 g(\vlinc{\leny} ∘ {β}) 𑊩 t &\ =\  %
    g(β) 𑊩 f(\vlinc{\lenz} ∘ {α}) 𑊩 (\id[\lenx] + σ_{\leny,\lenz}) ∗ t, \\
  \end{aligned}$$
  for any pair of generators $f ∈ Σ(X_{α(1)},...,X_{α(q)};
  Y_1,...,Y_{\leny})$ and $g ∈ Σ(X_{β(1)},...,X_{β(r)};
  Z_1,...,Z_{\lenz})$.

  The \kl{interchange axiom} generates an associated equivalence relation on
  \kl{arrow notation preterms}: quotienting by the minimal \kl{congruence} that
  contains the \kl{interchange axiom}, we obtain \emph{terms}.
\end{definition}

\begin{definition}[Arrow notation terms]
  An \intro{arrow notation term} is an equivalence class of \kl{arrow notation preterms}
  under the \kl{interchange axiom}.
\end{definition}

\begin{toappendix}
We have presented \arrowNotation{} as a simple type theory. Terms of this type
theory make sense only when considered in a context; terms over a context
correspond exactly to derivations on the theory. These properties simplify the
management of variables: $α$-equivalence works as usual, and we only need to
define substitution with non-fresh variables. 
    
\begin{definition}[Substitution]
  \label{def:substitution-do-notation}
  ""Substitution"" of the variable $x : X$ by the variable $u : X$ on a term
  $Γ₁,x : X, u : X, Γ₂ ⊢ t : Δ$ is the term $Γ₁, x : X, u : X, Γ₂ ⊢
  \subst{t}{x}{u} : Δ,$ defined inductively by: "substituting" every variable on
  a return statement,
    $$\subst{(\return(x₁, …, xₙ))}{x}{u} =
    \return(\subst{x_1}{x}{u}, …, \subst{x_n}{x}{u});$$
  and "substituting" every variable to the right of a generator statement,  
    \begin{align*} 
      \subst{(y₁, …,yₘ ← f(x_1, …, x_n) ⨾ t)}{x}{u} =
      y₁, …,yₘ ← f(\subst{x₁}{x}{u}, …, \subst{xₙ}{x}{u}) ⨾ \subst{t}{x}{u};
    \end{align*}
  where we use $\subst{y}{x}{u}$ to mean $u$ if $y = x$, and $y$ otherwise.
\end{definition}
\end{toappendix}

\begin{toappendix}
    \begin{propositionrep}\label{prop:substitution-welldef}
      "Substitution" is well-defined under the "interchange axiom".
    \end{propositionrep}
    \begin{proof}
        Let two terms be related by the "interchange axiom".
        \begin{align*}
          \left.\begin{array}{l}
          \vec{u} ← f(\vec{x}) \\
          \vec{v} ← g(\vec{y}) \\
          \cont
        \end{array}\right| & ≈
        \left.\begin{array}{l}
          \vec{v} ← g(\vec{y}) \\
          \vec{u} ← f(\vec{x}) \\    
          \cont
        \end{array}\right|;
        \end{align*}
        
        We must have $xᵢ ≠ vⱼ$ and $yᵢ ≠ uⱼ$. Then, using that $x$ and $u$ are in
        the input context while $\vec{u}$ and $\vec{v}$ appear in the output of a
        generator statement, we know they must be different variables. As a
        consequence, $\subst{xᵢ}{x}{u} ≠ vⱼ$ and $\subst{yᵢ}{x}{u} ≠ uⱼ$, and the
        following two terms are related by the "interchange axiom".
        \begin{align*}
          \left.\begin{array}{l}
          \vec{u} ← f(\subst{\vec{x}}{x}{u}) \\
          \vec{v} ← g(\subst{\vec{y}}{x}{u}) \\
          \cont
        \end{array}\right| & ≈
        \left.\begin{array}{l}
          \vec{v} ← g(\subst{\vec{y}}{x}{u}) \\
          \vec{u} ← f(\subst{\vec{x}}{x}{u}) \\    
          \cont
        \end{array}\right|;
        \end{align*}
        This concludes the proof.
      \end{proof}
    \end{toappendix}
    
    \begin{toappendix}
    \begin{definition}[Congruence]
      A relation on \kl{arrow notation preterms}, $(≈)$, is a \intro{congruence}
      when it relates only terms over the same context and
      \begin{itemize}
        \item it is reflexive on return statements, meaning
        $\return(x₁,…,xₙ) ≈ \return(x₁,…,xₙ);$
        \item and it is preserved by statements, meaning that $t ≈ t'$ implies
        $$(y₁,…,yₘ ← f(x₁,…,xₙ) 𑊩 t) ≈ (y₁,…,yₘ ← f(x₁,…,xₙ) 𑊩 t').$$
      \end{itemize}
    \end{definition}
    \end{toappendix}
     
\begin{toappendix}
\subsection{Algebra of arrow notation}
\label{sec:algebra-arrow-notationa}

Operations and reasoning on \arrowNotationTerms{} will benefit from a compact
notation. As humans, variables make a formal language easier to read; however,
for formal reasoning and computer implementation, variable handling quickly
becomes cumbersome. This is why it can be convenient to keep two versions of the
same syntax: one with variables and one based on combinators. The situation is
similar to that of the variable-based lambda calculus versus de Bruijn syntax.

In this part of the appendix, we work with the combinator version of
\arrowNotation{} (\Cref{fig:var:arrowNotation}), and we describe the algebra of
these combinators: these are all results we will need for the proof of freeness
(\Cref{th:doNotationInternalLanguage}).

\begin{remark}
  From this point on, we write lists of types in bold, $\vec{x} = X₁,...,X_x$;
  the length of a list, $\vec{x}$, will usually be denoted by the same
  roman letter, $x$.
\end{remark}

\begin{definition}[Finite function combinators]
  \defining{linkCdot}%
  Let us fix notation for some operations on functions between sets of finite
  cardinality.
  \begin{enumerate}
    \item Given any two numbers, the ""symmetry"" $\vsym{x}{y} ፡ x + y → y + x$
    is defined by $\vsym{x}{y}(i) = i - x$ when $i ≤ y$, and $\vsym{x}{y}(i) = i
    - y$ when $i > y$.
    \item Given any two numbers, the ""left inclusion"", $\vinc{x}{k} ፡ x → x + k$, is
    defined by $\vinc{x}{k}(i) = i$; while the ""right inclusion"", 
    $\vrnc{x}{k} ፡ x → k + x$, is defined by
    $\vrnc{x}{k}(i) = i + x$.
    \item Given any function $α ፡ x → y$, its ""left whiskering"" by a number
    $k$ is a function $k ⋉ α ፡ k + x → k + y$ defined by $(k ⋉ α)(i) = i$ when
    $i ≤ k$, and $(k ⋉ α)(i) = α(i - k)$ when $i > k$; its ""right whiskering""
    by a number $k$ is a function $α ⋊ k ፡ x + k → y + k$ defined by $(α ⋊ k)(i)
    = α(i)$ when $i ≤ x$, and $(α ⋊ k)(i) = i - x + y$ when $i > x$.
    \item The identity function $\id_x ፡ x → x$ is defined by $\id_x(i) = i$.
    \item Given two functions, $α ፡ x → y$ and $β ፡ y → z$, their composition is
    the function $β ∘ α ፡ x → z$ defined as $(β ∘ α)(i) = β(α(i))$.
    \item When needed, we encode functions between sets of finite cardinality as lists
    of integers. The generic function is written as follows,
    $$α = [α_1,...,α_m] ፡ m → n, \mbox{ where }α_i ∈ \{1,...,n\}\mbox{ for } i ∈
    \{1,...,m\}.$$
  \end{enumerate}
\end{definition}

\subsection{Rewiring}
\begin{proposition}[Rewiring is an action]
    \AP\phantomintro{rewiring is an action} "Rewiring" defines an action of finite
    function composition, in the sense that $\id_x ∗ t = t$ and $(φ ∘ ψ) ∗ t = φ ∗
    ψ ∗ t$.
\end{proposition}
\begin{proof}
    Follows directly from the definition and manipulation of finite functions.
\end{proof}
   \subsection{Composition}

  \begin{proposition}[Rewiring preserves composition]
    \defining{linkRewiringPreservesComposition}
    "Rewiring" preserves "composition", 
    $$φ ∗ (s ⨾ t) = (φ ∗ s) ⨾ t.$$
  \end{proposition}
  \begin{proof}
    We proceed by structural induction over $s$. Let it be a return statement,
    $s = \vret(α)$. We reason by \emph{(i)} the "definition of rewiring";
    \emph{(ii)} the fact that "rewiring is an action"; and \emph{(iii,iv)} the
    "definition of rewiring".
    $$
      φ ∗ (\vret(α) ⨾ t) \overset{(i)}{=}
      φ ∗ α ∗ t \overset{(ii)}{=}
      (φ ∘ α) ∗ t \overset{(iii)}{=}
      \vret(φ ∘ α) ⨾ t \overset{(iv)}{=}
      (φ ∗ \vret(α)) ⨾ t.
    $$
    Let it be a generator statement. We reason by \emph{(i)} the
    "definition of rewiring"; \emph{(ii)} the inductive hypothesis; and
    \emph{(iii)} the "definition of rewiring".
    \begin{align*}
      & φ ∗ (f(γ) 𑊩 s ⨾ t)
      & \ \smasheq{\overset{\emph{(i)}}{=}}\ \\
      & f(φ ∘ γ) 𑊩 ((φ ⋊ y) ∗ (s ⨾ t))
      & \ \smasheq{\overset{\emph{(ii)}}{=}}\ \\
      & f(φ ∘ γ) 𑊩 ((φ ⋊ y) ∗ s) ⨾ t 
      & \ \smasheq{\overset{\emph{(iii)}}{=}}\ \\
      & (φ ∗ (f(γ) 𑊩 s)) ⨾ t.
    \end{align*}
    These two cases complete the proof.
  \end{proof}
  
  \begin{lemma}[Term composition is associative]
    \label{lemma:var:compositionAssociative}%
    \AP\phantomintro{Term composition is associative}
    \kl{Composition} is associative, $(s ⨾ t) ⨾ r = s ⨾ (t ⨾ r)$. 
  \end{lemma}
  \begin{proof}
    Let us proceed by induction on the first term. Let it be a return statement,
    $s = \vret(α)$; we use \emph{(i)} the definition of "composition",
    \emph{(ii)} that "rewiring is an action", and \emph{(iii)} the
    definition of "composition".
    \begin{align*}
      (\vret(α) ⨾ t) ⨾ r \overset{\emph{(i)}}{=} 
      (α ∗ t) ⨾ r \overset{\emph{(ii)}}{=}
      α ∗ (t ⨾ r) \overset{\emph{(iii)}}{=}
      \vret(α) ∗ (t ⨾ r).
    \end{align*}
    Let it be a generator statement. We use \emph{(i,ii,iv)} the definition of
    "composition"; and \emph{(iii)} the inductive hypothesis.
    \begin{align*}
      & ((f(γ) 𑊩 s) ⨾ t) ⨾ r
      & \ \smasheq{\overset{\emph{(i)}}{=}}\ \\
      & (f(γ) 𑊩 (s ⨾ t)) ⨾ r
      & \ \smasheq{\overset{\emph{(ii)}}{=}}\ \\
      & f(γ) 𑊩 ((s ⨾ t) ⨾ r)
      & \ \smasheq{\overset{\emph{(iii)}}{=}}\ \\
      & f(γ) 𑊩 (s ⨾ (t ⨾ r))
      & \ \smasheq{\overset{\emph{(iv)}}{=}}\ \\
      & (f(γ) 𑊩 s) ⨾ (t ⨾ r).
    \end{align*}
    These two cases conclude the proof.
  \end{proof}
  
  \begin{lemma}[Term composition is unital]\label{lemma:var:compositionUnital}%
    \AP\kl{Term composition is unital} with returning the identity function,
    $\vret(\id[x]) ፡ \vec{x} → \vec{x}$. That is, for each term $s ፡ \vec{x} →
    \vec{y}$, we have $$\vret(\id[x]) ⨾ s = s = s ⨾ \vret(\id[y]).$$
  \end{lemma}
  \begin{proof}
    Left unitality holds using \emph{(i)} the definition of \kl{term composition},
    and \emph{(ii)} that \kl{rewiring} is an action.
    $$\vret(\id[x]) ⨾ s \overset{\emph{(i)}}{=} \id[x] ∗ s
    \overset{\emph{(ii)}}{=} s.$$
    For the second case, we proceed by induction on $s$. Let it be a return
    statement, $\vret(α) ⨾ \vret(\id[y]) = \vret(α \circ \id[y]) = \vret(α)$. Let
    it be a generator statement, $(f(γ) 𑊩 s) ⨾ \vret(y) = f(γ) 𑊩 (s ⨾ \vret(y))
    = f(γ) 𑊩 s$.
  \end{proof} 
  
  \begin{proposition}[Category of arrow notation terms]
  \label{prop:var:rawterm-category}\label{prop:rawterm-category}%
  \kl{Arrow notation terms} over a \kl{signature} \(Σ\) form a category with
  \kl{composition}, $\prearrow(Σ)$.
  \end{proposition}
  \begin{proof}
    Follows from \Cref{lemma:var:compositionUnital,lemma:var:compositionAssociative}.
  \end{proof}
   \subsection{Whiskering}

\begin{definition}[Whiskering]
    \label{def:var:whiskering}%
    \label{def:whiskering}%
    \defining{linkWhiskering}%
    \AP The \emph{left \intro{whiskering}} of a term $t : \vec{x} → \vec{y}$ by a
    list of types $\vec{z} = Z₁,..,Z_z$ is the term $(\vec{z} ⋊ t) ፡
    \vec{z},\vec{x} → \vec{z},\vec{y},$ inductively defined as follows.
    \begin{itemize}
      \item $\vec{z} ⋉ \vret(α) = \vret(z ⋉ α)$.
      \item $\vec{z} ⋉ (f(γ) 𑊩 s) = f(\vrinc{k} ∘ γ) 𑊩 (\vec{z} ⋉ s)$.
    \end{itemize}
  
    The \emph{right whiskering} of a term $t ፡ \vec{x} → \vec{y}$ by a list of
    types $\vec{z} = Z₁,..,Z_z$ is the term $(t ⋊ \vec{z}) ፡ \vec{x},\vec{z} →
    \vec{y},\vec{z},$ inductively defined as follows.
    \begin{itemize}
      \item $\vret(α) ⋊ \vec{z} = \vret(α ⋊ z)$.
      \item $(f(γ) 𑊩 s) ⋊ \vec{z} = f(\vlinc{z} ∘ γ) 𑊩 \vret(x₂ ⋉ \vsym{z}{m}) ⨾ (s ⋊ \vec{z})$.
    \end{itemize}
    Note how the last case requires us to reposition the outputs of the generator.
  \end{definition}
  
  \begin{proposition}[Rewiring preserves whiskering]%
    \label{prop:whiskeringpreservesrewiring}%
    \defining{linkWhiskeringPreservesRewiring}%
    \AP\phantomintro{whiskering preserves rewiring}
    $$(k ⋊ α) ∗ (\vec{k} ⋊ s) = \vec{k} ⋊ (α ∗ s).$$
  \end{proposition}
  \begin{proof}
    We proceed by induction on the term $s$. Let it be a return statement, $s =
    \vret(γ)$. We use \emph{(i)} the definition of \kl{whiskering}, \emph{(ii)}
    the \kl{definition of rewiring}, \emph{(iii)} finite function composition,
    \emph{(iv)} the definition of \kl{whiskering}, and \emph{(v)} the
    \kl{definition of rewiring}.
    \begin{align*}
      & (k ⋊ α) ∗ (\vec{k} ⋊ \vret(γ))
      & \smasheq{\overset{\emph{(i)}}{=}}\\
      & (k ⋊ α) ∗ \vret(k ⋊ γ)
      & \smasheq{\overset{\emph{(ii)}}{=}}\\
      & \vret((k ⋊ α) ∘ (k ⋊ γ))
      & \smasheq{\overset{\emph{(iii)}}{=}}\\
      & \vret(k ⋊ (α ∘ γ))
      & \smasheq{\overset{\emph{(iv)}}{=}}\\
      & \vec{k} ⋊ \vret(α ∘ γ)
      & \smasheq{\overset{\emph{(v)}}{=}}\\
      & \vec{k} ⋊ (α ∗ \vret(γ)).
    \end{align*}
  \end{proof}
  
  \begin{proposition}[Whiskering is functorial]%
    \label{prop:whiskeringFunctorial}%
    $$(\vec{k} ⋊ s) ⨾ (\vec{k} ⋊ t) = \vec{k} ⋊ (s ⨾ t).$$
  \end{proposition}
  \begin{proof}
    We proceed by induction on the first term, $s$. Let it be a return statement,
    $s = \vret(α)$. We use \emph{(i)} the definition of "whiskering"; \emph{(ii)}
    the "definition of rewiring"; \emph{(iii)} that "whiskering preserves rewiring";
    and \emph{(iv)}  the "definition of rewiring".
    \begin{align*}
      & (\vec{k} ⋊ \vret(α)) ⨾ (\vec{k} ⋊ t) 
      & \smasheq{\overset{\emph{(i)}}{=}}\\
      & \vret(k ⋊ α) ⨾ (\vec{k} ⋊ t)
      & \smasheq{\overset{\emph{(ii)}}{=}}\\
      & (k ⋊ α) ∗ (\vec{k} ⋊ t) 
      & \smasheq{\overset{\emph{(iii)}}{=}}\\
      & \vec{k} ⋊ (α ∗ t)
      & \smasheq{\overset{\emph{(iv)}}{=}}\\
      & \vec{k} ⋊ (\vret(α) ⨾ t).
    \end{align*}  
    Let it be a generator statement, $s = f(γ) 𑊩 s'$. We use \emph{(i)} the
    definition of "whiskering", \emph{(ii)} the inductive hypothesis, and
    \emph{(iii)} the definition of "whiskering".
    \begin{align*}
      & (\vec{k} ⋊ (f(γ) 𑊩 s)) ⨾ (\vec{k} ⋊ t)
      & \smasheq{\overset{\emph{(i)}}{=}}\\
      & f(\vrinc{k} ∘ γ) 𑊩 (\vec{k} ⋊ s) ⨾ (\vec{k} ⋊ t)
      & \smasheq{\overset{\emph{(ii)}}{=}}\\
      & f(\vrinc{k} ∘ γ) 𑊩 (\vec{k} ⋊  (s ⨾ t))
      & \smasheq{\overset{\emph{(iii)}}{=}}\\
      & \vec{k} ⋊ (f(γ) 𑊩 s ⨾ t).
    \end{align*}
    These two cases conclude the proof.
  \end{proof}
  
  \begin{proposition}
    The category of "arrow notation terms", $\preArrow(Σ)$, is a \premonoidalCategory{}.
  \end{proposition}
  \begin{proof}
    This follows from \Cref{lemma:var:compositionUnital,lemma:var:compositionAssociative}.
  \end{proof}

\subsection{Interchange law}

\begin{proposition}[Rewiring preserves interchange]
  \defining{linkRewiringPreservesInterchange}
  \AP\kl{Rewiring} preserves the \kl{interchange axiom}. Let $s ≈ t$, then $φ ∗
  s ≈ φ ∗ t$.
\end{proposition}
\begin{proof}
  Let two terms be related by the \kl{interchange axiom}, $s ≈ t$.
  We reason by \emph{(i,ii)} the "definition of rewiring"; \emph{(iii)} how
  "whiskering interacts with inclusion"; \emph{(iv)} the \kl{interchange axiom};
  \emph{(v)} how whiskering interacts with symmetries; \emph{(vi)} how
  "whiskering interacts with inclusion"; and \emph{(vii,viii)} the
  \kl{definition of rewiring}.
  \begin{align*}
    & φ ∗ f(α) 𑊩 g(\vcomp{\vlinc{y}}{β}) 𑊩 r
    & \ \smasheq{\overset{\emph{(i)}}{=}}\ \\
    & f(\vcomp{φ}{α}) 𑊩 (φ ⋊ y) ∗ g(\vcomp{\vlinc{y}}{β}) 𑊩 r
    & \ \smasheq{\overset{\emph{(ii)}}{=}}\ \\
    & f(φ ∘ α) 𑊩 g((φ ⋊ y) ∘ \vlinc{y} ∘ β) 𑊩 (φ ⋊ y ⋊ z) ∗ r
    & \ \smasheq{\overset{\emph{(iii)}}{=}}\ \\
    & f(φ ∘ α) 𑊩 g(\vlinc{y} ∘ φ ∘ β) 𑊩 (φ ⋊ y ⋊ z) ∗ r
    & \ \smasheq{\overset{\emph{(iv)}}{=}}\ \\
    & g(φ ∘ β) 𑊩 f(\vlinc{z} ∘ φ ∘ α) 𑊩 (x ⋊ \vsym{y}{z}) ∗ (φ ⋊ y ⋊ z) ∗ r
    & \ \smasheq{\overset{\emph{(v)}}{=}}\ \\
    & g(φ ∘ β) 𑊩 f(\vlinc{z} ∘ φ ∘ α) 𑊩 (φ ⋊ z ⋊ y) ∗ r
    & \ \smasheq{\overset{\emph{(vi)}}{=}}\ \\
    & g(φ ∘ β) 𑊩 f((φ ⋊ z) ∘ \vlinc{z} ∘ α) 𑊩 (φ ⋊ z ⋊ y) ∗ r
    & \ \smasheq{\overset{\emph{(vii)}}{=}}\ \\
    & g(φ ∘ β) 𑊩 (φ ⋊ z) ∗ f(\vlinc{z} ∘ α) 𑊩 r
    & \ \smasheq{\overset{\emph{(viii)}}{=}}\ \\
    & φ ∗ g(β) 𑊩 f(\vlinc{z} ∘ α) 𑊩 r.
  \end{align*}
We have shown the two terms are related by the "interchange axiom".
\end{proof}

\begin{proposition}[Composition preserves interchange]
  \label{prop:composition-welldef}%
  \kl{Term composition} is well-defined under the \kl{interchange axiom}:
  if $s₁ ≈ s₂$ and $t₁ ≈ t₂$, then $s₁ ⨾ s₂ ≈ t₁ ⨾ t₂$.
\end{proposition} 
\begin{proof}
We check separately that quotienting by the \kl{interchange axiom} $(≈)$
preserves pre-composition and post-composition. Let us start with
pre-composition: assume two terms are related by the
\kl{interchange axiom}, $s₁ ≈ s₂$. These must be of the form 
$$s₁ = f(α) 𑊩 g(\vinc{n}{m} ∘ β) 𑊩 s ≈ g(β) 𑊩 f(\vinc{n}{k} ∘ α) 𑊩 (n ⊗
σ_{k,m} ∗ s) = s₂.$$%
We will prove now that $s₁ ⨾ t ≈ s₂ ⨾ t$. We use \emph{(i)} the
\kl{interchange axiom}; and \emph{(ii)} that \kl{rewiring preserves composition}.
\begin{align*}
  & f(α) 𑊩 g(\vinc{n}{m} ∘ β) 𑊩 s ⨾ t
  & \ \smasheq{\overset{\emph{(i)}}{≈}}\ \\
  & g(β) 𑊩 f(\vinc{n}{k} ∘ α) 𑊩 ((n ⊗ σ_{k,m}) ∗ (s ⨾ t))
  & \ \smasheq{\overset{\emph{(ii)}}{=}}\ \\
  & g(β) 𑊩 f(\vinc{n}{k} ∘ α) 𑊩 (n ⊗ σ_{k,m} ∗ s) ⨾ t. 
\end{align*}
We will prove now that post-composition also preserves the \kl{interchange
axiom}: for any two related terms, $t₁ ≈ t₂$, we will prove that $s ⨾ t₁ ≈ s ⨾
t₂$. We proceed by induction over $s$, the first term of the composition. Let it
be a return statement, $\vret(α)$. We use \emph{(i)} the \kl{definition of
rewiring}; that \emph{(ii)} "rewiring preserves interchange"; and \emph{(iii)}
the \kl{definition of rewiring}.
\[
  \vret(α) ⨾ t₁ \overset{\emph{(i)}}{=}
  α ∗ t₁ \overset{\emph{(ii)}}{=}
  α ∗ t₂ \overset{\emph{(iii)}}{=}
  \vret(α) ⨾ t₂.
\]
Let it be a generator statement, $s = f(γ) 𑊩 s'$. We use the inductive
hypothesis to obtain $f(γ) 𑊩 s' ⨾ t₁ = f(γ) 𑊩 s' ⨾ t₂$. This
concludes the proof.
\end{proof}

\begin{proposition}[Whiskering preserves interchange]
  \kl{Whiskering} is well-defined under the \kl{interchange axiom}:
  if $s₁ ≈ s₂$, then $\vec{k} ⋊ s₁ ≈ \vec{k} ⋊ s₂$.
\end{proposition}
\begin{proof}
  Let two terms be related by the "interchange axiom". These must be of the form
  $s₁ = \vec{k} ⋊ (f(α) 𑊩 g(\vinc{n}{n} ∘ β) 𑊩 s)$ and $s₂ = \vec{k} ⋊ (g(β)
  𑊩 f(\vrnc{n}{k} ∘ α) 𑊩 s)$.
  \begin{align*}
    & \vec{k} ⋊ (f(α) 𑊩 g(\vinc{n}{n} ∘ β) 𑊩 s) = \\
    & f(\vrnc{n}{k} ∘ α) 𑊩 g(\vrnc{n + m}{k} ∘ \vinc{n}{n} ∘ β) 𑊩 (\vec{k} ⋊ s) = \\
    & f(\vrnc{n}{k} ∘ α) 𑊩 g(\vinc{n}{n} ∘ \vrnc{n + m}{k} ∘ β) 𑊩 (\vec{k} ⋊ s) ≈ \\
    & g(\vrnc{n}{k} ∘ β) 𑊩 f(\vinc{n}{m} ∘ \vrnc{n}{k} ∘ α) 𑊩 (\vec{k} ⋊ s) = \\
    & \vec{k} ⋊ (g(β) 𑊩 f(\vrnc{n}{k} ∘ α) 𑊩 s).
  \end{align*}
\end{proof}

\begin{proposition}
  \label{prop:var:term-category}%
  \label{prop:term-category}%
  Terms of \arrowNotation{} over a signature \(Σ\) quotiented by the
  \axiomsOfArrowNotation{} form a category $\arrowF(Σ)$.
\end{proposition}
\begin{proof}
  This follows from \Cref{lemma:var:compositionUnital,lemma:var:compositionAssociative}.
\end{proof}

\subsection{Monoidal structure}

\begin{lemma}[Interchange law]
  \label{lemma:interchangeLaw}%
  For any two terms, $t₁ ፡ \vec{x_1} → \vec{y_1}$ and $t_2 ፡
  \vec{x_2} → \vec{y_2}$.
  $$(t₁ ⋊ \vec{x₂}) ⨾ (\vec{y₁} ⋉ t_2) = (\vec{x₁} ⋉ t_2) ⨾ (t_1 ⋊ \vec{y_2}).$$
\end{lemma}
\begin{proof}
  Let us proceed by induction over the term $t_1$. Let it be a return statement,
  $t_1 = \vret(α)$; we employ \emph{(i)} the definition of \kl{whiskering},
  \emph{(ii)} \Cref{prop:returnInterchanges}, and \emph{(iii)} the definition of
  \kl{whiskering}.
  \begin{align*}
    & (\vret(α) ⋊ \vec{x₂}) ⨾ (\vec{y₁} ⋉ t₂) & \smash{\overset{(i)}{=}} \\
    & \vret(α ⋉ x₂) ⨾ (\vec{y₁} ⋉ t₂) & \smash{\overset{(ii)}{=}} \\
    & (\vec{x₁} ⋉ t₂) ⨾ \vret(α ⋉ y₂) & \smash{\overset{(iii)}{=}} \\
    & (\vec{x₁} ⋉ t₂) ⨾ (\vret(α) ⋉ \vec{y₂}).
  \end{align*}

  \noindent Let it be a generator statement, $t₁ = f(γ) 𑊩 s₁$, of output type
  $\vec{m}$; we employ \emph{(i)} the definition of \kl{whiskering} and
  \kl{composition}, \emph{(ii)} the inductive hypothesis, \emph{(iii)}
  \Cref{prop:generatorsInterchange}, \emph{(iv)} the definition of \kl{whiskering}
  and \kl{composition}.
  \begin{align*}
    & ((f(γ) 𑊩 s₁) ⋊ \vec{x₂}) ⨾ (\vec{y₁} ⋉ t₂) & \smash{\overset{(i)}{=}} \\
    & f(γ · \vlinc{x₂}) 𑊩 \vret(x₁ ⋉ \vsym{x₂}{m}) ⨾ (s₁ ⋊ \vec{x₂}) ⨾ (\vec{y₁} ⋉ t₂) & \smash{\overset{(ii)}{=}} \\
    & f(γ · \vlinc{x₂}) 𑊩 \vret(x₁ ⋉ \vsym{x₂}{m}) ⨾ (\vec{x₁} ⋉ \vec{m} ⋉ t₂) ⨾ (s₁ ⋊ \vec{y₂}) & \smash{\overset{(iii)}{=}} \\
    & (\vec{x₁} ⋉ t₂) ⨾ f(γ · \vlinc{y₂}) 𑊩 \vret(x₁ ⋉ \vsym{y₂}{m})  ⨾ (s₁ ⋊ \vec{y₂}) & \smash{\overset{(iv)}{=}} \\
    & (\vec{x₁} ⋉ t₂) ⨾ ((f(γ) 𑊩 s₁) ⋊ \vec{y₂}). & \qedhere
  \end{align*}
\end{proof}

\begin{proposition}[Return interchanges]
  \label{prop:returnInterchanges}%
  \defining{linkReturnInterchanges}%
  $$\vret(α ⋊ x₂) ⨾ (\vec{y₁} ⋉ t) = 
  (\vec{x₁} ⋉ t) ⨾ \vret(α ⋉ y₂).$$
\end{proposition}
\begin{proof}
  We proceed by induction on the term $t₂$. Let it be a return statement, $t =
  \vret(β)$. We use \emph{(i,iii)} the definition of "term composition", and \emph{(ii)}
  interchange of finite functions.
  \begin{align*}
    & \vret(α ⋊ x₂) ⨾ \vret(y₁ ⋉ β) 
    & \ \smasheq{\overset{\emph{(i)}}{=}}\ \\
    & \vret((α ⋊ x₂) ⨾ (y₁ ⋉ β))
    & \ \smasheq{\overset{\emph{(ii)}}{=}}\ \\
    & \vret((x₁ ⋉ β) ⨾ (α ⋊ y₂))
    & \ \smasheq{\overset{\emph{(iii)}}{=}}\ \\
    & \vret(x₁ ⋉ β) ⨾ \vret(α ⋊ y₂).
  \end{align*}
  Let it be a generator statement, $t = f(γ) 𑊩 t'$, of output type $\vec{m}$.
  \begin{flalign*}
    & \vret(α ⋊ x₂) ⨾ (\vec{y₁} ⋉ (f(γ) 𑊩 t')) = \\
    & \vret(α ⋊ x₂) ⨾ f(γ · \vrinc{y₁}) 𑊩 (\vec{y₁} ⋉ t') = \\
    & f(γ · \vrinc{y₁} ⨾ (α ⋊ x₂)) 𑊩 \vret(α ⋊ x₂ ⋊ m) ⨾ (\vec{y₁} ⋉ t') = \\
    & f(γ · \vrinc{x₁}) 𑊩 \vret(α ⋊ x₂ ⋊ m) ⨾ (\vec{y₁} ⋉ t') = \\
    & f(γ · \vrinc{x₁}) 𑊩 (\vec{y₁} ⋉ t') ⨾ \vret(α ⋊ y₂) = \\
    & (\vec{y₁} ⋉ (f(γ) 𑊩 t')) ⨾ \vret(α ⋊ y₂).
  \end{flalign*}
  These two cases complete the proof by induction.
\end{proof}

\begin{proposition}[Generators interchange]%
  \label{prop:generatorsInterchange}%
  \begin{flalign*}
    & f(γ · \vlinc{x₂}) 𑊩 \vret(x₁ ⋉ \vsym{x₂}{m}) ⨾ (\vec{x₁} ⋉ \vec{m} ⋉ t₂) ⨾ r = \\
    & (\vec{x₁} ⋉ t₂) ⨾ f(γ · \vlinc{y₂}) 𑊩 \vret(x₁ ⋉ \vsym{y₂}{m}) ⨾ r.
  \end{flalign*}
  \begin{proof}
    We proceed by induction on $t₂$. Let it be a return statement, $t₂ = \vret(α)$.
    \begin{flalign*}
      & f(γ ⨾ \vlinc{x₂}) 𑊩 \vret(x₁ ⋉ \vsym{x₂}{m}) ⨾ (\vec{x₁} ⋉ \vec{m} ⋉ \vret(α)) ⨾ r = \\
      & f(γ ⨾ \vlinc{x₂}) 𑊩 \vret(x₁ ⋉ \vsym{x₂}{m}) ⨾ \vret(x₁ ⋉ m ⋉ α) ⨾ r = \\
      & f(γ ⨾ \vlinc{x₂}) 𑊩 \vret((x₁ ⋉ \vsym{x₂}{m}) ⨾ (x₁ ⋉ m ⋉ α)) ⨾ r = \\
      & f(γ ⨾ \vlinc{x₂}) 𑊩 \vret((x₁ ⋉ α ⋊ m) ⨾ \vsym{y₂}{m}) ⨾ r = \\
      & f(γ ⨾ \vlinc{x₂} ⨾ (x₁ ⋉ α)) 𑊩 \vret(x₁ ⋉ α ⋊ m) ⨾ \vret(x₁ ⋉ \vsym{y₂}{m}) ⨾ r = \\
      & f(γ ⨾ \vlinc{x₂} ⨾ (x₁ ⋉ α)) 𑊩 ((x₁ ⋉ α ⋊ m) ∗ \vret(x₁ ⋉ \vsym{y₂}{m}) ⨾ r) = \\
      & (x₁ ⋉ α) ∗ f(γ ⨾ \vlinc{x₂}) 𑊩 \vret(x₁ ⋉ \vsym{y₂}{m}) ⨾ r = \\
      & (\vec{x₁} ⋉ \vret(α)) ⨾ f(γ ⨾ \vlinc{x₂}) 𑊩 \vret(x₁ ⋉ \vsym{y₂}{m}) ⨾ r.
    \end{flalign*}

    Let it be a generator statement, $t₂ = g(δ) 𑊩 s₂$, of output type $\vec{k}$.
    \begin{flalign*}
& f(γ · \vlinc{x₂}) 𑊩 \vret(x₁ ⋉ \vsym{x₂}{m}) ⨾ (\vec{x₁} ⋉ \vec{m} ⋉ (g(δ) 𑊩 s₂)) ⨾ r = \\
& f(γ · \vlinc{x₂}) 𑊩 \vret(x₁ ⋉ \vsym{x₂}{m}) ⨾ g(δ · \vrinc{x₁} · \vrinc{m}) 𑊩 (\vec{x₁} ⋉ \vec{m} ⋉ s₂) ⨾ r = \\
& f(γ · \vlinc{x₂}) 𑊩 g(δ · \vrinc{x₁} · \vlinc{m}) 𑊩 \vret(x₁ ⋉ \vsym{x₂}{m} ⋊ k) ⨾ (\vec{x₁} ⋉ \vec{m} ⋉ s₂) ⨾ r = \\
& g(δ · \vrinc{x₁}) 𑊩 f(γ · \vlinc{x₂} · \vlinc{k}) 𑊩  \vret(x₁ ⋉ \vsym{x₂ + k}{m}) ⨾ (\vec{x₁} ⋉ \vec{m} ⋉ s₂) ⨾ r = \\
& g(δ · \vrinc{x₁}) 𑊩 f(γ · \vlinc{x₂ + k}) 𑊩  \vret(x₁ ⋉ \vsym{x₂ + k}{m}) ⨾ (\vec{x₁} ⋉ \vec{m} ⋉ s₂) ⨾ r = \\
& g(δ · \vrinc{x₁}) 𑊩 (\vec{x₁} ⋉ s₂) ⨾ f(γ · \vlinc{y₂}) 𑊩  \vret(x₁ ⋉ \vsym{y₂ + k}{m}) ⨾  r = \\
& g(δ · \vrinc{x₁}) 𑊩 (\vec{x₁} ⋉ s₂) ⨾ f(γ · \vlinc{y₂}) 𑊩  \vret(x₁ ⋉ \vsym{y₂}{m}) ⨾  r = \\
& (\vec{x₁} ⋉ (g(δ · \vrinc{x₁}) 𑊩 s₂)) ⨾ f(γ · \vlinc{y₂}) 𑊩  \vret(x₁ ⋉ \vsym{y₂}{m}) ⨾  r.
    \end{flalign*}
  \end{proof}
\end{proposition}

\begin{definition}[Tensoring]%
  \label{def:var:tensor}%
  The ""tensoring"" of two terms, $t₁ ፡ \vec{x₁} → \vec{y₁}$ and $t₂ ፡ \vec{x₂}
  → \vec{y₂}$, is a term $t₁ ⊗ t₂ ፡ \vec{x₁},\vec{x₂} → \vec{y₁},\vec{y₂}$,
  defined as
  $$t₁ ⊗ t₂ = (t₁ ⋊ \vec{x₂}) ⨾ (\vec{y₁} ⋉ t₂) = (\vec{x₁} ⋉ t₂) ⨾ (t₁ ⋊ \vec{x₂}).$$
\end{definition}

\begin{proposition}
  \label{prop:var:term-category}%
  \label{prop:term-category}%
  \label{prop:term-smc}%
  "Arrow notation terms" over a "signature" \(Σ\)
  form a strict and symmetric monoidal category, $\arrowF(Σ)$.
\end{proposition}
\begin{proof}
  This follows from \Cref{lemma:var:compositionUnital,lemma:var:compositionAssociative}.
  This follows from \Cref{prop:whiskeringFunctorial,lemma:interchangeLaw}.
\end{proof}

\begin{proposition}[Partial Frobenius algebra]%
  \label{prop:term-partial-frobenius}%
  Every object of $\arrowF(Σ)$ has a partial Frobenius algebra structure.
\end{proposition}
\begin{proof}
  Let us write $\vec{X} = X_1,...,X_n$ and $\vec{x} = x₁,...,x_n$, with $\vec{x}
  : \vec{X}$ representing the context $x_1 : X₁, ..., x_n : X_n$. The counit and
  the comultiplication of the commutative comonoid structure are given by the
  following terms.
  \begin{align*}
    & \vec{x} : \vec{X} ⊢ \return() \\
    & \vec{x} : \vec{X} ⊢ \return(\vec{x},\vec{x})
  \end{align*}
  Proving that these are associative, unital, and commutative, is straightforward.
  For instance, in order to prove associativity, we note the following equation and its
  mirrored analogue.
  \begin{align*}
    & (\vec{x} : \vec{X} ⊢ \return(\vec{x},\vec{x})) ⨾ (\vec{X} ⋊ (\vec{x} : \vec{X} ⊢ \return(\vec{x},\vec{x}))) = \\    
    & (\vec{x} : \vec{X} ⊢ \return(\vec{x},\vec{x})) ⨾ (\vec{x₁} : \vec{X} , \vec{x} : \vec{X} ⊢ \return(\vec{x₁},\vec{x},\vec{x})) = \\   
    & (\vec{x} : \vec{X} ⊢ \return(\vec{x},\vec{x},\vec{x})).
  \end{align*}

  Our candidate multiplication requires us to introduce the shorthand
  $\observe(\vec{x} = \vec{y})$, where $\vec{x}$ and $\vec{y}$ are lists of
  variables of the same length, to mean $\observe(x₁ = y₁) 𑊩 \dots 𑊩 \observe(x_n = y_n)$.
  Under this notation, the multiplication can be defined by the following term,
  $$
  \vec{x} : \vec{X}, \vec{y} : \vec{X} ⊢ \observe(\vec{x} = \vec{y}) 𑊩 \return(\vec{x}) : \vec{X}.
  $$
  One may check that this is commutative — thanks to the symmetry axiom. Let us
  show directly that it satisfies the axiom of partial Frobenius algebras.
  \begin{flalign*}
    & (\vec{X} ⋊ (\vec{y} : \vec{X} ⊢ \return(\vec{y},\vec{y}))) ⨾ 
    ((\vec{x} : \vec{X}, \vec{y} : \vec{Y} ⊢ \observe(\vec{x} = \vec{y}) 𑊩 \return(\vec{x})) ⋉ \vec{X}) & = \\
    & (\vec{x} : \vec{X}, \vec{y} : \vec{X} ⊢ \return(\vec{x},\vec{y},\vec{y})) ⨾ 
      (\vec{x} : \vec{X}, \vec{y} : \vec{Y}, \vec{y₂} : \vec{Y} ⊢ \observe(\vec{x} = \vec{y}) 𑊩 \return(\vec{x},\vec{y₂}))
      & = \\
    & (\vec{x} : \vec{X}, \vec{y} : \vec{X} ⊢ \observe(\vec{x} = \vec{y}) 𑊩 \return(\vec{x},\vec{y})).
  \end{flalign*}
  This is enough to prove that every object has a partial Frobenius algebra. By
  construction, these partial Frobenius algebras are uniform.
\end{proof}
\end{toappendix}

\subsection{Observe-arrow notation}%
\label{sec:copyDiscardCompareArrowNotation}%

Our only addition to \arrowNotation{} will be an ``$\observe(x=y)$'' statement.
It declares that the equation $x=y$ should hold for the elements in the support
of the \subdistribution{}, at that point in the computation\footnote{In
\Cref{fig:sailorchild:solution} we use statements of the form $\observe(A
\in\ports)$, where the inhabitation $A\in\ports$ can be reformulated as equation
$\{A\} = \ports \cap \{A\}$.}. All monomials containing other elements are
removed, as in Definition~\ref{def:updating}~\eqref{def:updating:restriction}.
The ``$\observe$'' statement allows us to translate an operational description
of the probabilistic decision problem into a formal mathematical object.

\begin{definition}[Observe-arrow notation]%
  \label{fig:do-notation-axioms}%
  \defining{linkAxiomsOfArrowNotation}{}
""Observe-arrow notation"" is \arrowNotation{} endowed with an extra binary
generator, $\observe ∈ Σ(X, X; X)$, for each type of the signature, $X ∈
Σ_{obj}$. Terms are additionally quotiented by the least \kl{congruence} relation
$(≈)$ satisfying
  \begin{enumerate}
    \item commutativity, $(\observe(x₁, x₂) 𑊩 t) ≈ (\observe(x₂, x₁) 𑊩 t)$;
    \item Frobenius rule, $(\observe(x₁, x₂) 𑊩 t) ≈ (\observe(x₁, x₂) 𑊩 \subst{t}{x₁}{x₂})$;
    \item idempotency, $(\observe(x, x) 𑊩 t) ≈ t$.
  \end{enumerate}
\end{definition}

\begin{remark}
  For semantic clarity, we write $\observe(x = y)$ for $\observe(x,y)$. Given
  any nullary generator, $a ∈ Σ(;Y)$, we write $\observe(x = a)$ as a shorthand
  for $y ← a() 𑊩 \observe(x = y)$. Alternatively, we could formalise a
  separation between values and computations --- as done in the context of Freyd
  categories \cite{power1999closed} --- but this is out of the scope of this
  paper.
\end{remark}

\subsection{The categorical view}

\ArrowNotationTerms{} over a signature $Σ$ compose when they have matching
output and context types; this composition is associative and unital. In other
words, \arrowNotationTerms{} over a signature $Σ$ form a category. 
We write composition in diagrammatic order $(⨾)$: given a term
$Γ ⊢ s ፡ Δ$ with $Δ = Y_1, ..., Y_{\leny}$, and a term $Γ' ⊢ t ፡ Δ'$ with
$Γ' = y_1 : Y_1, ..., y_{\leny} : Y_{\leny}$, we can derive a term
$Γ ⊢ (s ⨾ t) ፡ Δ'$ defined by structural induction and the following two clauses:
\begin{enumerate}
  \item $(\return(x_{α(1)},...,x_{α(m)})) ⨾ t = α ∗ t$, using \kl{rewiring};
  \item $(f(x_{β(1)},...,x_{β(k)}) 𑊩 s') ⨾ t = f(x_{β(1)},...,x_{β(k)}) 𑊩 (s' ⨾ t)$.
\end{enumerate}

\begin{definition}[Composition]
  \label{def:term-composition}
  \defining{linkTermComposition}
  The \emph{\intro{composition}} of an \arrowNotationTerm{} $s ፡ X_1,...,X_n →
  Y_1,...,Y_m$ and an \arrowNotationTerm{} $t ፡ Y_1,...,Y_m → Z_1,...,Z_p$ —
  with matching output and context types — is a term, $(s ⨾ t) ፡ \vec{x} →
  \vec{z}$, inductively defined by $\vret(α) ⨾ t = α ∗ t$ and $(f(γ) 𑊩 s) ⨾ t =
  f(γ) 𑊩 (s ⨾ t)$. 
\end{definition}

It is a direct proof to show that \kl{composition} is well-defined under
the \interchangeAxiom{} and that is, moreover, unital and associative. Because of this
definition, we can omit parentheses when nesting "composition" $(⨾)$ and
generators $(𑊩)$. 

"Observe-arrow notation terms" form, moreover, a \copyDiscardCompareCategory{}
--- in fact, the free one over a \kl{signature}: they form a sound and complete
language for \kl{copy-discard-compare categories}.

\begin{definition}[Copy-discard-compare category]
  \label{def:copyDiscardCompareCategory}%
  \defining{linkCopyDiscardCompareCategory}%
  \defining{linkCopyDiscardCategory}%
  \defining{linkStrictCopyDiscardCompareCategory}%
  A \intro{copy-discard-compare category} is a symmetric monoidal category, $𝔸$,
  in which every object, $X ∈ 𝔸$, has a compatible partial Frobenius structure,
  consisting of a counit or \emph{discard}, $ε_X ፡ X → I$; a comultiplication or
  \emph{copy}, $δ_X ፡ X → X ⊗ X$; and multiplication or \emph{compare}, $μ_x ፡ X
  ⊗ X → X$; all satisfying the following axioms.
  \begin{enumerate}
    \item Comultiplication is associative, $δ_X ⨾ (δ_X ⊗ \id{}) = δ_X ⨾ (\id{} ⊗ δ_X)$.
    \item Counit is neutral for comultiplication, $δ_X ⨾ (ε_X ⊗ \id{}) = \id{} = δ_X ⨾ (\id{} ⊗ ε_X)$.
    \item Multiplication is associative, $μ_X ⨾ (μ_X ⊗ \id{}) = μ_X ⨾ (\id{} ⊗ μ_X)$.
    \item Multiplication is right inverse to comultiplication, $δ_X ⨾ μ_X = \id{}$.
    \item Multiplication satisfies the Frobenius rule, $(δ_X ⊗ \id{}) ⨾ (\id{} ⊗ μ_X) = (\id{} ⊗ δ_X) ⨾ (μ_X ⊗ \id{}).$
    \item Comultiplication is uniform, $δ_{X ⊗ Y} = (δ_{X} ⊗ δ_{Y}) ⨾ (\id{} ⊗ σ ⊗ \id{})$, and $δ_I = \id{}$.
    \item Multiplication is uniform, $μ_{X ⊗ Y} = (\id{} ⊗ σ_{Y,X} ⊗ \id{}) ⨾ (μ_{X} ⊗ μ_{Y})$, and $μ_I = \id{}$.
    \item Counit is uniform, $ε_{X ⊗ Y} = ε_X ⊗ ε_Y$ and $ε_I = \id{}$.
    \item Comultiplication is commutative, $δ_X ⨾ σ_{X,X} = δ_X$.
    \item Multiplication is commutative, $σ_{X,X} ⨾ μ_{X} = μ_X$.
  \end{enumerate}
  In particular, note that these components are not required to form
  natural transformations, and also that there is no unit for the
  multiplication, $μ ፡ X ⊗ X → X$.

\emph{Strict copy-discard-compare categories} are symmetric
strict monoidal categories with this same structure.
\StrictCopyDiscardCompareCategories{} and strict monoidal functors preserving
copy, discard, and compare form a category, written as \cdcCat{}.
\emph{Copy-discard categories} are defined as \kl{copy-discard-compare categories}
without multiplication \cite{corradini1999algebraic,fritz2020synthetic}.
\end{definition}

\ArrowNotationTerms{} form a \copyDiscardCompareCategory{}: for the single variable case,
the copy, discard, and compare terms are given by the following terms, respectively.
\begin{enumerate}
  \item $δ_X = (x : X ⊢ \return(x,x))$;
  \item $ε_X = (x : X ⊢ \return())$;
  \item $μ_X = (x₁ : X, x₂ : X ⊢ \observe(x₁ = x₂) 𑊩 \return(x₁))$.
  \end{enumerate}

As a second example, functions of type \(X \to \Subd(Y)\) are the morphisms of a
\copyDiscardCompareCategory{}. All the examples from \Cref{sec:illustrations}
take semantics in this category, see \Cref{sec:functorialSemantics}. Details of
the next result, which further justify this semantics, will appear as a separate
appendix.

\begin{toappendix}
\begin{definition}[Copy-discard-compare functor]%
  \label{def:copyDiscardCompareFunctor}%
  \defining{linkCopyDiscardCompareFunctor}%
  A \emph{copy-discard-compare functor} between two
  \copyDiscardCompareCategories{} is a strict monoidal functor that preserves
  the commutative comonoid and multiplication structures.
  \CopyDiscardCompareFunctors{} between \copyDiscardCompareCategories{} form a
  category $\cdcCat$.
\end{definition}
\end{toappendix}

\intro*
\begin{propositionrep}[Copy-discard-compare category of terms]
  \defining{linkArrowF}%
  \AP\kl[observe-arrow notation copy-discard-compare category of terms]{}%
  \kl{Observe-arrow notation terms} over a \kl{signature} $Σ$, quotiented by
  $α$-equivalence and the \axiomsOfArrowNotation{}, form a
  \strictCopyDiscardCompareCategory{}, $\obsArrow(Σ)$.
\end{propositionrep}  
\begin{proof}[Proof]
  We have already defined term composition, and
  \Cref{prop:composition-welldef,prop:term-category} show that it is
  well-defined, associative, and unital. \Cref{def:var:tensor} gives monoidal
  products, and \Cref{prop:term-smc} shows that it is well-defined, symmetric,
  associative, and unital. These give a symmetric strict monoidal category.
  \Cref{prop:term-partial-frobenius} shows the copy-discard-compare structure.
  These results construct a \copyDiscardCompareCategory{}.
\end{proof}

In fact, \kl{observe-arrow notation terms} construct an adjunction between
\strictCopyDiscardCompareCategories{} and \kl{signatures}. Details are again in
the appendix.

\begin{toappendix}
\begin{lemma}
  \label{lemma:inclusion}%
  Let $Σ$ be any \signature{}. There exists a \signatureHomomorphism{}
  $u_{Σ} ፡ Σ → \forgetF(\arrowF(Σ))$.
\end{lemma}
\begin{proof}
  Let us define the homomorphism on types: given any \signature{} type, $X ∈
  Σ_{type}$, the homomorphism sends it to the list with a single element, $u(X)
  = [X]$. Let us now define it on generators; given $f ፡ \vec{x} → \vec{y}$ for any
  two lists of types $\vec{x} = X₁,...,X_x$ and $\vec{y} = Y₁,...,Y_y$, the
  homomorphism sends it to the term containing a single generator statement,  
  $$u(f) = f(x) ⨾ \vret(y).$$
  This term is of type $[X₁],...,[X_x] → [Y₁],...,[Y_y]$, as needed to determine 
  a signature homomorphism.
\end{proof}
\end{toappendix}

\begin{toappendix}
\begin{lemma}%
  \label{lemma:existsFactoring}%
  \defining{linkSignatureFixed}%
  Let $Σ$ be a \signature{} and let $ℂ$ be a \copyDiscardCompareCategory{}
  endowed with a \signatureHomomorphism{}, $\vv ፡ Σ → \forgetF(ℂ)$. There exists at
  most a unique \copyDiscardCompareFunctor{}, $F ፡ \arrowF(Σ) → ℂ$, factoring
  the homomorphism, $\vv = u ⨾ \forgetF(F)$, through the inclusion in \Cref{lemma:inclusion}.
\end{lemma}
\begin{proof}
  Firstly, let us note that the fact that $F$ is a
  strict monoidal functor and that the objects of $\arrowF(Σ)$ are lists forces
  $F$ to be defined on objects as $F(\vec{x}) = F(X₁,...,Xₙ) = v(X₁) ⊗ … ⊗ v(Xₙ).$%
  
  Let $t : \vec{x} → \vec{y}$ in $\arrowF(Σ)$ be a term; we will prove by
  structural induction that its image under $F$ is determined by the fact that
  it is a \copyDiscardCompareFunctor{}.   
  Let the term be a $\return$ statement, $t = \vret(α)$. The statement is
  determined by the function $α ፡ m → n$, which can be decomposed as a symmetric
  monoidal term with commutative comonoids --- uniquely up to the axioms of
  symmetric monoidal categories and commutative comonoids, by 
  \Cref{prop:freeCopyDiscardCategory}. Because $F$ is a strict symmetric
  monoidal functor and must additionally preserve comonoids, its value on this
  term is determined to be $F(\vret(α)) = \{α\}$.

  Let the term be an $\observe$ statement, $t = \vobs(β) 𑊩 s$. The following
  decomposition will rewrite it in terms of composition, whiskering, $\return$
  statements, compare statements, and statements on the image of $v$. These must
  all be preserved by the \copyDiscardCompareFunctor{}, $F$; moreover, its image
  must be determined in the rest of the term, $s$, by structural induction.  We
  use \emph{(i)} unitality of copying, \emph{(ii)} naturality of copying,
  \emph{(iii)} functoriality, \emph{(iv)} the properties of
  \copyDiscardCompareFunctors{}, \emph{(v)} preservation of whiskering,
  \emph{(vi)} definition of \kl{rewiring}, \emph{(vii)} functoriality, and
  \emph{(viii)} the properties of \copyDiscardCompareFunctors{}.
  \begin{align*}
    & F(\vobs(β) 𑊩 s) & \smash{\overset{(i)}{=}} \\
    & F(\vobs(β) 𑊩 \vret(δ_x) 𑊩 \vret(x) ⨾ s) & \smash{\overset{(ii)}{=}} \\
    & F(\vret(δ_x) ⨾ \vobs(β · \vrinc{x}) 𑊩 \vret(x) ⨾ s ) & \smash{\overset{(iii)}{=}}  \\
    & F(\vret(δ_x)) ⨾ F(\vobs(β · \vrinc{x}) 𑊩 \vret(x)) ⨾ F(s) & \smash{\overset{(iv)}{=}}  \\
    & δ_{\vec{x}} ⨾ F(\vec{x} ⋉ (\vobs(β) 𑊩 \vret([]))) ⨾ F(s) & \smash{\overset{(v)}{=}} \\
    & δ_{\vec{x}} ⨾ \id_{\vec{x}} ⊗ F(\vobs(β) 𑊩 \vret([])) ⨾ F(s) & \smash{\overset{(vi)}{=}} \\
    & δ_{\vec{x}} ⨾ \id_{\vec{x}} ⊗ (F(\vret(β)) ⨾ F(\vobs([1,2]) 𑊩 \vret([]))) ⨾ F(s) & \smash{\overset{(vii)}{=}} \\
    & δ_{\vec{x}} ⨾ \id_{\vec{x}} ⊗ (\{β\} ⨾ F(\vobs([1,2]) 𑊩 \vret([1])) ⨾ F(\vret([]))) ⨾ F(s) & \smash{\overset{(viii)}{=}} \\
    & δ_{\vec{x}} ⨾ \id_{\vec{x}} ⊗ (\{β\} ⨾ μ ⨾ ε) ⨾ F(s).
  \end{align*}%
  
  Let the term be a generator statement, $t = f(γ) 𑊩 s$. Again, the following
  decomposition uses only composition, whiskering, $\return$ statements, and
  statements on the image of $v$. These must all be preserved by the
  \copyDiscardCompareFunctor{}, $F$, and moreover its image must be determined
  in the rest of the term, $s$, by structural induction.
  \begin{align*}
    & F(f(γ) ⨾ s) = \\
    & F(f(γ) ⨾ (x + m) ∗ s) = \\
    & F(f(γ) ⨾ (x + m) ∗ s) = \\
    & F(f(γ) ⨾ \vret(x + m) ⨾ s) = \\
    & F(f(γ) ⨾ \vret((x ⋉ \vrinc{x}) ⨾ δ_{x}) ⨾ s) = \\
    & F(f(γ · \vrinc{x} · δ_{x}) ⨾ δ_{x} ∗ \vret(x ⋉ \vrinc{x}) ⨾ s) = \\
    & F(δ_{x} ∗ (f(γ · \vrinc{x}) ⨾ \vret(x ⋉ \vrinc{x})) ⨾ s) = \\
    & F(δ_{x} ∗ (\vec{x} ⋉ (f(γ) ⨾ \vret(\vrinc{x}))) ⨾ s) = \\
    & F(δ_{x} ∗ (\vec{x} ⋉ (f(γ) ⨾ \vret(\vrinc{n} ⨾ (γ ⋊ m)))) ⨾ s) = \\
    & F(δ_{x} ∗ (\vec{x} ⋉ (γ ∗ f(n) 𑊩 (γ ⋊ m) ∗ \vret(\vrinc{n}))) ⨾ s) = \\
    & F(δ_{x} ∗ (\vec{x} ⋉ (\vret(γ) ⨾ f(n) 𑊩 \vret(\vrinc{n}))) ⨾ s) = \\
    & F(\vret(δ_{x}) ⨾ (\vec{x} ⋉ (\vret(γ) ⨾ f(n) 𑊩 \vret(\vrinc{n}))) ⨾ s) = \\
    & F(\vret(δ_{x})) ⨾ F(\vec{x} ⋉ ((\vret(γ)) ⨾ F(f(n) 𑊩 \vret(\vrinc{n})))) ⨾ F(s) = \\
    & F(\vret(δ_{x})) ⨾ (\id_{\vec{x}} ⊗ (F(\vret(γ)) ⨾ F(f(n) 𑊩 \vret(\vrinc{n})))) ⨾ F(s) = \\
    & δ_{\vec{x}} ⨾ (\id_{\vec{x}} ⊗ (\{γ\} ⨾ \vv(f))) ⨾ F(s). \qedhere
  \end{align*}
\end{proof}
\end{toappendix}

\begin{lemmarep}
  \label{lemma:factoring}
  Let $Σ$ be a \kl{signature} and let $ℂ$ be a \copyDiscardCompareCategory{}
  endowed with a \signatureHomomorphism{} $v ፡ Σ → \forgetF(ℂ)$. There exists 
  a unique \copyDiscardCompareFunctor{} $F ፡ \arrowF(Σ) → ℂ$ factoring
  the homomorphism, $v = u ⨾ \forgetF(F)$, through the canonical inclusion
  $u ፡ Σ → \forgetF(\arrowF(Σ))$.
\end{lemmarep}
\begin{proof}
  We have already shown that the only possible functor must be determined by the
  argument in \Cref{lemma:existsFactoring}. Let us rewrite that assignment in
  terms of string diagrams \cite{selinger2010survey}: the few calculations with
  monoidal terms we need will be easier to follow in this notation.
  \begin{enumerate}
    \item On $\return$ statements, we define
    $F(\vret(α)) = \{α\},$
    where $α ፡ m → n$.
    This is \Cref{diag:fig:functor-definition}, right.
    \item On $\observe$ statements, we define
    $$
    F(\vobs(β) 𑊩 s) = 
    δ_{\vec{x}} ⨾ (\id_{\vec{x}} ⊗ (\{β\} ⨾ μ ⨾ ε)) ⨾ F(s),
    $$
    must be assigned to, where $β ፡ 2 → n$. This is
    \Cref{diag:fig:functor-definition}, middle.
    \item On generator statements, we define
    $$
    F(f(γ) 𑊩 s) = δ_{\vec{x}} ⨾ (\id_{\vec{x}} ⊗ (\{γ\} ⨾ \vv(f))) ⨾ F(s),
    $$
    where $f ∈ Σ(\vec{X};\vec{Y})$ and $γ ፡ m → n$.
    This is \Cref{diag:fig:functor-definition}, left.
  \end{enumerate}
  \begin{figure}[h!t]
    \centering

\tikzset{every picture/.style={line width=0.75pt}} %

\begin{tikzpicture}[x=0.75pt,y=0.75pt,yscale=-1,xscale=1]
\draw    (460,85) -- (460,75) ;
\draw    (375,45) .. controls (374.83,35.43) and (366.5,34.85) .. (355,35) ;
\draw    (355,35) -- (355,20) ;
\draw  [fill={rgb, 255:red, 0; green, 0; blue, 0 }  ,fill opacity=1 ] (352,35) .. controls (352,33.34) and (353.34,32) .. (355,32) .. controls (356.66,32) and (358,33.34) .. (358,35) .. controls (358,36.66) and (356.66,38) .. (355,38) .. controls (353.34,38) and (352,36.66) .. (352,35) -- cycle ;
\draw   (360,45) -- (390,45) -- (390,60) -- (360,60) -- cycle ;
\draw    (355,95) -- (355,35) ;
\draw   (345,95) -- (385,95) -- (385,110) -- (345,110) -- cycle ;
\draw    (375,70) -- (375,60) ;
\draw   (360,70) -- (390,70) -- (390,85) -- (360,85) -- cycle ;
\draw    (375,95) -- (375,85) ;
\draw    (365,120) -- (365,110) ;
\draw    (460,45) .. controls (459.83,35.43) and (451.5,34.85) .. (440,35) ;
\draw    (440,35) -- (440,20) ;
\draw  [fill={rgb, 255:red, 0; green, 0; blue, 0 }  ,fill opacity=1 ] (437,35) .. controls (437,33.34) and (438.34,32) .. (440,32) .. controls (441.66,32) and (443,33.34) .. (443,35) .. controls (443,36.66) and (441.66,38) .. (440,38) .. controls (438.34,38) and (437,36.66) .. (437,35) -- cycle ;
\draw   (445,45) -- (475,45) -- (475,60) -- (445,60) -- cycle ;
\draw    (440,95) -- (440,35) ;
\draw   (430,95) -- (470,95) -- (470,110) -- (430,110) -- cycle ;
\draw    (450,65) -- (450,60) ;
\draw    (450,120) -- (450,110) ;
\draw    (470,65) -- (470,60) ;
\draw    (460,75) .. controls (450.09,75.13) and (450.09,74.95) .. (450,65) ;
\draw    (460,75) .. controls (469.91,74.77) and (470.09,74.95) .. (470,65) ;
\draw  [fill={rgb, 255:red, 255; green, 255; blue, 255 }  ,fill opacity=1 ] (457.5,75) .. controls (457.5,73.62) and (458.62,72.5) .. (460,72.5) .. controls (461.38,72.5) and (462.5,73.62) .. (462.5,75) .. controls (462.5,76.38) and (461.38,77.5) .. (460,77.5) .. controls (458.62,77.5) and (457.5,76.38) .. (457.5,75) -- cycle ;
\draw  [fill={rgb, 255:red, 0; green, 0; blue, 0 }  ,fill opacity=1 ] (457.5,85) .. controls (457.5,83.62) and (458.62,82.5) .. (460,82.5) .. controls (461.38,82.5) and (462.5,83.62) .. (462.5,85) .. controls (462.5,86.38) and (461.38,87.5) .. (460,87.5) .. controls (458.62,87.5) and (457.5,86.38) .. (457.5,85) -- cycle ;
\draw   (520,70) -- (550,70) -- (550,85) -- (520,85) -- cycle ;
\draw    (535,100) -- (535,85) ;
\draw    (535,70) -- (535,55) ;
\draw  [draw opacity=0] (400,70) -- (420,70) -- (420,85) -- (400,85) -- cycle ;
\draw  [draw opacity=0] (485,70) -- (505,70) -- (505,85) -- (485,85) -- cycle ;
\draw  [draw opacity=0] (555,71) -- (575,71) -- (575,86) -- (555,86) -- cycle ;

\draw (375,52.5) node  [font=\small]  {$\{\alpha \}$};
\draw (375,77.5) node  [font=\small]  {$f$};
\draw (365,102.5) node  [font=\small]  {$cont$};
\draw (353,23.4) node [anchor=north east] [inner sep=0.75pt]    {$\Gamma $};
\draw (363,113.4) node [anchor=north east] [inner sep=0.75pt]    {$\Delta $};
\draw (392,88.4) node [anchor=north west][inner sep=0.75pt]  [font=\scriptsize]  {$Y$};
\draw (392,63.4) node [anchor=north west][inner sep=0.75pt]  [font=\scriptsize]  {$X$};
\draw (460,52.5) node  [font=\small]  {$\{\alpha \}$};
\draw (450,102.5) node  [font=\small]  {$cont$};
\draw (438,23.4) node [anchor=north east] [inner sep=0.75pt]    {$\Gamma $};
\draw (448,113.4) node [anchor=north east] [inner sep=0.75pt]    {$\Delta $};
\draw (472,63.4) node [anchor=north west][inner sep=0.75pt]  [font=\scriptsize]  {$X$};
\draw (535,77.5) node  [font=\small]  {$\{\alpha \}$};
\draw (533,88.4) node [anchor=north east] [inner sep=0.75pt]    {$\Delta $};
\draw (532,52.4) node [anchor=north east] [inner sep=0.75pt]    {$\Gamma $};
\draw (410,77.5) node  [font=\small]  {$;$};
\draw (495,77.5) node  [font=\small]  {$;$};
\draw (565,78.5) node  [font=\small]  {$;$};

\end{tikzpicture}
     \caption{Translation of generator, observe, and return statements, respectively.}
    \label{diag:fig:functor-definition}
  \end{figure}

  We must first prove that this assingment is well-defined with respect to the
  \axiomsOfArrowNotation{}: (1) the \interchangeAxiom{} follows from
  associativity of the comonoid structure (see
  \Cref{fig:diagramThreeInterchange}); (2) the \symmetryAxiom{} follows from
  commutativity of the comparator structure (see
  \Cref{fig:diagramThreeFrobenius}, left); (3) the \FrobeniusAxiom{} follows
  from the properties of a partial Frobenius algebra (see
  \Cref{fig:diagramThreeFrobenius}, middle); (4) the
  \idempotencyAxiom{} follows from the special axiom of a partial Frobenius
  algebra — multiplication is a right inverse to comultiplication (see
  \Cref{fig:diagramThreeFrobenius}, right).%
  \begin{figure}[ht!]
    \centering
    \makebox[\textwidth][c]{

\tikzset{every picture/.style={line width=0.75pt}} %

\begin{tikzpicture}[x=0.75pt,y=0.75pt,yscale=-1,xscale=1]
\draw    (70,100) .. controls (69.83,90.43) and (36.5,89.85) .. (25,90) ;
\draw  [draw opacity=0][fill={rgb, 255:red, 255; green, 255; blue, 255 }  ,fill opacity=1 ] (42,90) .. controls (42,88.34) and (43.34,87) .. (45,87) .. controls (46.66,87) and (48,88.34) .. (48,90) .. controls (48,91.66) and (46.66,93) .. (45,93) .. controls (43.34,93) and (42,91.66) .. (42,90) -- cycle ;
\draw    (45,40) .. controls (44.83,30.43) and (36.5,29.85) .. (25,30) ;
\draw    (25,30) -- (25,15) ;
\draw  [fill={rgb, 255:red, 0; green, 0; blue, 0 }  ,fill opacity=1 ] (22,30) .. controls (22,28.34) and (23.34,27) .. (25,27) .. controls (26.66,27) and (28,28.34) .. (28,30) .. controls (28,31.66) and (26.66,33) .. (25,33) .. controls (23.34,33) and (22,31.66) .. (22,30) -- cycle ;
\draw   (30,40) -- (60,40) -- (60,55) -- (30,55) -- cycle ;
\draw    (25,90) -- (25,30) ;
\draw   (15,150) -- (85,150) -- (85,165) -- (15,165) -- cycle ;
\draw    (45,65) -- (45,55) ;
\draw   (30,65) -- (60,65) -- (60,80) -- (30,80) -- cycle ;
\draw    (45,150) -- (45,80) ;
\draw    (50,175) -- (50,165) ;
\draw  [fill={rgb, 255:red, 0; green, 0; blue, 0 }  ,fill opacity=1 ] (22,90) .. controls (22,88.34) and (23.34,87) .. (25,87) .. controls (26.66,87) and (28,88.34) .. (28,90) .. controls (28,91.66) and (26.66,93) .. (25,93) .. controls (23.34,93) and (22,91.66) .. (22,90) -- cycle ;
\draw   (55,100) -- (85,100) -- (85,115) -- (55,115) -- cycle ;
\draw    (70,125) -- (70,115) ;
\draw   (55,125) -- (85,125) -- (85,140) -- (55,140) -- cycle ;
\draw    (70,150) -- (70,140) ;
\draw    (25,150) -- (25,90) ;
\draw    (160,95) .. controls (159.83,85.43) and (126.5,84.85) .. (115,85) ;
\draw  [draw opacity=0][fill={rgb, 255:red, 255; green, 255; blue, 255 }  ,fill opacity=1 ] (132,85) .. controls (132,83.34) and (133.34,82) .. (135,82) .. controls (136.66,82) and (138,83.34) .. (138,85) .. controls (138,86.66) and (136.66,88) .. (135,88) .. controls (133.34,88) and (132,86.66) .. (132,85) -- cycle ;
\draw    (135,40) .. controls (134.83,30.43) and (126.5,29.85) .. (115,30) ;
\draw    (115,30) -- (115,15) ;
\draw  [fill={rgb, 255:red, 0; green, 0; blue, 0 }  ,fill opacity=1 ] (112,30) .. controls (112,28.34) and (113.34,27) .. (115,27) .. controls (116.66,27) and (118,28.34) .. (118,30) .. controls (118,31.66) and (116.66,33) .. (115,33) .. controls (113.34,33) and (112,31.66) .. (112,30) -- cycle ;
\draw   (120,40) -- (150,40) -- (150,55) -- (120,55) -- cycle ;
\draw    (115,85) -- (115,25) ;
\draw    (135,60) -- (135,55) ;
\draw   (120,60) -- (150,60) -- (150,75) -- (120,75) -- cycle ;
\draw    (135,130) -- (135,75) ;
\draw  [fill={rgb, 255:red, 0; green, 0; blue, 0 }  ,fill opacity=1 ] (112,85) .. controls (112,83.34) and (113.34,82) .. (115,82) .. controls (116.66,82) and (118,83.34) .. (118,85) .. controls (118,86.66) and (116.66,88) .. (115,88) .. controls (113.34,88) and (112,86.66) .. (112,85) -- cycle ;
\draw   (145,95) -- (175,95) -- (175,110) -- (145,110) -- cycle ;
\draw    (160,115) -- (160,110) ;
\draw   (145,115) -- (175,115) -- (175,130) -- (145,130) -- cycle ;
\draw    (115,150) -- (115,85) ;
\draw   (105,150) -- (175,150) -- (175,165) -- (105,165) -- cycle ;
\draw    (140,175) -- (140,165) ;
\draw    (135,130) .. controls (135,139.48) and (160.33,141.15) .. (160,150) ;
\draw  [draw opacity=0][fill={rgb, 255:red, 255; green, 255; blue, 255 }  ,fill opacity=1 ] (143.83,140.17) .. controls (143.83,138.51) and (145.18,137.17) .. (146.83,137.17) .. controls (148.49,137.17) and (149.83,138.51) .. (149.83,140.17) .. controls (149.83,141.82) and (148.49,143.17) .. (146.83,143.17) .. controls (145.18,143.17) and (143.83,141.82) .. (143.83,140.17) -- cycle ;
\draw    (160,130) .. controls (159.33,140.48) and (135.33,139.15) .. (135,150) ;
\draw  [draw opacity=0] (80,75) -- (110,75) -- (110,95) -- (80,95) -- cycle ;
\draw  [draw opacity=0] (175,75) -- (205,75) -- (205,95) -- (175,95) -- cycle ;
\draw  [draw opacity=0][fill={rgb, 255:red, 255; green, 255; blue, 255 }  ,fill opacity=1 ] (232,90) .. controls (232,88.34) and (233.34,87) .. (235,87) .. controls (236.66,87) and (238,88.34) .. (238,90) .. controls (238,91.66) and (236.66,93) .. (235,93) .. controls (233.34,93) and (232,91.66) .. (232,90) -- cycle ;
\draw    (235,40) .. controls (234.83,30.43) and (226.5,29.85) .. (215,30) ;
\draw    (215,30) -- (215,15) ;
\draw  [fill={rgb, 255:red, 0; green, 0; blue, 0 }  ,fill opacity=1 ] (212,30) .. controls (212,28.34) and (213.34,27) .. (215,27) .. controls (216.66,27) and (218,28.34) .. (218,30) .. controls (218,31.66) and (216.66,33) .. (215,33) .. controls (213.34,33) and (212,31.66) .. (212,30) -- cycle ;
\draw   (220,40) -- (250,40) -- (250,55) -- (220,55) -- cycle ;
\draw    (215,90) -- (215,30) ;
\draw   (205,150) -- (275,150) -- (275,165) -- (205,165) -- cycle ;
\draw    (240,175) -- (240,165) ;
\draw  [fill={rgb, 255:red, 0; green, 0; blue, 0 }  ,fill opacity=1 ] (212,90) .. controls (212,88.34) and (213.34,87) .. (215,87) .. controls (216.66,87) and (218,88.34) .. (218,90) .. controls (218,91.66) and (216.66,93) .. (215,93) .. controls (213.34,93) and (212,91.66) .. (212,90) -- cycle ;
\draw   (245,100) -- (275,100) -- (275,115) -- (245,115) -- cycle ;
\draw    (260,125) -- (260,115) ;
\draw   (245,125) -- (275,125) -- (275,140) -- (245,140) -- cycle ;
\draw    (260,150) -- (260,140) ;
\draw    (215,150) -- (215,90) ;
\draw    (350,95) .. controls (349.83,85.43) and (316.5,84.85) .. (305,85) ;
\draw  [draw opacity=0][fill={rgb, 255:red, 255; green, 255; blue, 255 }  ,fill opacity=1 ] (322,85) .. controls (322,83.34) and (323.34,82) .. (325,82) .. controls (326.66,82) and (328,83.34) .. (328,85) .. controls (328,86.66) and (326.66,88) .. (325,88) .. controls (323.34,88) and (322,86.66) .. (322,85) -- cycle ;
\draw    (325,40) .. controls (324.83,30.43) and (316.5,29.85) .. (305,30) ;
\draw    (305,30) -- (305,15) ;
\draw  [fill={rgb, 255:red, 0; green, 0; blue, 0 }  ,fill opacity=1 ] (302,30) .. controls (302,28.34) and (303.34,27) .. (305,27) .. controls (306.66,27) and (308,28.34) .. (308,30) .. controls (308,31.66) and (306.66,33) .. (305,33) .. controls (303.34,33) and (302,31.66) .. (302,30) -- cycle ;
\draw   (310,40) -- (340,40) -- (340,55) -- (310,55) -- cycle ;
\draw    (305,85) -- (305,25) ;
\draw    (325,60) -- (325,55) ;
\draw   (310,60) -- (340,60) -- (340,75) -- (310,75) -- cycle ;
\draw    (325,130) -- (325,75) ;
\draw  [fill={rgb, 255:red, 0; green, 0; blue, 0 }  ,fill opacity=1 ] (302,85) .. controls (302,83.34) and (303.34,82) .. (305,82) .. controls (306.66,82) and (308,83.34) .. (308,85) .. controls (308,86.66) and (306.66,88) .. (305,88) .. controls (303.34,88) and (302,86.66) .. (302,85) -- cycle ;
\draw   (335,95) -- (365,95) -- (365,110) -- (335,110) -- cycle ;
\draw    (305,150) -- (305,85) ;
\draw   (295,150) -- (365,150) -- (365,165) -- (295,165) -- cycle ;
\draw    (330,175) -- (330,165) ;
\draw  [draw opacity=0][fill={rgb, 255:red, 255; green, 255; blue, 255 }  ,fill opacity=1 ] (333.83,140.17) .. controls (333.83,138.51) and (335.18,137.17) .. (336.83,137.17) .. controls (338.49,137.17) and (339.83,138.51) .. (339.83,140.17) .. controls (339.83,141.82) and (338.49,143.17) .. (336.83,143.17) .. controls (335.18,143.17) and (333.83,141.82) .. (333.83,140.17) -- cycle ;
\draw  [draw opacity=0] (270,75) -- (300,75) -- (300,95) -- (270,95) -- cycle ;
\draw  [draw opacity=0] (365,75) -- (395,75) -- (395,95) -- (365,95) -- cycle ;
\draw    (235,75) -- (235,65) ;
\draw    (235,65) .. controls (225.09,65.13) and (225.09,64.95) .. (225,55) ;
\draw    (235,65) .. controls (244.91,64.77) and (245.09,64.95) .. (245,55) ;
\draw  [fill={rgb, 255:red, 255; green, 255; blue, 255 }  ,fill opacity=1 ] (232.5,65) .. controls (232.5,63.62) and (233.62,62.5) .. (235,62.5) .. controls (236.38,62.5) and (237.5,63.62) .. (237.5,65) .. controls (237.5,66.38) and (236.38,67.5) .. (235,67.5) .. controls (233.62,67.5) and (232.5,66.38) .. (232.5,65) -- cycle ;
\draw  [draw opacity=0][fill={rgb, 255:red, 0; green, 0; blue, 0 }  ,fill opacity=1 ] (232,75) .. controls (232,73.34) and (233.34,72) .. (235,72) .. controls (236.66,72) and (238,73.34) .. (238,75) .. controls (238,76.66) and (236.66,78) .. (235,78) .. controls (233.34,78) and (232,76.66) .. (232,75) -- cycle ;
\draw    (260,100) .. controls (259.83,90.43) and (226.5,89.85) .. (215,90) ;
\draw    (349.99,130) -- (349.99,120) ;
\draw    (349.99,120) .. controls (340.08,120.13) and (340.08,119.95) .. (339.99,110) ;
\draw    (349.99,120) .. controls (359.9,119.77) and (360.08,119.95) .. (359.99,110) ;
\draw  [fill={rgb, 255:red, 255; green, 255; blue, 255 }  ,fill opacity=1 ] (347.49,120) .. controls (347.49,118.62) and (348.61,117.5) .. (349.99,117.5) .. controls (351.37,117.5) and (352.49,118.62) .. (352.49,120) .. controls (352.49,121.38) and (351.37,122.5) .. (349.99,122.5) .. controls (348.61,122.5) and (347.49,121.38) .. (347.49,120) -- cycle ;
\draw  [draw opacity=0][fill={rgb, 255:red, 0; green, 0; blue, 0 }  ,fill opacity=1 ] (346.99,130) .. controls (346.99,128.34) and (348.33,127) .. (349.99,127) .. controls (351.64,127) and (352.99,128.34) .. (352.99,130) .. controls (352.99,131.66) and (351.64,133) .. (349.99,133) .. controls (348.33,133) and (346.99,131.66) .. (346.99,130) -- cycle ;
\draw    (325,130) .. controls (325,139.48) and (350.33,141.15) .. (350,150) ;
\draw  [draw opacity=0] (365,75) -- (395,75) -- (395,95) -- (365,95) -- cycle ;
\draw  [draw opacity=0][fill={rgb, 255:red, 255; green, 255; blue, 255 }  ,fill opacity=1 ] (422,90) .. controls (422,88.34) and (423.34,87) .. (425,87) .. controls (426.66,87) and (428,88.34) .. (428,90) .. controls (428,91.66) and (426.66,93) .. (425,93) .. controls (423.34,93) and (422,91.66) .. (422,90) -- cycle ;
\draw    (425,40) .. controls (424.83,30.43) and (416.5,29.85) .. (405,30) ;
\draw    (405,30) -- (405,15) ;
\draw  [fill={rgb, 255:red, 0; green, 0; blue, 0 }  ,fill opacity=1 ] (402,30) .. controls (402,28.34) and (403.34,27) .. (405,27) .. controls (406.66,27) and (408,28.34) .. (408,30) .. controls (408,31.66) and (406.66,33) .. (405,33) .. controls (403.34,33) and (402,31.66) .. (402,30) -- cycle ;
\draw   (410,40) -- (440,40) -- (440,55) -- (410,55) -- cycle ;
\draw    (405,90) -- (405,30) ;
\draw   (395,150) -- (465,150) -- (465,165) -- (395,165) -- cycle ;
\draw    (430,175) -- (430,165) ;
\draw  [fill={rgb, 255:red, 0; green, 0; blue, 0 }  ,fill opacity=1 ] (402,90) .. controls (402,88.34) and (403.34,87) .. (405,87) .. controls (406.66,87) and (408,88.34) .. (408,90) .. controls (408,91.66) and (406.66,93) .. (405,93) .. controls (403.34,93) and (402,91.66) .. (402,90) -- cycle ;
\draw   (435,100) -- (465,100) -- (465,115) -- (435,115) -- cycle ;
\draw    (405,150) -- (405,90) ;
\draw  [draw opacity=0][fill={rgb, 255:red, 255; green, 255; blue, 255 }  ,fill opacity=1 ] (512,85) .. controls (512,83.34) and (513.34,82) .. (515,82) .. controls (516.66,82) and (518,83.34) .. (518,85) .. controls (518,86.66) and (516.66,88) .. (515,88) .. controls (513.34,88) and (512,86.66) .. (512,85) -- cycle ;
\draw    (515,40) .. controls (514.83,30.43) and (506.5,29.85) .. (495,30) ;
\draw    (495,30) -- (495,15) ;
\draw  [fill={rgb, 255:red, 0; green, 0; blue, 0 }  ,fill opacity=1 ] (492,30) .. controls (492,28.34) and (493.34,27) .. (495,27) .. controls (496.66,27) and (498,28.34) .. (498,30) .. controls (498,31.66) and (496.66,33) .. (495,33) .. controls (493.34,33) and (492,31.66) .. (492,30) -- cycle ;
\draw   (500,40) -- (530,40) -- (530,55) -- (500,55) -- cycle ;
\draw    (495,85) -- (495,25) ;
\draw  [fill={rgb, 255:red, 0; green, 0; blue, 0 }  ,fill opacity=1 ] (492,85) .. controls (492,83.34) and (493.34,82) .. (495,82) .. controls (496.66,82) and (498,83.34) .. (498,85) .. controls (498,86.66) and (496.66,88) .. (495,88) .. controls (493.34,88) and (492,86.66) .. (492,85) -- cycle ;
\draw   (525,95) -- (555,95) -- (555,110) -- (525,110) -- cycle ;
\draw    (495,150) -- (495,85) ;
\draw   (485,150) -- (555,150) -- (555,165) -- (485,165) -- cycle ;
\draw    (520,175) -- (520,165) ;
\draw  [draw opacity=0][fill={rgb, 255:red, 255; green, 255; blue, 255 }  ,fill opacity=1 ] (523.83,140.17) .. controls (523.83,138.51) and (525.18,137.17) .. (526.83,137.17) .. controls (528.49,137.17) and (529.83,138.51) .. (529.83,140.17) .. controls (529.83,141.82) and (528.49,143.17) .. (526.83,143.17) .. controls (525.18,143.17) and (523.83,141.82) .. (523.83,140.17) -- cycle ;
\draw  [draw opacity=0] (460,75) -- (490,75) -- (490,95) -- (460,95) -- cycle ;
\draw  [draw opacity=0] (555,75) -- (585,75) -- (585,95) -- (555,95) -- cycle ;
\draw    (425,75) -- (425,65) ;
\draw    (425,65) .. controls (415.09,65.13) and (415.09,64.95) .. (415,55) ;
\draw    (425,65) .. controls (434.91,64.77) and (435.09,64.95) .. (435,55) ;
\draw  [fill={rgb, 255:red, 255; green, 255; blue, 255 }  ,fill opacity=1 ] (422.5,65) .. controls (422.5,63.62) and (423.62,62.5) .. (425,62.5) .. controls (426.38,62.5) and (427.5,63.62) .. (427.5,65) .. controls (427.5,66.38) and (426.38,67.5) .. (425,67.5) .. controls (423.62,67.5) and (422.5,66.38) .. (422.5,65) -- cycle ;
\draw  [draw opacity=0][fill={rgb, 255:red, 0; green, 0; blue, 0 }  ,fill opacity=1 ] (422,75) .. controls (422,73.34) and (423.34,72) .. (425,72) .. controls (426.66,72) and (428,73.34) .. (428,75) .. controls (428,76.66) and (426.66,78) .. (425,78) .. controls (423.34,78) and (422,76.66) .. (422,75) -- cycle ;
\draw    (450,100) .. controls (449.83,90.43) and (416.5,89.85) .. (405,90) ;
\draw    (539.99,130) -- (539.99,120) ;
\draw    (539.99,120) .. controls (530.08,120.13) and (530.08,119.95) .. (529.99,110) ;
\draw    (539.99,120) .. controls (549.9,119.77) and (550.08,119.95) .. (549.99,110) ;
\draw  [fill={rgb, 255:red, 255; green, 255; blue, 255 }  ,fill opacity=1 ] (537.49,120) .. controls (537.49,118.62) and (538.61,117.5) .. (539.99,117.5) .. controls (541.37,117.5) and (542.49,118.62) .. (542.49,120) .. controls (542.49,121.38) and (541.37,122.5) .. (539.99,122.5) .. controls (538.61,122.5) and (537.49,121.38) .. (537.49,120) -- cycle ;
\draw  [draw opacity=0][fill={rgb, 255:red, 0; green, 0; blue, 0 }  ,fill opacity=1 ] (536.99,130) .. controls (536.99,128.34) and (538.33,127) .. (539.99,127) .. controls (541.64,127) and (542.99,128.34) .. (542.99,130) .. controls (542.99,131.66) and (541.64,133) .. (539.99,133) .. controls (538.33,133) and (536.99,131.66) .. (536.99,130) -- cycle ;
\draw    (540,95) .. controls (539.83,85.43) and (506.5,84.85) .. (495,85) ;
\draw    (450,135) -- (450,125) ;
\draw    (450,125) .. controls (440.09,125.13) and (440.09,124.95) .. (440,115) ;
\draw    (450,125) .. controls (459.91,124.77) and (460.09,124.95) .. (460,115) ;
\draw  [fill={rgb, 255:red, 255; green, 255; blue, 255 }  ,fill opacity=1 ] (447.5,125) .. controls (447.5,123.62) and (448.62,122.5) .. (450,122.5) .. controls (451.38,122.5) and (452.5,123.62) .. (452.5,125) .. controls (452.5,126.38) and (451.38,127.5) .. (450,127.5) .. controls (448.62,127.5) and (447.5,126.38) .. (447.5,125) -- cycle ;
\draw  [draw opacity=0][fill={rgb, 255:red, 0; green, 0; blue, 0 }  ,fill opacity=1 ] (447,135) .. controls (447,133.34) and (448.34,132) .. (450,132) .. controls (451.66,132) and (453,133.34) .. (453,135) .. controls (453,136.66) and (451.66,138) .. (450,138) .. controls (448.34,138) and (447,136.66) .. (447,135) -- cycle ;
\draw    (514.99,75) -- (514.99,65) ;
\draw    (514.99,65) .. controls (505.08,65.13) and (505.08,64.95) .. (504.99,55) ;
\draw    (514.99,65) .. controls (524.9,64.77) and (525.08,64.95) .. (524.99,55) ;
\draw  [fill={rgb, 255:red, 255; green, 255; blue, 255 }  ,fill opacity=1 ] (512.49,65) .. controls (512.49,63.62) and (513.61,62.5) .. (514.99,62.5) .. controls (516.37,62.5) and (517.49,63.62) .. (517.49,65) .. controls (517.49,66.38) and (516.37,67.5) .. (514.99,67.5) .. controls (513.61,67.5) and (512.49,66.38) .. (512.49,65) -- cycle ;
\draw  [draw opacity=0][fill={rgb, 255:red, 0; green, 0; blue, 0 }  ,fill opacity=1 ] (511.99,75) .. controls (511.99,73.34) and (513.33,72) .. (514.99,72) .. controls (516.64,72) and (517.99,73.34) .. (517.99,75) .. controls (517.99,76.66) and (516.64,78) .. (514.99,78) .. controls (513.33,78) and (511.99,76.66) .. (511.99,75) -- cycle ;

\draw (45,47.5) node  [font=\small]  {$\{\alpha \}$};
\draw (45,72.5) node  [font=\small]  {$f$};
\draw (70,107.5) node  [font=\small]  {$\{\beta \}$};
\draw (70,132.5) node  [font=\small]  {$g$};
\draw (50,157.5) node  [font=\small]  {$cont$};
\draw (135,47.5) node  [font=\small]  {$\{\beta \}$};
\draw (135,67.5) node  [font=\small]  {$g$};
\draw (160,102.5) node  [font=\small]  {$\{\alpha \}$};
\draw (160,123.5) node  [font=\small]  {$f$};
\draw (140,157.5) node  [font=\small]  {$cont$};
\draw (95,85) node  [font=\small]  {$=$};
\draw (190,85) node  [font=\small]  {$;$};
\draw (235,47.5) node  [font=\small]  {$\{\alpha \}$};
\draw (260,107.5) node  [font=\small]  {$\{\beta \}$};
\draw (260,132.5) node  [font=\small]  {$g$};
\draw (240,157.5) node  [font=\small]  {$cont$};
\draw (325,47.5) node  [font=\small]  {$\{\beta \}$};
\draw (325,67.5) node  [font=\small]  {$g$};
\draw (350,102.5) node  [font=\small]  {$\{\alpha \}$};
\draw (330,157.5) node  [font=\small]  {$cont$};
\draw (285,85) node  [font=\small]  {$=$};
\draw (380,85) node  [font=\small]  {$;$};
\draw (425,47.5) node  [font=\small]  {$\{\alpha \}$};
\draw (450,107.5) node  [font=\small]  {$\{\beta \}$};
\draw (430,157.5) node  [font=\small]  {$cont$};
\draw (515,47.5) node  [font=\small]  {$\{\beta \}$};
\draw (540,102.5) node  [font=\small]  {$\{\alpha \}$};
\draw (520,157.5) node  [font=\small]  {$cont$};
\draw (475,85) node  [font=\small]  {$=$};
\draw (570,85) node  [font=\small]  {$;$};

\end{tikzpicture}
     }
    \caption{Three cases of the \protect\interchangeAxiom{}.}
    \label{fig:diagramThreeInterchange}
  \end{figure}
  \begin{figure}[ht!]
    \centering
    \makebox[\textwidth][c]{

\tikzset{every picture/.style={line width=0.75pt}} %

\begin{tikzpicture}[x=0.75pt,y=0.75pt,yscale=-1,xscale=1]
\draw    (70,65) -- (70,55) ;
\draw    (60,30) .. controls (59.43,25.03) and (61.5,24.85) .. (50,25) ;
\draw    (50,25) -- (50,10) ;
\draw    (50,70) -- (50,25) ;
\draw    (70,55) .. controls (60.09,55.13) and (60.09,54.95) .. (60,45) ;
\draw    (70,55) .. controls (79.91,54.77) and (80.09,54.95) .. (80,45) ;
\draw  [fill={rgb, 255:red, 255; green, 255; blue, 255 }  ,fill opacity=1 ] (67.5,55) .. controls (67.5,53.62) and (68.62,52.5) .. (70,52.5) .. controls (71.38,52.5) and (72.5,53.62) .. (72.5,55) .. controls (72.5,56.38) and (71.38,57.5) .. (70,57.5) .. controls (68.62,57.5) and (67.5,56.38) .. (67.5,55) -- cycle ;
\draw  [fill={rgb, 255:red, 0; green, 0; blue, 0 }  ,fill opacity=1 ] (67.5,65) .. controls (67.5,63.62) and (68.62,62.5) .. (70,62.5) .. controls (71.38,62.5) and (72.5,63.62) .. (72.5,65) .. controls (72.5,66.38) and (71.38,67.5) .. (70,67.5) .. controls (68.62,67.5) and (67.5,66.38) .. (67.5,65) -- cycle ;
\draw    (60,45) .. controls (59.9,39.62) and (80.1,40.02) .. (80,30) ;
\draw    (80,45) .. controls (79.9,39.62) and (60.1,40.02) .. (60,30) ;
\draw  [fill={rgb, 255:red, 0; green, 0; blue, 0 }  ,fill opacity=1 ] (47.5,25) .. controls (47.5,23.62) and (48.62,22.5) .. (50,22.5) .. controls (51.38,22.5) and (52.5,23.62) .. (52.5,25) .. controls (52.5,26.38) and (51.38,27.5) .. (50,27.5) .. controls (48.62,27.5) and (47.5,26.38) .. (47.5,25) -- cycle ;
\draw    (80,30) .. controls (79.83,20.43) and (67.43,19.89) .. (40,20) ;
\draw    (40,70) -- (40,10) ;
\draw  [fill={rgb, 255:red, 0; green, 0; blue, 0 }  ,fill opacity=1 ] (37.5,20) .. controls (37.5,18.62) and (38.62,17.5) .. (40,17.5) .. controls (41.38,17.5) and (42.5,18.62) .. (42.5,20) .. controls (42.5,21.38) and (41.38,22.5) .. (40,22.5) .. controls (38.62,22.5) and (37.5,21.38) .. (37.5,20) -- cycle ;
\draw    (134.99,50) -- (134.99,40) ;
\draw    (124.99,30) .. controls (124.42,25.03) and (126.49,24.85) .. (114.99,25) ;
\draw    (114.99,25) -- (114.99,10) ;
\draw    (114.99,70) -- (114.99,25) ;
\draw    (134.99,40) .. controls (125.08,40.13) and (125.08,39.95) .. (124.99,30) ;
\draw    (134.99,40) .. controls (144.9,39.77) and (145.08,39.95) .. (144.99,30) ;
\draw  [fill={rgb, 255:red, 255; green, 255; blue, 255 }  ,fill opacity=1 ] (132.49,40) .. controls (132.49,38.62) and (133.61,37.5) .. (134.99,37.5) .. controls (136.37,37.5) and (137.49,38.62) .. (137.49,40) .. controls (137.49,41.38) and (136.37,42.5) .. (134.99,42.5) .. controls (133.61,42.5) and (132.49,41.38) .. (132.49,40) -- cycle ;
\draw  [fill={rgb, 255:red, 0; green, 0; blue, 0 }  ,fill opacity=1 ] (132.49,50) .. controls (132.49,48.62) and (133.61,47.5) .. (134.99,47.5) .. controls (136.37,47.5) and (137.49,48.62) .. (137.49,50) .. controls (137.49,51.38) and (136.37,52.5) .. (134.99,52.5) .. controls (133.61,52.5) and (132.49,51.38) .. (132.49,50) -- cycle ;
\draw  [fill={rgb, 255:red, 0; green, 0; blue, 0 }  ,fill opacity=1 ] (112.49,25) .. controls (112.49,23.62) and (113.61,22.5) .. (114.99,22.5) .. controls (116.37,22.5) and (117.49,23.62) .. (117.49,25) .. controls (117.49,26.38) and (116.37,27.5) .. (114.99,27.5) .. controls (113.61,27.5) and (112.49,26.38) .. (112.49,25) -- cycle ;
\draw    (144.99,30) .. controls (144.82,20.43) and (132.42,19.89) .. (104.99,20) ;
\draw    (104.99,70) -- (104.99,10) ;
\draw  [fill={rgb, 255:red, 0; green, 0; blue, 0 }  ,fill opacity=1 ] (102.49,20) .. controls (102.49,18.62) and (103.61,17.5) .. (104.99,17.5) .. controls (106.37,17.5) and (107.49,18.62) .. (107.49,20) .. controls (107.49,21.38) and (106.37,22.5) .. (104.99,22.5) .. controls (103.61,22.5) and (102.49,21.38) .. (102.49,20) -- cycle ;
\draw  [draw opacity=0] (75,30) -- (105,30) -- (105,50) -- (75,50) -- cycle ;
\draw  [draw opacity=0] (35,50) -- (65,50) -- (65,70) -- (35,70) -- cycle ;
\draw    (330,60) -- (330,50) ;
\draw    (320,40) .. controls (319.43,35.03) and (321.5,34.85) .. (310,35) ;
\draw    (310,35) -- (310,20) ;
\draw    (330,50) .. controls (320.09,50.13) and (320.09,49.95) .. (320,40) ;
\draw    (330,50) .. controls (339.91,49.77) and (340.09,49.95) .. (340,40) ;
\draw  [fill={rgb, 255:red, 255; green, 255; blue, 255 }  ,fill opacity=1 ] (327.5,50) .. controls (327.5,48.62) and (328.62,47.5) .. (330,47.5) .. controls (331.38,47.5) and (332.5,48.62) .. (332.5,50) .. controls (332.5,51.38) and (331.38,52.5) .. (330,52.5) .. controls (328.62,52.5) and (327.5,51.38) .. (327.5,50) -- cycle ;
\draw  [fill={rgb, 255:red, 0; green, 0; blue, 0 }  ,fill opacity=1 ] (327.5,60) .. controls (327.5,58.62) and (328.62,57.5) .. (330,57.5) .. controls (331.38,57.5) and (332.5,58.62) .. (332.5,60) .. controls (332.5,61.38) and (331.38,62.5) .. (330,62.5) .. controls (328.62,62.5) and (327.5,61.38) .. (327.5,60) -- cycle ;
\draw  [fill={rgb, 255:red, 0; green, 0; blue, 0 }  ,fill opacity=1 ] (307.5,35) .. controls (307.5,33.62) and (308.62,32.5) .. (310,32.5) .. controls (311.38,32.5) and (312.5,33.62) .. (312.5,35) .. controls (312.5,36.38) and (311.38,37.5) .. (310,37.5) .. controls (308.62,37.5) and (307.5,36.38) .. (307.5,35) -- cycle ;
\draw    (340,40) .. controls (339.83,30.43) and (337.43,24.89) .. (310,25) ;
\draw    (310,70) -- (310,10) ;
\draw  [fill={rgb, 255:red, 0; green, 0; blue, 0 }  ,fill opacity=1 ] (307.5,25) .. controls (307.5,23.62) and (308.62,22.5) .. (310,22.5) .. controls (311.38,22.5) and (312.5,23.62) .. (312.5,25) .. controls (312.5,26.38) and (311.38,27.5) .. (310,27.5) .. controls (308.62,27.5) and (307.5,26.38) .. (307.5,25) -- cycle ;
\draw  [draw opacity=0] (295,45) -- (325,45) -- (325,65) -- (295,65) -- cycle ;
\draw    (370,70) -- (370,10) ;
\draw  [draw opacity=0] (280,35) -- (310,35) -- (310,55) -- (280,55) -- cycle ;
\draw  [draw opacity=0] (370,35) -- (400,35) -- (400,55) -- (370,55) -- cycle ;
\draw  [draw opacity=0] (370,35) -- (400,35) -- (400,55) -- (370,55) -- cycle ;
\draw    (205,50) -- (205,40) ;
\draw    (195,30) .. controls (194.43,25.03) and (196.5,24.85) .. (185,25) ;
\draw    (185,25) -- (185,10) ;
\draw    (185,70) -- (185,25) ;
\draw    (205,40) .. controls (195.09,40.13) and (195.09,39.95) .. (195,30) ;
\draw    (205,40) .. controls (214.91,39.77) and (215.09,39.95) .. (215,30) ;
\draw  [fill={rgb, 255:red, 255; green, 255; blue, 255 }  ,fill opacity=1 ] (202.5,40) .. controls (202.5,38.62) and (203.62,37.5) .. (205,37.5) .. controls (206.38,37.5) and (207.5,38.62) .. (207.5,40) .. controls (207.5,41.38) and (206.38,42.5) .. (205,42.5) .. controls (203.62,42.5) and (202.5,41.38) .. (202.5,40) -- cycle ;
\draw  [fill={rgb, 255:red, 0; green, 0; blue, 0 }  ,fill opacity=1 ] (202.5,50) .. controls (202.5,48.62) and (203.62,47.5) .. (205,47.5) .. controls (206.38,47.5) and (207.5,48.62) .. (207.5,50) .. controls (207.5,51.38) and (206.38,52.5) .. (205,52.5) .. controls (203.62,52.5) and (202.5,51.38) .. (202.5,50) -- cycle ;
\draw  [fill={rgb, 255:red, 0; green, 0; blue, 0 }  ,fill opacity=1 ] (182.5,25) .. controls (182.5,23.62) and (183.62,22.5) .. (185,22.5) .. controls (186.38,22.5) and (187.5,23.62) .. (187.5,25) .. controls (187.5,26.38) and (186.38,27.5) .. (185,27.5) .. controls (183.62,27.5) and (182.5,26.38) .. (182.5,25) -- cycle ;
\draw    (215,30) .. controls (214.83,20.43) and (202.43,19.89) .. (175,20) ;
\draw    (175,70) -- (175,10) ;
\draw  [fill={rgb, 255:red, 0; green, 0; blue, 0 }  ,fill opacity=1 ] (172.5,20) .. controls (172.5,18.62) and (173.62,17.5) .. (175,17.5) .. controls (176.38,17.5) and (177.5,18.62) .. (177.5,20) .. controls (177.5,21.38) and (176.38,22.5) .. (175,22.5) .. controls (173.62,22.5) and (172.5,21.38) .. (172.5,20) -- cycle ;
\draw  [draw opacity=0] (170,35) -- (200,35) -- (200,55) -- (170,55) -- cycle ;
\draw    (269.99,50) -- (269.99,40) ;
\draw    (259.99,30) .. controls (259.42,25.03) and (261.49,24.85) .. (249.99,25) ;
\draw    (249.99,25) -- (249.99,10) ;
\draw    (250,40) -- (249.99,25) ;
\draw    (269.99,40) .. controls (260.08,40.13) and (260.08,39.95) .. (259.99,30) ;
\draw    (269.99,40) .. controls (279.9,39.77) and (280.08,39.95) .. (279.99,30) ;
\draw  [fill={rgb, 255:red, 255; green, 255; blue, 255 }  ,fill opacity=1 ] (267.49,40) .. controls (267.49,38.62) and (268.61,37.5) .. (269.99,37.5) .. controls (271.37,37.5) and (272.49,38.62) .. (272.49,40) .. controls (272.49,41.38) and (271.37,42.5) .. (269.99,42.5) .. controls (268.61,42.5) and (267.49,41.38) .. (267.49,40) -- cycle ;
\draw  [fill={rgb, 255:red, 0; green, 0; blue, 0 }  ,fill opacity=1 ] (267.49,50) .. controls (267.49,48.62) and (268.61,47.5) .. (269.99,47.5) .. controls (271.37,47.5) and (272.49,48.62) .. (272.49,50) .. controls (272.49,51.38) and (271.37,52.5) .. (269.99,52.5) .. controls (268.61,52.5) and (267.49,51.38) .. (267.49,50) -- cycle ;
\draw  [fill={rgb, 255:red, 0; green, 0; blue, 0 }  ,fill opacity=1 ] (247.49,25) .. controls (247.49,23.62) and (248.61,22.5) .. (249.99,22.5) .. controls (251.37,22.5) and (252.49,23.62) .. (252.49,25) .. controls (252.49,26.38) and (251.37,27.5) .. (249.99,27.5) .. controls (248.61,27.5) and (247.49,26.38) .. (247.49,25) -- cycle ;
\draw    (279.99,30) .. controls (279.82,20.43) and (267.42,19.89) .. (239.99,20) ;
\draw    (239.99,70) -- (239.99,10) ;
\draw  [fill={rgb, 255:red, 0; green, 0; blue, 0 }  ,fill opacity=1 ] (237.49,20) .. controls (237.49,18.62) and (238.61,17.5) .. (239.99,17.5) .. controls (241.37,17.5) and (242.49,18.62) .. (242.49,20) .. controls (242.49,21.38) and (241.37,22.5) .. (239.99,22.5) .. controls (238.61,22.5) and (237.49,21.38) .. (237.49,20) -- cycle ;
\draw  [draw opacity=0] (210,30) -- (240,30) -- (240,50) -- (210,50) -- cycle ;
\draw  [draw opacity=0] (280,35) -- (310,35) -- (310,55) -- (280,55) -- cycle ;
\draw  [fill={rgb, 255:red, 0; green, 0; blue, 0 }  ,fill opacity=1 ] (247.5,40) .. controls (247.5,38.62) and (248.62,37.5) .. (250,37.5) .. controls (251.38,37.5) and (252.5,38.62) .. (252.5,40) .. controls (252.5,41.38) and (251.38,42.5) .. (250,42.5) .. controls (248.62,42.5) and (247.5,41.38) .. (247.5,40) -- cycle ;
\draw    (250,60) .. controls (249.83,50.43) and (250,50.12) .. (240,50) ;
\draw  [fill={rgb, 255:red, 0; green, 0; blue, 0 }  ,fill opacity=1 ] (237.5,50) .. controls (237.5,48.62) and (238.62,47.5) .. (240,47.5) .. controls (241.38,47.5) and (242.5,48.62) .. (242.5,50) .. controls (242.5,51.38) and (241.38,52.5) .. (240,52.5) .. controls (238.62,52.5) and (237.5,51.38) .. (237.5,50) -- cycle ;
\draw    (250,70) -- (250,60) ;
\draw  [draw opacity=0] (280,35) -- (310,35) -- (310,55) -- (280,55) -- cycle ;
\draw  [draw opacity=0] (280,35) -- (310,35) -- (310,55) -- (280,55) -- cycle ;

\draw (90,40) node  [font=\small]  {$=$};
\draw (350,40) node  [font=\small]  {$=$};
\draw (160,45) node  [font=\small]  {$;$};
\draw (385,45) node  [font=\small]  {$;$};
\draw (225,40) node  [font=\small]  {$=$};
\draw (295,45) node  [font=\small]  {$;$};

\end{tikzpicture}
     }
    \caption{The \protect\symmetryAxiom{}, \protect\FrobeniusAxiom{}, and \protect\idempotencyAxiom{}.}
    \label{fig:diagramThreeFrobenius}
  \end{figure}

  We must now prove that this assignment defines a functor. We will show, by
  structural induction on the first term, that $F(s) ⨾ F(t) = F(s ⨾ t).$ For the
  $\return$ case, where $s = \vret(α)$, this amounts to proving that $\{α\} ∗
  F(t) = F(α ∗ t)$. By structural induction over $t$, we can distinguish three
  cases: the third case — composition of two $\return$ statements — follows
  by definition. The first two cases, which we depict below
  (\Cref{fig:diagramReturn}), follow from the fact that any function $\{α\}$
  preserves comonoids and the induction hypothesis.
  \begin{figure}
    \centering
    \makebox[\textwidth][c]{

\tikzset{every picture/.style={line width=0.75pt}} %

\begin{tikzpicture}[x=0.75pt,y=0.75pt,yscale=-1,xscale=1]
\draw    (350,55) .. controls (349.83,45.43) and (341.5,44.85) .. (330,45) ;
\draw    (330,45) -- (330,35) ;
\draw  [fill={rgb, 255:red, 0; green, 0; blue, 0 }  ,fill opacity=1 ] (327,45) .. controls (327,43.34) and (328.34,42) .. (330,42) .. controls (331.66,42) and (333,43.34) .. (333,45) .. controls (333,46.66) and (331.66,48) .. (330,48) .. controls (328.34,48) and (327,46.66) .. (327,45) -- cycle ;
\draw   (335,55) -- (365,55) -- (365,70) -- (335,70) -- cycle ;
\draw    (330,105) -- (330,45) ;
\draw   (320,105) -- (360,105) -- (360,120) -- (320,120) -- cycle ;
\draw    (350,80) -- (350,70) ;
\draw   (335,80) -- (365,80) -- (365,95) -- (335,95) -- cycle ;
\draw    (350,105) -- (350,95) ;
\draw    (340,130) -- (340,120) ;
\draw  [draw opacity=0] (365,66) -- (385,66) -- (385,81) -- (365,81) -- cycle ;
\draw   (315,20) -- (345,20) -- (345,35) -- (315,35) -- cycle ;
\draw    (330,20) -- (330,10) ;
\draw    (435,30) .. controls (434.83,20.43) and (426.5,19.85) .. (415,20) ;
\draw    (415,20) -- (415,10) ;
\draw  [fill={rgb, 255:red, 0; green, 0; blue, 0 }  ,fill opacity=1 ] (412,20) .. controls (412,18.34) and (413.34,17) .. (415,17) .. controls (416.66,17) and (418,18.34) .. (418,20) .. controls (418,21.66) and (416.66,23) .. (415,23) .. controls (413.34,23) and (412,21.66) .. (412,20) -- cycle ;
\draw   (420,55) -- (450,55) -- (450,70) -- (420,70) -- cycle ;
\draw    (415,80) -- (415,20) ;
\draw   (405,105) -- (445,105) -- (445,120) -- (405,120) -- cycle ;
\draw    (435,80) -- (435,70) ;
\draw   (420,80) -- (450,80) -- (450,95) -- (420,95) -- cycle ;
\draw    (435,105) -- (435,95) ;
\draw    (425,130) -- (425,120) ;
\draw   (420,30) -- (450,30) -- (450,45) -- (420,45) -- cycle ;
\draw    (435,55) -- (435,45) ;
\draw   (390,80) -- (420,80) -- (420,95) -- (390,95) -- cycle ;
\draw    (415,105) -- (415,95) ;
\draw    (515,55) .. controls (514.83,45.43) and (506.5,44.85) .. (495,45) ;
\draw    (495,45) -- (495,35) ;
\draw  [fill={rgb, 255:red, 0; green, 0; blue, 0 }  ,fill opacity=1 ] (492,45) .. controls (492,43.34) and (493.34,42) .. (495,42) .. controls (496.66,42) and (498,43.34) .. (498,45) .. controls (498,46.66) and (496.66,48) .. (495,48) .. controls (493.34,48) and (492,46.66) .. (492,45) -- cycle ;
\draw   (500,55) -- (530,55) -- (530,70) -- (500,70) -- cycle ;
\draw    (495,105) -- (495,45) ;
\draw   (485,105) -- (525,105) -- (525,120) -- (485,120) -- cycle ;
\draw    (505,130) -- (505,120) ;
\draw  [draw opacity=0] (530,66) -- (550,66) -- (550,81) -- (530,81) -- cycle ;
\draw   (480,20) -- (510,20) -- (510,35) -- (480,35) -- cycle ;
\draw    (495,20) -- (495,10) ;
\draw    (600,30) .. controls (599.83,20.43) and (591.5,19.85) .. (580,20) ;
\draw    (580,20) -- (580,10) ;
\draw  [fill={rgb, 255:red, 0; green, 0; blue, 0 }  ,fill opacity=1 ] (577,20) .. controls (577,18.34) and (578.34,17) .. (580,17) .. controls (581.66,17) and (583,18.34) .. (583,20) .. controls (583,21.66) and (581.66,23) .. (580,23) .. controls (578.34,23) and (577,21.66) .. (577,20) -- cycle ;
\draw    (580,80) -- (580,20) ;
\draw   (570,105) -- (610,105) -- (610,120) -- (570,120) -- cycle ;
\draw    (590,130) -- (590,120) ;
\draw   (585,30) -- (615,30) -- (615,45) -- (585,45) -- cycle ;
\draw   (555,80) -- (585,80) -- (585,95) -- (555,95) -- cycle ;
\draw    (580,105) -- (580,95) ;
\draw    (514.99,90) -- (514.99,80) ;
\draw    (514.99,80) .. controls (505.08,80.13) and (505.08,79.95) .. (504.99,70) ;
\draw    (514.99,80) .. controls (524.9,79.77) and (525.08,79.95) .. (524.99,70) ;
\draw  [fill={rgb, 255:red, 255; green, 255; blue, 255 }  ,fill opacity=1 ] (512.49,80) .. controls (512.49,78.62) and (513.61,77.5) .. (514.99,77.5) .. controls (516.37,77.5) and (517.49,78.62) .. (517.49,80) .. controls (517.49,81.38) and (516.37,82.5) .. (514.99,82.5) .. controls (513.61,82.5) and (512.49,81.38) .. (512.49,80) -- cycle ;
\draw  [fill={rgb, 255:red, 0; green, 0; blue, 0 }  ,fill opacity=1 ] (512.49,90) .. controls (512.49,88.62) and (513.61,87.5) .. (514.99,87.5) .. controls (516.37,87.5) and (517.49,88.62) .. (517.49,90) .. controls (517.49,91.38) and (516.37,92.5) .. (514.99,92.5) .. controls (513.61,92.5) and (512.49,91.38) .. (512.49,90) -- cycle ;
\draw    (599.99,90) -- (599.99,80) ;
\draw    (599.99,80) .. controls (590.08,80.13) and (590.08,79.95) .. (589.99,70) ;
\draw    (599.99,80) .. controls (609.9,79.77) and (610.08,79.95) .. (609.99,70) ;
\draw  [fill={rgb, 255:red, 255; green, 255; blue, 255 }  ,fill opacity=1 ] (597.49,80) .. controls (597.49,78.62) and (598.61,77.5) .. (599.99,77.5) .. controls (601.37,77.5) and (602.49,78.62) .. (602.49,80) .. controls (602.49,81.38) and (601.37,82.5) .. (599.99,82.5) .. controls (598.61,82.5) and (597.49,81.38) .. (597.49,80) -- cycle ;
\draw  [fill={rgb, 255:red, 0; green, 0; blue, 0 }  ,fill opacity=1 ] (597.49,90) .. controls (597.49,88.62) and (598.61,87.5) .. (599.99,87.5) .. controls (601.37,87.5) and (602.49,88.62) .. (602.49,90) .. controls (602.49,91.38) and (601.37,92.5) .. (599.99,92.5) .. controls (598.61,92.5) and (597.49,91.38) .. (597.49,90) -- cycle ;
\draw   (585,55) -- (615,55) -- (615,70) -- (585,70) -- cycle ;
\draw    (600,55) -- (600,45) ;
\draw  [draw opacity=0] (450,65) -- (480,65) -- (480,85) -- (450,85) -- cycle ;
\draw  [draw opacity=0] (450,65) -- (480,65) -- (480,85) -- (450,85) -- cycle ;

\draw (350,62.5) node  [font=\small]  {$\{\beta \}$};
\draw (350,87.5) node  [font=\small]  {$f$};
\draw (340,112.5) node  [font=\small]  {$cont$};
\draw (375,73.5) node  [font=\small]  {$=$};
\draw (330,27.5) node  [font=\small]  {$\{\alpha \}$};
\draw (435,62.5) node  [font=\small]  {$\{\beta \}$};
\draw (435,87.5) node  [font=\small]  {$f$};
\draw (425,112.5) node  [font=\small]  {$cont$};
\draw (435,37.5) node  [font=\small]  {$\{\alpha \}$};
\draw (405,87.5) node  [font=\small]  {$\{\alpha \}$};
\draw (515,62.5) node  [font=\small]  {$\{\beta \}$};
\draw (505,112.5) node  [font=\small]  {$cont$};
\draw (540,73.5) node  [font=\small]  {$=$};
\draw (495,27.5) node  [font=\small]  {$\{\alpha \}$};
\draw (590,112.5) node  [font=\small]  {$cont$};
\draw (600,37.5) node  [font=\small]  {$\{\alpha \}$};
\draw (570,87.5) node  [font=\small]  {$\{\alpha \}$};
\draw (600,62.5) node  [font=\small]  {$\{\beta \}$};
\draw (465,75) node  [font=\small]  {$;$};

\end{tikzpicture}
}
    \caption{Functoriality for the \protect{$\return$} case.}
    \label{fig:diagramReturn}
  \end{figure}
 
  For the rest of the cases (generator and $\observe{}$), functoriality follows
  by definition and the induction hypothesis --- we do not detail these cases
  here.

  We must now prove that the assingment defines a strict monoidal functor: we
  will show it preserves whiskering. In other words, that $F(\vec{y} ⋉ t) =
  F(\vec{y}) ⋉ F(t)$. We proceed again by structural induction on the term:
  the return case is immediate because whiskering is defined to coincide with
  whiskering in $\mathbf{FinSet}^{op}$; the generator and $\observe$ cases follow
  from the string diagrammatic equations in \Cref{fig:diagramPreserveWhiskering}.
  \begin{figure}[ht!]
    \centering
    \makebox[\textwidth][c]{

\tikzset{every picture/.style={line width=0.75pt}} %

\begin{tikzpicture}[x=0.75pt,y=0.75pt,yscale=-1,xscale=1]
\draw    (50,35) .. controls (49.83,25.43) and (41.5,24.85) .. (30,25) ;
\draw    (30,25) -- (30,15) ;
\draw  [fill={rgb, 255:red, 0; green, 0; blue, 0 }  ,fill opacity=1 ] (27,25) .. controls (27,23.34) and (28.34,22) .. (30,22) .. controls (31.66,22) and (33,23.34) .. (33,25) .. controls (33,26.66) and (31.66,28) .. (30,28) .. controls (28.34,28) and (27,26.66) .. (27,25) -- cycle ;
\draw   (35,35) -- (65,35) -- (65,50) -- (35,50) -- cycle ;
\draw    (30,85) -- (30,25) ;
\draw   (20,85) -- (60,85) -- (60,100) -- (20,100) -- cycle ;
\draw    (50,60) -- (50,50) ;
\draw   (35,60) -- (65,60) -- (65,75) -- (35,75) -- cycle ;
\draw    (50,85) -- (50,75) ;
\draw    (40,110) -- (40,100) ;
\draw  [draw opacity=0] (65,46) -- (85,46) -- (85,61) -- (65,61) -- cycle ;
\draw    (15,110) -- (15,15) ;
\draw    (130,35) .. controls (129.83,25.43) and (116.5,24.85) .. (105,25) ;
\draw    (105,25) -- (105,15) ;
\draw  [fill={rgb, 255:red, 0; green, 0; blue, 0 }  ,fill opacity=1 ] (102,25) .. controls (102,23.34) and (103.34,22) .. (105,22) .. controls (106.66,22) and (108,23.34) .. (108,25) .. controls (108,26.66) and (106.66,28) .. (105,28) .. controls (103.34,28) and (102,26.66) .. (102,25) -- cycle ;
\draw   (115,35) -- (145,35) -- (145,50) -- (115,50) -- cycle ;
\draw    (105,85) -- (105,25) ;
\draw   (100,85) -- (140,85) -- (140,100) -- (100,100) -- cycle ;
\draw    (130,60) -- (130,50) ;
\draw   (115,60) -- (145,60) -- (145,75) -- (115,75) -- cycle ;
\draw    (130,85) -- (130,75) ;
\draw    (120,110) -- (120,100) ;
\draw    (90,110) -- (90,15) ;
\draw    (110,40) .. controls (109.83,30.43) and (101.5,24.85) .. (90,25) ;
\draw  [fill={rgb, 255:red, 0; green, 0; blue, 0 }  ,fill opacity=1 ] (107,45) .. controls (107,43.34) and (108.34,42) .. (110,42) .. controls (111.66,42) and (113,43.34) .. (113,45) .. controls (113,46.66) and (111.66,48) .. (110,48) .. controls (108.34,48) and (107,46.66) .. (107,45) -- cycle ;
\draw    (110,45) -- (110,40) ;
\draw    (215,35) .. controls (214.83,25.43) and (206.5,24.85) .. (195,25) ;
\draw    (195,25) -- (195,15) ;
\draw  [fill={rgb, 255:red, 0; green, 0; blue, 0 }  ,fill opacity=1 ] (192,25) .. controls (192,23.34) and (193.34,22) .. (195,22) .. controls (196.66,22) and (198,23.34) .. (198,25) .. controls (198,26.66) and (196.66,28) .. (195,28) .. controls (193.34,28) and (192,26.66) .. (192,25) -- cycle ;
\draw   (200,35) -- (230,35) -- (230,50) -- (200,50) -- cycle ;
\draw    (195,85) -- (195,25) ;
\draw   (185,85) -- (225,85) -- (225,100) -- (185,100) -- cycle ;
\draw    (205,110) -- (205,100) ;
\draw  [draw opacity=0] (230,46) -- (250,46) -- (250,61) -- (230,61) -- cycle ;
\draw    (180,110) -- (180,15) ;
\draw    (295,35) .. controls (294.83,25.43) and (281.5,24.85) .. (270,25) ;
\draw    (270,25) -- (270,15) ;
\draw  [fill={rgb, 255:red, 0; green, 0; blue, 0 }  ,fill opacity=1 ] (267,25) .. controls (267,23.34) and (268.34,22) .. (270,22) .. controls (271.66,22) and (273,23.34) .. (273,25) .. controls (273,26.66) and (271.66,28) .. (270,28) .. controls (268.34,28) and (267,26.66) .. (267,25) -- cycle ;
\draw   (280,35) -- (310,35) -- (310,50) -- (280,50) -- cycle ;
\draw    (270,85) -- (270,25) ;
\draw   (265,85) -- (305,85) -- (305,100) -- (265,100) -- cycle ;
\draw    (285,110) -- (285,100) ;
\draw    (255,110) -- (255,15) ;
\draw    (275,40) .. controls (274.83,30.43) and (266.5,24.85) .. (255,25) ;
\draw  [fill={rgb, 255:red, 0; green, 0; blue, 0 }  ,fill opacity=1 ] (272,45) .. controls (272,43.34) and (273.34,42) .. (275,42) .. controls (276.66,42) and (278,43.34) .. (278,45) .. controls (278,46.66) and (276.66,48) .. (275,48) .. controls (273.34,48) and (272,46.66) .. (272,45) -- cycle ;
\draw    (275,45) -- (275,40) ;
\draw    (215,70) -- (215,60) ;
\draw    (215,60) .. controls (205.09,60.13) and (205.09,59.95) .. (205,50) ;
\draw    (215,60) .. controls (224.91,59.77) and (225.09,59.95) .. (225,50) ;
\draw  [fill={rgb, 255:red, 255; green, 255; blue, 255 }  ,fill opacity=1 ] (212.5,60) .. controls (212.5,58.62) and (213.62,57.5) .. (215,57.5) .. controls (216.38,57.5) and (217.5,58.62) .. (217.5,60) .. controls (217.5,61.38) and (216.38,62.5) .. (215,62.5) .. controls (213.62,62.5) and (212.5,61.38) .. (212.5,60) -- cycle ;
\draw  [draw opacity=0][fill={rgb, 255:red, 0; green, 0; blue, 0 }  ,fill opacity=1 ] (212,70) .. controls (212,68.34) and (213.34,67) .. (215,67) .. controls (216.66,67) and (218,68.34) .. (218,70) .. controls (218,71.66) and (216.66,73) .. (215,73) .. controls (213.34,73) and (212,71.66) .. (212,70) -- cycle ;
\draw    (294.99,70) -- (294.99,60) ;
\draw    (294.99,60) .. controls (285.08,60.13) and (285.08,59.95) .. (284.99,50) ;
\draw    (294.99,60) .. controls (304.9,59.77) and (305.08,59.95) .. (304.99,50) ;
\draw  [fill={rgb, 255:red, 255; green, 255; blue, 255 }  ,fill opacity=1 ] (292.49,60) .. controls (292.49,58.62) and (293.61,57.5) .. (294.99,57.5) .. controls (296.37,57.5) and (297.49,58.62) .. (297.49,60) .. controls (297.49,61.38) and (296.37,62.5) .. (294.99,62.5) .. controls (293.61,62.5) and (292.49,61.38) .. (292.49,60) -- cycle ;
\draw  [draw opacity=0][fill={rgb, 255:red, 0; green, 0; blue, 0 }  ,fill opacity=1 ] (291.99,70) .. controls (291.99,68.34) and (293.33,67) .. (294.99,67) .. controls (296.64,67) and (297.99,68.34) .. (297.99,70) .. controls (297.99,71.66) and (296.64,73) .. (294.99,73) .. controls (293.33,73) and (291.99,71.66) .. (291.99,70) -- cycle ;
\draw  [draw opacity=0] (145,55) -- (175,55) -- (175,75) -- (145,75) -- cycle ;
\draw  [draw opacity=0] (145,55) -- (175,55) -- (175,75) -- (145,75) -- cycle ;
\draw  [draw opacity=0] (310,55) -- (340,55) -- (340,75) -- (310,75) -- cycle ;
\draw  [draw opacity=0] (310,55) -- (340,55) -- (340,75) -- (310,75) -- cycle ;

\draw (50,67.5) node  [font=\small]  {$f$};
\draw (40,92.5) node  [font=\small]  {$cont$};
\draw (75,53.5) node  [font=\small]  {$=$};
\draw (50,42.5) node  [font=\small]  {$\{\alpha \}$};
\draw (130,67.5) node  [font=\small]  {$f$};
\draw (120,92.5) node  [font=\small]  {$cont$};
\draw (130,42.5) node  [font=\small]  {$\{\alpha \}$};
\draw (205,92.5) node  [font=\small]  {$cont$};
\draw (240,53.5) node  [font=\small]  {$=$};
\draw (215,42.5) node  [font=\small]  {$\{\alpha \}$};
\draw (285,92.5) node  [font=\small]  {$cont$};
\draw (295,42.5) node  [font=\small]  {$\{\alpha \}$};
\draw (160,65) node  [font=\small]  {$;$};
\draw (325,65) node  [font=\small]  {$;$};

\end{tikzpicture}
     }
    \caption{Preservation of whiskering.}
    \label{fig:diagramPreserveWhiskering}
  \end{figure}

  Finally, we need to show that the functor is symmetric and that it preserves
  the partial Frobenius structure. The functor preserves the braiding,
  comultiplication, and the counit because they were defined by the $\return$
  statements corresponding to the braiding, comultiplication, and counits of
  $\mathbf{FinSet}^{op}$. The functor preserves the multiplication thanks to the
  string diagrammatic equation in \Cref{fig:diagramPreserveMultiplication}.
  \begin{figure}[ht!]
    \centering
    \makebox[\textwidth][c]{

\tikzset{every picture/.style={line width=0.75pt}} %

\begin{tikzpicture}[x=0.75pt,y=0.75pt,yscale=-1,xscale=1]
\draw    (200,60) -- (200,35) ;
\draw    (200,35) .. controls (190.09,35.13) and (190.09,34.95) .. (190,25) ;
\draw    (200,35) .. controls (209.91,34.77) and (210.09,34.95) .. (210,25) ;
\draw  [fill={rgb, 255:red, 255; green, 255; blue, 255 }  ,fill opacity=1 ] (197.5,35) .. controls (197.5,33.62) and (198.62,32.5) .. (200,32.5) .. controls (201.38,32.5) and (202.5,33.62) .. (202.5,35) .. controls (202.5,36.38) and (201.38,37.5) .. (200,37.5) .. controls (198.62,37.5) and (197.5,36.38) .. (197.5,35) -- cycle ;
\draw    (269.99,50) -- (269.99,40) ;
\draw    (259.99,30) .. controls (259.42,25.03) and (261.49,24.85) .. (249.99,25) ;
\draw    (249.99,25) -- (249.99,10) ;
\draw    (250,40) -- (249.99,25) ;
\draw    (269.99,40) .. controls (260.08,40.13) and (260.08,39.95) .. (259.99,30) ;
\draw    (269.99,40) .. controls (279.9,39.77) and (280.08,39.95) .. (279.99,30) ;
\draw  [fill={rgb, 255:red, 255; green, 255; blue, 255 }  ,fill opacity=1 ] (267.49,40) .. controls (267.49,38.62) and (268.61,37.5) .. (269.99,37.5) .. controls (271.37,37.5) and (272.49,38.62) .. (272.49,40) .. controls (272.49,41.38) and (271.37,42.5) .. (269.99,42.5) .. controls (268.61,42.5) and (267.49,41.38) .. (267.49,40) -- cycle ;
\draw  [fill={rgb, 255:red, 0; green, 0; blue, 0 }  ,fill opacity=1 ] (267.49,50) .. controls (267.49,48.62) and (268.61,47.5) .. (269.99,47.5) .. controls (271.37,47.5) and (272.49,48.62) .. (272.49,50) .. controls (272.49,51.38) and (271.37,52.5) .. (269.99,52.5) .. controls (268.61,52.5) and (267.49,51.38) .. (267.49,50) -- cycle ;
\draw  [fill={rgb, 255:red, 0; green, 0; blue, 0 }  ,fill opacity=1 ] (247.49,25) .. controls (247.49,23.62) and (248.61,22.5) .. (249.99,22.5) .. controls (251.37,22.5) and (252.49,23.62) .. (252.49,25) .. controls (252.49,26.38) and (251.37,27.5) .. (249.99,27.5) .. controls (248.61,27.5) and (247.49,26.38) .. (247.49,25) -- cycle ;
\draw    (279.99,30) .. controls (279.82,20.43) and (267.42,19.89) .. (239.99,20) ;
\draw    (240,55) -- (239.99,10) ;
\draw  [fill={rgb, 255:red, 0; green, 0; blue, 0 }  ,fill opacity=1 ] (237.49,20) .. controls (237.49,18.62) and (238.61,17.5) .. (239.99,17.5) .. controls (241.37,17.5) and (242.49,18.62) .. (242.49,20) .. controls (242.49,21.38) and (241.37,22.5) .. (239.99,22.5) .. controls (238.61,22.5) and (237.49,21.38) .. (237.49,20) -- cycle ;
\draw  [draw opacity=0] (210,20) -- (240,20) -- (240,40) -- (210,40) -- cycle ;
\draw  [fill={rgb, 255:red, 0; green, 0; blue, 0 }  ,fill opacity=1 ] (247.5,40) .. controls (247.5,38.62) and (248.62,37.5) .. (250,37.5) .. controls (251.38,37.5) and (252.5,38.62) .. (252.5,40) .. controls (252.5,41.38) and (251.38,42.5) .. (250,42.5) .. controls (248.62,42.5) and (247.5,41.38) .. (247.5,40) -- cycle ;
\draw    (190,25) -- (190,10) ;
\draw    (210,25) -- (210,10) ;

\draw (225,30) node  [font=\small]  {$=$};
\draw (295,35) node  [font=\small]  {$;$};

\end{tikzpicture}
     }
    \caption{Preservation of multiplication.}
    \label{fig:diagramPreserveMultiplication}
  \end{figure}

  This concludes the proof: we have shown that the assingment constructs a
  well-defined, strict symmetric monoidal functor that preserves the partial
  Frobenius structure.
\end{proof}

\begin{propositionrep}[see {{\cite{coya16:corelations,grandis01:finite}}}]
  \label{prop:freeCopyDiscardCategory}%
  The free \copyDiscardCategory{} over a single generator is the opposite
  category of finite sets and functions endowed with the coproduct monoidal
  tensor, $(\mathbf{FinSet}^{op},+)$.
\end{propositionrep}
\begin{proof}[Proof sketch]
  Let us first prove that $(\mathbf{FinSet}^{op},+)$ is a
  \copyDiscardCategory{}. The copy morphism $δ ፡ x + x → x$ is defined by $δ(i)
  = i$ when $i ≤ x$, and $δ(i) = i - x$ when $i > x$. The discard morphism $ε ፡
  0 → x$ is defined vacuously. It is straightforward to show that these morphisms
  satisfy the axioms of \copyDiscardCategories{}.

  Every function can be decomposed, inductively over the size of its codomain,
  in primitives of \copyDiscardCategories{}. Given $f_x ፡ x → y$, we can
  consider two cases: if $x = 0$, then $f_0 = ε$ coincides with the counit; if
  $x = 1 + x'$, then $f(1) = y_i$ and we can express the function as 
  $$f_{1+x'} = (\id[1] ⋊ f_{x'}) ⨾ (σ_{1,i} ⋊ \id[y]) ⨾ (\id[i-1] ⋉ δ_i ⋊
  \id[y-i]).$$ 
  In this way, the interpretation of any function on a
  \copyDiscardCategory{} is determined. It remains to prove that this
  interpretation is functorial. We omit this part of the proof, which is
  mostly straightforward.
\end{proof}

\begin{theoremrep}[Observe-arrow notation is an internal language]
  \label{th:doNotationInternalLanguage}%
  \label{th:arrowNotationInternalLanguage}%
  \label{ax:th:doNotationInternalLanguage}%
  \label{ax:th:arrowNotationInternalLanguage}%
  The category of "observe-arrow notation terms" over a \signature{},
  $\obsArrow(Σ)$, is the free \strictCopyDiscardCompareCategory{} over that
  \signature{}. In other words, we have an adjunction, $\obsArrow ⊣ \forgetF$,
  where constructing \observeDoNotationTerms{} provides a left adjoint, $\obsArrow
  ፡ \Sig → \cdcCat$, to the forgetful functor, $\forgetF ፡ \cdcCat → \Sig$.
\end{theoremrep}
\begin{proof}
  We will construct the adjunction by exhibiting the universal arrow, 
  $$u_{Σ} ፡ Σ → \forgetF(\obsArrow(Σ)),$$
  defined as in \Cref{lemma:inclusion}. Given any morphism, we have shown that
  there exists at most a unique way of factoring through the universal arrow
  (\Cref{lemma:existsFactoring}); we have then shown that the assignment defines
  a \copyDiscardCompareFunctor{} (\Cref{lemma:factoring}). This exhibits
  $\obsArrow(Σ)$ as the free \copyDiscardCompareCategory{} over a \signature{}
  $Σ$.
\end{proof}

Freeness, in particular, gives soundness and completeness of the
\arrowNotation{} for \copyDiscardCompareCategories{}. A morphism of signatures
\(\Sigma \to \forgetF(𝔸)\) gives an \emph{interpretation} of the signature
\(\Sigma\). Thanks to the adjunction, every such morphism determines a unique
\copyDiscardCompareFunctor{} \(\obsArrow(\Sigma) \to 𝔸\) specifying a
\emph{model} for the \signature{} \(\Sigma\) (\Cref{sec:functorialSemantics}
instantiates this correspondence in the case of \subdistributions{}). Therefore,
every equation that is true in \(\obsArrow(\Sigma)\) is also true in all models
(soundness). Conversely, if an equation is true in all models, in particular, it
is true in \(\obsArrow(\Sigma)\) as it is also a model, in fact, the free one
(completeness).

\begin{remark}
  Networks (or \emph{network diagrams}) are combinatorial objects
  defined by a set of nodes and a set of directed wires linking
  them. Networks are informally employed as graphical models in
  probabilistic inference and learning.  Multiple-output networks form
  the free \copyDiscardCategory{} over a signature
  \cite{fritzLiang23:freeMarkov}. It follows that they are in
  bijective correspondence with \arrowNotation{} expresions without
  $\observe$ statements. Explicitly, variables represent wires and
  statements represent nodes --- a topological ordering is induced by
  the order of the statements, but the term is invariant to the choice
  of a topological ordering thanks to the \interchangeAxiom{}.
\end{remark}

\subsection{Functorial Semantics}%
\label{sec:functorialSemantics}%

"Observe-arrow notation" is interesting not only as a syntax for
\copyDiscardCompareCategories{} in general, but also for its semantics
in \subdistributions{}.  This section formalises the reading of
\arrowNotation{} we introduced in \Cref{sec:reading-arrow}, where
the interpretations appeared on the side.

\begin{definition}[Signature interpretation]%
  \defining{linkSignatureInterpretation}%
  A \emph{signature interpretation} for the \signature{} $Σ$ is a pair of
  functions: one assigning a set to each type, $⟦ • ⟧ ፡ Σ_{type} → \Set$; and
  one assigning a substochastic channel to each generator,
    $$⟦ • ⟧(-) ፡ Σ(X₁,…,Xₙ; Y₁,…,Yₘ) × ⟦X_1⟧ × … × ⟦ Xₙ ⟧ → \Subd (⟦ Y₁ ⟧ × … ⟦
      Yₘ ⟧).$$
\end{definition}

What follows is a recipe to interpret an \arrowNotation{} term given
a \signature{} interpretation. This is not a novel recipe: it can be
recognized as a simplified form of Moggi's monadic semantics
\cite{moggi1991notions} for the case of the \subdistribution{} monad
\cite{cho2017kleisli}. We restate it here for completeness.

\begin{definition}[Kleisli extension]
  Any function to a set of \subdistributions{} $f ፡ X → \Subd(Y)$ has
  a Kleisli extension to a function between \subdistributions{},
  \smasheq{f_{\ast} ፡ \Subd(X) → \Subd(Y)}.
  The extension is defined by being linear and acting on any monomial as $f$
  did:
$$
f_{\ast}\left( \sum\nolimits_i rᵢ \ket{x_i} \right) =
\sum\nolimits_i rᵢ \cdot f(x_{i}).
$$

\noindent The latter multiplication $rᵢ \cdot (-)$ applied to a
subdistribution means that $rᵢ \cdot (-)$ is applied to all the
monomials of the distribution.
\end{definition}

\begin{definition}[Extension of an interpretation]
Every \signatureInterpretation{}, $⟦ • ⟧$ extends to a function assigning to
each term a function of the form: %
$$⦃ • ⦄(-) ፡ (Γ ⊢ Δ) × ⟦Γ⟧ → \Subd\big(⟦Δ⟧\big),$$

\noindent where the context interpretation $⟦Γ⟧$ is $⟦x₁:A_1, …,
  xₙ:A_n⟧ = ⟦A_1⟧ \times \cdots \times ⟦A_n⟧$.  This extension is
inductively defined as follows, for a vector $\vec{v}\in ⟦Γ⟧$.
\begin{enumerate}
  \item A $\return$ statement yields the value of some variables,  $$⦃ \return
  (x_{α(1)}, …, x_{α(m)}) ⦄(\vec{v}) = 1\bigket{v_{α(1)}, …, v_{α(m)}}.$$

  \item A generator statement computes a \subdistribution{} over some
    variables, which is handled by the Kleisli extension of the
    interpretation of the subsequent statement:
  $$\textstyle ⦃ y ← f(x_{α(1)}, …, x_{α(m)}) ⨾ \cont ⦄(\vec{v}) =
  ⦃ \cont ⦄_{\ast} \Big(1\bigket{\vec{v}} \tensor ⦃f⦄(v_{α(1)}, …, v_{α(m)})\Big).$$

  \item An $\observe$ statement multiplies by zero all terms not satisfying a
  certain condition, described here abstractly as a subset $U$,
  $$⦃ \observe(U) ⨾ \cont ⦄(\vec{v}) =  
  \begin{cases}
  ⦃ \cont ⦄(v₁ … vₙ) & \mbox{if }\vec{v}\in U, \\
  \textbf{0} & \mbox{otherwise,}
  \end{cases}$$
  \noindent where $\textbf{0}$ is the zero distribution.
\end{enumerate}
\end{definition}

\begin{remark}[Soundness and completeness for subdistributions]
  Our language is sound and complete for \copyDiscardCompareCategories{}, but
  the reader may ask if it can be made sound and complete for
  \subdistributions{}. This could be achieved by additionally including the
  theory of convex sums, which is well-known (e.g., Fritz illustrates a
  diagrammatic version \cite{fritz09:presentation}). However, that would also
  dramatically reduce its class of models (e.g., excluding continuous
  probability).
\end{remark}
\subsection{Normalisation}\label{sec:normalisation}

\Cref{sec:reading-arrow,sec:functorialSemantics} show the process of evaluating
\arrowNotation{} statements: at each line, the ``state'' of the current
computation is a \subdistribution{}. It corresponds to a failure, or a
distribution together with a validity, obtained via rescaling, see
Definition~\ref{def:updating}~\eqref{def:updating:rescaling}. In many cases,
however, we only care about the normalised version of such \subdistribution{}.
In these cases, we will show that it is also possible to work with normalised
distributions and disregard the corresponding validities without altering the
result. Thus, in the present setting, \normalisation{} is rescaling while
forgetting the validity. The \normalisation{} of a term is obtained not as a new
primitive to the language, but as a derived operation
(\Cref{def:normalisation}).

\Normalisation{} is a famously non-compositional operation: given two
composable Kleisli maps, also called channels, \(f ፡ X
→ \Subd(Y)\) and \(g ፡ Y → \Subd(Z)\), the
\normalisation{} of their Kleisli composition, \smasheq{\norm{f ⨾ g}},
is different from the composition of their \normalisations{},
\smasheq{\norm{f} ⨾ \norm{g}}. This failure of compositionality suggests that
it might not be possible to naively discard the validity of
\subdistributions{} when computing.

This section proves that discarding the validity at each line in
the computation is in fact possible. It is sufficient to keep the normalised
\subdistribution{} corresponding to the ``state'' of the computation at each
line: the \normalisation{} of a composition, \(\norm{f ⨾ g}\), can be recovered
from the \normalisation{} of its first factor and its second factor, as
\(\norm{\norm{f} ⨾ g}\) (\Cref{prop:normalise-whenever}) --- the validity of the
first factor is not needed.

The \copyDiscardCategory{} structure allows for an abstract definition of
\kl{normalisation}~\cite{dilavore23:evidential}. Intuitively, a
\kl{normalisation} of \(f\) is a morphism that behaves like \(f\) while being
total on its domain.

\begin{definition}
  \label{def:normalisation}%
  \defining{linkNormalisation}{}%
  \AP In a \copyDiscardCategory{}, a morphism $n ፡ X → A$ is a
  \intro{normalisation} of another morphism $f ፡ X → A$ when
  \[\left.
      \begin{array}{l}
        a ← f(x) \\
        \return (a)
      \end{array}\right| =
    \left.
      \begin{array}{l}
        a' ← f(x) \\
        a ← n(x) \\
        \return (a)
      \end{array}\right| \quad\text{and}\quad
    \left.
      \begin{array}{l}
        a ← n(x) \\
        \return (a)
      \end{array}\right| =
    \left.
      \begin{array}{l}
        a' ← n(x) \\
        a ← n(x) \\
        \return (a)
      \end{array}\right|\ .\]
\end{definition}

\kl{Subdistributions} give a semantic universe for \kl{observe-arrow notation}. As
implicitly shown in \Cref{sec:illustrations}, \kl{subdistributions} do support
\kl{normalisation}.

\begin{example}[Normalisations in subdistributions]
  Let $f ፡ X → A$ be a partial stochastic channel, i.e.\ a function \(X →
  \Subd(A)\). Consider the probability that \(f\) does not fail on an input
  \(x\), \(v(x) = \sum_{a ∈ A} f(a \given x)\). A \normalisation{} of \(f\) is
  a partial stochastic channel $n(f) ፡ X → A$ defined by
  \[
  \norm{f}(a \given x) = \frac{f(a \given x)}{v(x)},\mbox{ when }v(x) ≠ 0,
  \]
  and $\norm{f}(a \given x) = 0$ otherwise. When $v(x) = 0$,
  \normalisation{} is not unique.
\end{example}

The following result is known for the case of
\kl{subdistributions}~\cite{heunenKammar16:semanticsProbabilistic}.
The structure of \copyDiscardCategories{} %
turns out to be sufficient for this result to hold: whenever \kl{normalisations}
exist, validities may be discarded at each step of the execution. We prove this
fact using \arrowNotation{}.

\begin{proposition}
  \label{prop:normalise-whenever}%
  Let $f ፡ X → Y$ and $g ፡ Y → Z$ be morphisms in a \copyDiscardCategory{}. %
  Let \(\norm{f} ፡ X → Y\) be a \kl{normalisation} of \(f\).
  Then, any \kl{normalisation} of \(\norm{f} ⨾ g\) is a \kl{normalisation} of \(f ⨾ g\).
\end{proposition}
\begin{proof}
  Let \(n\) be a \kl{normalisation} of \(\norm{f} ⨾ g\). We prove it is a
  \kl{normalisation} of the composition \(f ⨾ g\).
  \begin{align*}
    \left.
      \begin{array}{l}
        y ← f(x) \\
        z' ← g(y) \\
        z ← n(x) \\
        \return (z)
      \end{array}\right| \overset{(i)}{=}
    \left.
      \begin{array}{l}
        y' ← f(x) \\
        y  ← \norm{f}(x) \\
        z' ← g(y) \\
        z ← n(x) \\
        \return (z)
      \end{array}\right| \overset{(ii)}{=}
    \left.
      \begin{array}{l}
        y' ← f(x) \\
        y  ← \norm{f}(x) \\
        z ← g(y) \\
        \return (z)
      \end{array}\right| \overset{(iii)}{=}
    \left.
      \begin{array}{l}
        y ← f(x) \\
        z ← g(y) \\
        \return (z)
      \end{array}\right|
  \end{align*}
  Here, \((i)\) applies the definition of \kl{normalisation} for \(f\), \((ii)\)
  applies the definition of \kl{normalisation} for  \(\norm{f} ⨾ g\); and
  $(iii)$ applies the definition of \kl{normalisation} for $f$. Finally, $n$
  satisfies equation on the right in \Cref{def:normalisation} because it is a
  \kl{normalisation} of $\norm{f} ⨾ g$.
\end{proof}

\Cref{fig:solveMontyHall:donotation:normalise} shows a new formalisation of the
\MontyHallProblem{}, where on-the-side \kl{subdistributions} are \kl{normalised} at each line.
Compare the final result with that obtained in
\Cref{fig:solveMontyHall:donotation}: the final \kl{normalised} outcomes coincide.

\begin{figure}[ht!]
\centering
\begin{tabular}{cll}
(1) & \smasheq{\car \gets \uniform\{L, M, R\}} \quad\quad
& \smasheq{\frac{1}{3}\ket{L} + \frac{1}{3}\ket{M} + \frac{1}{3}\ket{R}}
\\[+0.2em]
(2) & \smasheq{\player \gets \uniform\{L, M, R\}} \quad\quad
& \smasheq{\frac{1}{9}\ket{L,L} + \frac{1}{9}\ket{L,M} + \frac{1}{9}\ket{L,R} + \frac{1}{9}\ket{M,L}}
\\
    & & \smasheq{\quad + \frac{1}{9}\ket{M,M} + \frac{1}{9}\ket{M,R} + \frac{1}{9}\ket{R,L}}
    \\
    & & \smasheq{\quad + \frac{1}{9}\ket{R,M} + \frac{1}{9}\ket{R,R}}
    \\[+0.2em]
(3) & \smasheq{\host ← \caseOf{(\car, \player)}} \\
    & \smasheq{\quad (x,x) \mapsto \frac{1}{2}\ket{y} + \frac{1}{2}\ket{z}} & \smasheq{\frac{1}{18}\ket{L,L,M} + \frac{1}{18}\ket{L,L,R} + \frac{1}{18}\ket{M,M,L}}\\
    & \smasheq{\quad (x,y) \mapsto 1\ket{z},
       \quad\mbox{for }x\neq y \neq z \neq x} & \smasheq{\quad + \frac{1}{18}\ket{M,M,R} + \frac{1}{18}\ket{R,R,L} + \frac{1}{18}\ket{R,R,M}}\\[+0.2em]
  & & \smasheq{\quad + \frac{1}{9}\ket{L,M,R} + \frac{1}{9}\ket{L,R,M} + \frac{1}{9}\ket{M,R,L}}\\
  & & \smasheq{\quad + \frac{1}{9}\ket{M,L,R} + \frac{1}{9}\ket{R,L,M} + \frac{1}{9}\ket{R,M,L}}\\
  (4) &\smasheq{\observe(\player = M)}
& \smasheq{\frac{1}{6} \ket{M,M,L} + \frac{1}{6} \ket{M,M,R} + \frac{1}{3} \ket{L,M,R}}\\
  & & \smasheq{\quad + \frac{1}{3} \ket{R,M,L}}
    \\[+0.2em]
(5) &\smasheq{\observe(\host = L)} &
    \smasheq{\frac{1}{3} \ket{M,M,L} + \frac{2}{3} \ket{R,M,L}}
    \\[+0.2em]
    (6) &\smasheq{\return(\car)} &
    \smasheq{\frac{1}{3}\ket{M} + \frac{2}{3}\ket{R}}
  \end{tabular}
  \caption{Calculations, normalising at each line, for the \protect\MontyHallProblem{}.}
  \label{fig:solveMontyHall:donotation:normalise}
\end{figure}

\section{Conclusions}

We have introduced \kl{observe-arrow notation} for \copyDiscardCompareCategories{}, a
simple formal language --- based on Haskell's arrow notation --- that can be
used to formulate and solve problems in decision theory and basic statistics.
Terms in this language have a formal semantics as Kleisli maps of the
\subdistribution{} monad, that is, as partial stochastic channels. We have
illustrated how to compute this semantics step-by-step.  Given the amount of
literature devoted to the discussion of problems in decision theory --- and the
lack of a current agreement on how to solve some basic problems --- we believe
that this formal language may be helpful to reach consensus.

We have described a variant of \arrowNotation{} that is sound and complete for
\copyDiscardCompareCategories{} (\Cref{th:arrowNotationInternalLanguage}). In
particular, it is sound for partial stochastic channels. When a
\copyDiscardCompareCategory{} has \normalisations{}, we have shown that the
operation of normalisation associates to the right
(\Cref{prop:normalise-whenever}): we may work with normalised channels without
affecting the final result.

\subsection{Further work}

A potential variant of our construction in \Cref{th:doNotationInternalLanguage}
uses single-output signatures: \arrowNotationTerms{} over a single output signature
coincide with Bayesian networks \cite{fritzLiang23:freeMarkov,jacobs2021causal},
and so Bayesian networks extended with comparator nodes may be seen as a
language for \copyDiscardCompareCategories{}.

The continuous case is not developed in this article. It has been previously
shown that a continuous language with exact observations can be given semantics
in terms of Markov categories via a standard construction that uses \kl{partial
Markov categories} \cite{stein2021structural,dilavore23:evidential}. We could
employ \kl{observe-arrow notation} to discuss continuous probability with exact
observations, even if this text is restricted to the simpler semantics of
\subdistributions{}.

\section*{Acknowledgements}
We thank Paweł Sobociński, Márk Széles, and Paolo Perrone for discussion. This
article is based upon work from COST Action EuroProofNet, CA20111 supported by
COST (European Cooperation in Science and Technology). Mario Román was supported
by the Air Force Office of Scientific Research under award number
FA9550-21-1-0038. Elena Di Lavore and Mario Román were supported by the Advanced
Research + Invention Agency (ARIA) Safeguarded AI Programme.

\newpage

\bibliographystyle{alpha}
\bibliography{main}

\clearpage
\appendix
 \end{document}